\documentclass[12pt]{article} %***
\usepackage[authoryear,sectionbib]{natbib}
\usepackage{array,epsfig,rotating}
\usepackage[]{hyperref}  %<----modified by Ivan
\usepackage{amsmath,amsthm,float,url}
\usepackage{times}
\usepackage{bm}
\usepackage[plain,noend]{algorithm2e}
\DeclareMathAlphabet{\mathpzc}{OT1}{pzc}{m}{it}
\usepackage{amssymb,amsfonts}
\usepackage{array,xr}
\usepackage{booktabs,multirow}
\usepackage{graphicx,epsfig,epstopdf}
\usepackage[usenames]{color}

\usepackage{colortbl}

\makeatletter
\def\references{\bibliography{fdarep.bib}}
\usepackage{sectsty, secdot}
\sectionfont{\fontsize{12}{14pt plus.8pt minus .6pt}\selectfont}
\renewcommand{\theequation}{\thesection\arabic{equation}}
\subsectionfont{\fontsize{12}{14pt plus.8pt minus .6pt}\selectfont}
\usepackage[margin=1in]{geometry}

\usepackage{amsmath}
\usepackage{amssymb}
\usepackage{amsfonts}
\usepackage{multirow}
\usepackage{amsthm}

\newtheorem{thm}{Theorem}
\newtheorem{lem}{Lemma}

\newtheorem{prop}{Proposition}
\theoremstyle{definition}

\usepackage{adjustbox}
\usepackage{amstext}
\usepackage{mathrsfs}
\usepackage{mathtools}
\usepackage{enumitem}
\usepackage{color}
\usepackage{float}
\usepackage{subcaption}
\usepackage{booktabs}
\usepackage{multirow}
\usepackage[table]{xcolor}

\usepackage{graphics}
\usepackage{rotating}
\graphicspath{{figures/}}

\newcommand{\bea}{\begin{eqnarray*}}
\newcommand{\eea}{\end{eqnarray*}}
\newcommand{\ed}{\end{document}}

\newcommand{\btab}{\begin{tabular}}
\newcommand{\etab}{\end{tabular}}

\newcommand{\bfi}{\begin{figure}}
\newcommand{\efi}{\end{figure}}
\newcommand{\ben}{\begin{enumerate}}
\newcommand{\een}{\end{enumerate}}
\newcommand{\bay}{\begin{array}}
\newcommand{\eay}{\end{array}}

\def\bco{\iffalse}

\newcommand{\cp}{\citep}

\newcommand{\bc}{\begin{center}}
\newcommand{\ec}{\end{center}}

\DeclareMathOperator*{\argmin}{argmin}

\newcommand{\diag}{\,\text{diag}}

\newcommand{\tr}{\,\text{tr}}
\newcommand{\E}{\,\text{E}}

\newcommand{\Var}{\,\text{Var}}
\newcommand{\var}{\,\text{var}}
\newcommand{\Cov}{\,\text{Cov}}
\newcommand{\cov}{\,\text{cov}}

\renewcommand{\exp}{\,\text{exp}}
\newcommand*{\bbR}{\mathbb{R}}
\newcommand*{\cA}{\mathcal{A}}
\newcommand*{\cT}{\mathcal{T}}
\newcommand*{\cW}{\mathcal{W}}
\newcommand*{\cG}{\mathcal{G}}
\newcommand*{\cR}{\mathcal{R}}

\newcommand*{\cD}{\mathcal{D}}
\newcommand*{\hcD}{\hat{\cD}}
\newcommand*{\cP}{\mathcal{P}}
\newcommand*{\hcP}{\hat{\cP}}
\newcommand*{\tR}{\tilde{R}}
\newcommand*{\hC}{\hat{C}}
\newcommand*{\tcP}{\tilde{\cP}}
\newcommand*{\tZ}{\tilde{Z}}

\newcommand{\sumin}{\sum_{i=1}^n}
\newcommand{\sumjn}{\sum_{j=1}^n}
\newcommand{\sumM}{\sum_{m=1}^M}
\newcommand{\sumlm}{\sum_{l=1}^m}

\newcommand{\hf}{\hat{f}}
\newcommand{\hI}{\hat{I}}
\newcommand{\hr}{\hat{r}}
\newcommand{\hu}{\hat{u}}

\newcommand{\hbeta}{\hat{\beta}}

\newcommand{\heta}{\hat{\eta}}
\newcommand{\hsigma}{\hat{\sigma}}

\newcommand{\hmu}{\hat{\mu}}

\newcommand{\hF}{\hat{F}}

\newcommand{\hxi}{\hat{\xi}}
\newcommand{\hXi}{\hat{\Xi}}

\newcommand{\hlambda}{\hat{\lambda}}
\newcommand{\hphi}{\hat{\phi}}
\newcommand{\hbphi}{\hat{\bphi}}

\newcommand{\hGamma}{\hat{\Gamma}}

\newcommand{\tf}{\tilde{f}}
\newcommand{\tsigma}{\tilde{\sigma}}
\newcommand{\tu}{\tilde{u}}
\newcommand{\tF}{\tilde{F}}

\newcommand{\tmu}{\tilde{\mu}}

\newcommand{\teta}{\tilde{\eta}}

\newcommand{\tX}{\tilde{X}}

\newcommand{\txi}{\tilde{\xi}}

\newcommand{\bX}{\mathbf{X}}
\newcommand{\bY}{\mathbf{Y}}

\newcommand{\bU}{\mathbf{U}}
\newcommand{\bT}{\mathbf{T}}

\newcommand{\bA}{\mathbf{A}}

\newcommand{\bW}{\mathbf{W}}

\newcommand{\bmu}{\boldsymbol{\mu}}
\newcommand{\hbmu}{\hat{\boldsymbol{\mu}}}
\newcommand{\bxi}{\boldsymbol{\xi}}

\newcommand{\bphi}{\boldsymbol{\phi}}
\newcommand{\bPhi}{\boldsymbol{\Phi}}
\newcommand{\beps}{\boldsymbol{\epsilon}}
\newcommand{\bepsi}{\beps_i}
\newcommand{\bSigma}{\boldsymbol{\Sigma}}
\newcommand{\bmui}{\bmu_i}
\newcommand{\hbmui}{\hat{\bmu}_i}
\newcommand{\bYi}{\bY_i}

\newcommand{\bXi}{\bX_i}
\newcommand{\bTi}{\bT_i}
\newcommand{\be}{\mathbf{e}}
\newcommand{\bejs}{\mathbf{e}_j^*}
\newcommand{\bejsT}{\mathbf{e}_j^{*T}}
\newcommand{\hbe}{\hat{\mathbf{e}}}
\newcommand{\hbejs}{\hat{\mathbf{e}}_j^*}
\newcommand{\hbejsT}{\hat{\mathbf{e}}_j^{*T}}

\newcommand{\bLambda}{\boldsymbol{\Lambda}}

\newcommand{\hbLambda}{\hat{\bLambda}}

\newcommand{\hbSigma}{\hat{\bSigma}}

\newcommand*{\bbeta}{\boldsymbol{\beta}}
\newcommand*{\hbbeta}{\hat{\boldsymbol{\beta}}}
\newcommand*{\tbxi}{\tilde{\boldsymbol{\xi}}}
\newcommand{\tbxiiK}{\tilde{\bxi}_{iK}}
\newcommand{\hbxiiK}{\hat{\bxi}_{iK}}
\newcommand{\hbxiKs}{\hat{\bxi}_{K}^*}
\newcommand{\hbxiKsT}{\hat{\bxi}_{K}^{*T}}
\newcommand{\tbxiKs}{\tilde{\bxi}_{K}^*}
\newcommand{\tbxiKsT}{\tilde{\bxi}_{K}^{*T}}

\newcommand*{\hbPhi}{\hat{\boldsymbol{\Phi}}}

\newcommand{\ntoinf}{n \rightarrow \infty}
\newcommand{\toinf}{\rightarrow \infty}
\newcommand{\tozero}{\rightarrow 0}

\def\emptyset{\varnothing}

\newcommand{\txiik}{\txi_{ik}}
\newcommand{\xiik}{\xi_{ik}}

\DeclarePairedDelimiterX{\inner}[2]{\langle}{\rangle}{#1, #2}

\newcommand{\bphiik}{\bphi_{ik}}
\newcommand{\bPhiiK}{\bPhi_{iK}}
\newcommand{\hbPhiiK}{\hbPhi_{iK}}

\newcommand{\tinT}{t\in \cT}
\newcommand{\supt}{\sup_{t\in\cT}}

\newcommand{\almostsure}{\text{a.s.}}
\newcommand{\quadas}{\quad \almostsure}

\newcommand{\norm}[1]{\left\lVert#1\right\rVert}
\newcommand{\normtwo}[1]{\left\lVert#1\right\rVert_2}
\newcommand{\normop}[1]{\left\lVert#1\right\rVert_{op}}

\newcommand{\normInf}[1]{\left\lVert#1\right\rVert_\infty}

\newcommand*{\Xic}{X_i^c}

\newcommand*{\inv}{^{-1}}

\newcommand{\sumjK}{\sum_{j=1}^K}
\newcommand{\sumkK}{\sum_{k=1}^K}

\newcommand{\sumkinf}{\sum_{k=1}^\infty}
\newcommand{\sumkKinf}{\sum_{k=K+1}^\infty}

\newcommand{\Xii}{\Xi_{i}}
\newcommand{\hXii}{\hXi_{i}}
\newcommand{\XiiK}{\Xi_{iK}}
\newcommand{\hXiiK}{\hXi_{iK}}

\newcommand{\cGi}{\cG_i}

\newcommand{\hcGi}{\hat{\cG}_i}

\newcommand{\cGKi}{\cG_{iK}}

\newcommand{\hcGKi}{\hat{\cG}_{iK}}
\newcommand{\hcGKs}{\hat{\cG}_{K}^*}

\newcommand{\hGammaiK}{\hat{\Gamma}_{iK}}
\newcommand{\bSigmaiK}{\bSigma_{iK}}
\newcommand{\hbSigmaiK}{\hat{\bSigma}_{iK}}

\newcommand{\hDelta}{\hat{\Delta}}

\newcommand{\single}{\renewcommand{\baselinestretch}{1.2}\normalsize}
\newcommand{\double}{\renewcommand{\baselinestretch}{2}\normalsize}

\theoremstyle{plain}

\theoremstyle{remark}

\begin{document}

%%%%%%%%%%%%%%%%%%%%%%%%%%%%%%%%%%%%%%%%%%%%%%%%%%%%%%%%%%%%%%%%%%%%%%%%%%%%%%%%%%%%%%%%%%%%%%%%%%%%%%%%%%%%%%%%%%%%%%%%%%%%
%%%%%%%%%%%%%%%%%%%%%%%%%%%%%%%%%%%%%%%%%%%%%%%%%%%%%%%%%%%%%%%%%%%%%%%%%%%%%%%%%%%%%%%%%%%%%%%%%%%%%%%%%%%%%%%%%%%%%%%%%%%%

\renewcommand{\baselinestretch}{2}

%\markright{ \hbox{\footnotesize\rm Statistica Sinica
%{\footnotesize\bf 24} (201?), 000-000
%}\hfill\\[-13pt]
%\hbox{\footnotesize\rm
%\href{http://dx.doi.org/10.5705/ss.20??.???}{doi:http://dx.doi.org/10.5705/ss.20??.???}
%}\hfill }

%{\hfill{\footnotesize\rm Alvaro Gajardo AND Xiongtao Dai AND Hans-Georg M\"uller} \hfill}
%{\hfill {\footnotesize\rm Sparse to Dense Functional Data} \hfill}

\renewcommand{\thefootnote}{}
$\ $\par

%%%%%%%%%%%%%%%%%%%%%%%%%%%%%%%%%%%%%%%%%%%%%%%%%%%%%%%%%%%%%%%%%%%%%%%%%%%%%%%%%%%%%%%%%%%%%%%%%%%%%%%%%%%%%%%%%%%%%%%%%%%%

\fontsize{12}{14pt plus.8pt minus .6pt}\selectfont \vspace{0.8pc}
\centerline{\large\bf Predictive  Distributions and the Transition from Sparse to Dense Functional Data\footnote{This research was done while Alvaro Gajardo was a PhD student at the University of California, Davis.}}
\vspace{.4cm} 
\centerline{Alvaro Gajardo$^1$, Xiongtao Dai$^2$, and Hans-Georg M\"uller$^1$} 
\vspace{.4cm} 
\centerline{\it $^1$Department of Statistics, University of California, Davis, USA}
\centerline{\it $^2$Division of Biostatistics, University of California, Berkeley, USA}
 \vspace{.55cm} \fontsize{9}{11.5pt plus.8pt minus.6pt}\selectfont

%%%%%%%%%%%%%%%%%%%%%%%%%%%%%%%%%%%%%%%%%%%%%%%%%%%%%%%%%%%%%%%%%%%%%%%%%%%%%%%%%%%%%%%%%%%%%%%%%%%%%%%%%%%%%%%%%%%%%%%%%%%%

\begin{quotation}
\noindent {\it Abstract:}
Gaussian distributed sparsely sampled longitudinal data can be represented as Gaussian distributions of   their functional principal component scores, conditional on the available data.
Since these conditional distributions reflect the entire information available about these scores and therefore about the unknown trajectories that constitute the realizations of the stochastic process that generates the functional data,  they are referred to  as   predictive distributions. This  motivates a deeper investigation of    the convergence of the predicted functional principal component scores given noisy longitudinal observations towards the true but unobservable scores as the designs transition from sparse (longitudinal) to dense (functional) and of the  shrinkage of the predictive distributions towards a point mass located at the true score as the number of observations per subject increases.  %the shrinkage of functional $K$-truncated predictive distributions. % when the truncation point $K=K(n)$ diverges with sample size $n$. 
Our study is motivated by the theoretical and practically relevant  challenge that point predictions in the sparse sampling regime are not consistent for the true functional principal component scores. Our proposal is to change the  perspective towards a  focus  on predictive distributions, which  can be consistently estimated.  The emphasis is thus shifted to uncertainty quantification.  This   approach is also demonstrated  for the case of  sparsely sampled longitudinal predictors in  functional linear models where again one does not have consistent point predictors. Theoretical justification is provided through the asymptotic rates of convergence for the $2$-Wasserstein metric between true and estimated predictive distributions. The application of the predictive distribution approach for functional principal component analysis  is illustrated for longitudinal data from the Baltimore Longitudinal Study of Aging.

\vspace{9pt}
\par
\end{quotation}\par

\def\thefigure{\arabic{figure}}
\def\thetable{\arabic{table}}

\renewcommand{\theequation}{\thesection.\arabic{equation}}

\fontsize{12}{14pt plus.8pt minus .6pt}\selectfont

\section{Introduction}\label{sec: intro}

\subsection{General perspective and background} Functional Data Analysis  has found a wide range of applications 
\citep{rams:05,horv:12,mull:16}. These include longitudinal studies, where functional principal component analysis \citep{klef:73, castr:86}, a core technique of Functional Data Analysis, was  shown to play  a central role, due to its interpretability and ease of implementation.  % A major reason for this is that FPCA can be  naturally adjusted to account for the commonly observed 
A key feature of many longitudinal  studies is the sparsity of the available observations per subject, which are inherently correlated and are often available at only a few irregular times and usually contaminated with measurement error.

When subjects are recorded densely over time, one can consistently recover underlying  random trajectories from the Karhunen--Lo\`eve representation. Starting from 
 the auto-covariance function of the process $X$ given by 
\begin{align} \label{auto}
	\Gamma(s, t)& = \cov(X(s), X(t)    ) = \sum_{k=1}^\infty \lambda_k \phi_k(s) \phi_k(t),\quad  s,t\in\cT ,
\end{align}
where $\lambda_1> \lambda_2 > \dots \ge 0$ are the ordered eigenvalues, satisfying   $\sum_{k=1}^\infty \lambda_k < \infty$, and $\phi_k$, $k\ge 1$, are the orthonormal eigenfunctions associated with the Hilbert--Schmidt operator $\Xi(g)=\int_\cT \Gamma(\cdot,t) g(t) dt$, $g\in L^2(\cT)$. Define eigengaps $\delta_k=\min(\lambda_{k-1}-\lambda_{k},\lambda_{k}-\lambda_{k+1})$, $k=1,2,\dots$, and denote by $\mu(t) = E(X_i(t))$ the mean function,  by $\Xic(t) = X_i(t) -\mu(t)$ the centered process, and by $\xi_{ik} = \int_\cT \Xic(t) \phi_k(t) d{t}$ the $k$th functional principal component, $k=1,2,\dots$, which satisfies  $E(\xi_{ik})=0$, $E(\xi_{ik}^2)=\lambda_k$ and $E(\xi_{ik}\xi_{il})=0$ for $k,l=1,2,\dots,\, l\neq k$. 
Trajectories can then be represented through the Karhunen--Lo\`eve decomposition, also referred to as functional principal component analysis (FPCA),
\begin{align} \label{kl}
X_i(t)=\mu(t)+\sum_{k=1}^\infty \xi_{ik}\phi_k(t), \end{align} where in practice it is often useful to consider a truncated expansion using the first $K>0$ components that explain most of the variation, for example through the fraction of variance explained or FVE criterion \citep{mull:05:4}. 

 A  common approach is to employ Riemann sums to recover the integrals that represent the projections of the trajectories on the eigenfunctions of the auto-covariance operator of the underlying stochastic process. These integrals  correspond to the functional principal components  and their approximation by Riemann sums is known to  improve  as the number of observations per subject increases \citep{mull:05:8}. However, when functional data are sparsely observed, which means that only a finite number of observations are available for each subject, this approximation is not feasible. % over time one faces the challenge that simple Riemann sums are not close to the integrals that they target,  due to the low number of support points.
To address this challenge,  \cite{mull:05:4} introduced the Principal Analysis through Conditional Expectation (PACE) approach,  which  aims to recover the underlying trajectories by targeting the best  predictions  conditional on the observations under Gaussianity assumptions and otherwise the best linear predictor. A related  approach has been Gaussian Processes \cp{shi:11, shi:14}. Best predictions of the functional principal components can be consistently estimated based on consistent nonparametric estimates of  mean and covariance functions that are obtained by pooling all observations across subjects, borrowing strength from the entire sample. 
While  these best  predictions  are unbiased, they do not lead to consistent trajectory recovery \citep{mull:05:4}.

A second scenario where consistent predictions are unavailable in the sparse case when one has only a finite number of observations per subject is the  Functional Linear Regression Model  for the relationship between a scalar or functional response $Y$  and functional predictors $X(t),\, t \in \cT$,  a compact interval \citep{rams:05,hall:07:1,shi:11,
knei:16,chio:16}, 
\begin{equation}
	E[Y\vert X]=\mu_Y + \int_\mathcal{T} \beta(t) X^c(t) dt. \label{flm}
\end{equation}
Here $\mu_Y=E(Y)$, $X^c(t)=X(t)-E(X(t))$ and the slope function $\beta$ lies in $L^2(\cT)$. %It  is a direct extension of the standard linear regression model to the case where the  predictors lie in the Hilbert space $L^2(\mathcal{T})$. 
The unavailability of consistent predictions in the functional linear model is a consequence of the fact that the integral appearing in \eqref{flm} cannot be consistently approximated in the sparse sampling case, even in the case where the slope function $\beta$ is known.
In contrast to the prediction task,  a consistent estimate of the slope function $\beta$ in model \eqref{flm} can be obtained through consistent cross-covariance estimation by pooling the sparse data \citep{mull:05:5}. 

\bco

even in the presence of sparse observations. To achieve this one  can use the fact that the linear model structure allows to express the slope in terms of the cross covariance and covariance functions of the predictor process $X$ and the response, which are quantities that can be consistently estimated under mild assumptions \citep{mull:05:5}, even in the sparse design case.  Alternative multiple imputation methods based on conditioning on both the predictor observations and the response $Y$ have also  been explored  \citep{petr:18}, and also rely on cross-covariance estimation, for which consistent estimation is available.  However, these consistent estimates of the slope parameter function $\beta$  do not lead to consistent predictions, i.e. consistent 
estimates of $E[Y\vert X]$, as explained above.
 
 \fi
 
 The behavior of estimating mean functions and covariance functions, as well as the associated eigenvalues and eigenfunctions as per \eqref{auto} in dense and sparse cases   has been the subject of 
 numerous studies \citep{cai:06,hall:07:1,chio:16,knei:16,hall:06:1,mull:10}.   Specifically, the effect  of the transition from sparse to dense designs on the convergence rates for the estimation of mean functions, auto-covariance functions and cross-covariance functions and related phase transitions have been  studied in detail  \cp{li:10, zhan:16,zhan:28}.
 However, there are only very few studies about the sparse to dense behavior of estimates of the principal components $\xi$ in FPCA \eqref{kl}  \citep{mull:05:8,dai:16:1}
 and we are not aware of any study about  the behavior of predictions in the FLM \eqref{flm}. %A novel feature of this paper  is a study of the shrinkage of conditional distributions for vectors of functional principal components and of predicted responses, conditional on the observed data. Throughout, these distributions are referred to as  predictive distributions and  their convergence towards atomic distributions as the number of observations per subject transitions from sparse to dense to fully observed is established. 

\subsection{Innovation and outline of proposed approach} In this paper we  address the challenge of obtaining consistent predictions for trajectories or responses 
when one has sparse data in the context of functional principal component analysis (FPCA)  \eqref{auto} or the functional linear model (FLM)  \eqref{flm}. Specifically, this work includes  (1) The study of the behavior of functional principal components predicted from data when the data sampling transitions from the sparse to the dense case, which  complements previous studies on the behavior of mean and covariance function estimates under this transition;  
(2) the idea to replace the inherently inconsistent point estimates of functional principal components and of responses in functional linear models under sparse sampling  by consistent estimates of predictive distributions that correspond to the conditional distributions of outcomes of interest (functional trajectories or predicted responses); these distributions  indicate where the quantities to be predicted are likely situated based on the available data without providing a precise location; of interest is the shrinkage of these predictive distributions in the transition from sparse to dense designs; (3) consistent estimates for 
predictive distributions and their convergence to the true predictive distributions. 

Our proposal is to rephrase the prediction problem for  trajectories in the FPCA case and of scalar outcomes in the FLM case by shifting the target from point prediction, i.e., the problem of predicting conditional expectations for which consistency is unachievable,  to the problem of estimating a  predictive distribution, i.e., a conditional distribution rather than a conditional expectation.  This new perspective leads to a target for which consistent estimation is indeed feasible.      
To  study the behavior of predictive distributions under Gaussian assumptions, it proves  useful to consider a map from sparse and irregularly sampled data  to a multivariate Gaussian predictive distribution for a vector of truncated functional principal components and to investigate  its  behavior as the number of observations per subject increases. 

One of our goals is to quantify the accompanying  shrinkage of the conditional predictive distributions given the data and their convergence towards a point mass located at the true but unobserved functional component scores.
To predict  the expected response $E[Y\vert X]$ in model  \eqref{flm} in the sparse case,  a feasible approach is to construct  predictive distributions for the  expected response given the information available for a subject. 
These predictive distributions  can be consistently estimated in both the Wasserstein and Kolmogorov metric  \citep{vill:03} and we adopt a  Wasserstein discrepancy measure  to assess the predictability of the response by the predictive distribution. This measure is interpretable, can be consistently recovered under mild assumptions and is supported in simulations for various sparse designs and noise levels.

The paper is structured as follows: Preliminary results are in Section \ref{S:setup}, where  the convergence of the best predicted FPCs towards the true unobserved FPCs is established when transitioning from sparse to dense data. Crucially, this study does not require distributional assumptions. The concept of representing sparse functional/longitudinal data by predictive distributions for the FPCs in the case of Gaussian processes and our main results are the theme of Section \ref{S:preliminariesGaussian}, followed by an analysis of the shrinkage of the predictive distributions towards a point mass
located at the true scores. Extensions to the shrinkage of the entire functional predictive distribution
in the $2$-Wasserstein metric are also presented in Section \ref{S:preliminariesGaussian}. This is  followed by a study of the prediction of scalar responses $Y$ in a Functional Linear Model \eqref{flm} when predictors are sparsely
observed in Section \ref{S:FLM}, extending  the concept of predictive distributions for the predictable part of the response
and assessing  the predictability of $Y$ by the predictive distribution through a Wasserstein discrepancy measure. Asymptotic results for the consistent estimation of both the predictive distributions and Wasserstein discrepancy in the sparse case are presented and a study of the  behavior of the predictive distribution in the transition from sparse to dense sampling is also included.
This is followed by simulation results in Section \ref{S:Simulations} to demonstrate the finite sample performance of the proposed methods. Finally, data illustrations for the proposed predictive distributions  are presented
in Section \ref{S:DataApplication}. The paper concludes with a discussion of the new predictive perspective in Section \ref{S:DataApplication}.  Proofs and auxiliary results can be found in the Supplement.

\section{Convergence of Predicted Functional Principal Components When Transitioning  from Sparse to Dense Sampling}\label{S:setup}
Assume that for each individual $i = 1,\dots, n$,  there is an underlying unobserved function $X_i(t)$, where the functions $X_i$ are i.i.d.\ realizations of a $L^2$-stochastic process $X(t), \, t\in \cT$, and $\cT$ is a closed and bounded interval on the real line. Without loss of generality let $\cT=[0, 1]$. Sparsely sampled and error-contaminated observations $\tX_{ij} = X_i(T_{ij}) + \epsilon_{ij}$, $j = 1, \dots, n_i$, $n_i \le N_0$ for a finite $N_0$ are obtained at random times $T_{ij}\in \cT$ that are distributed according to a continuous smooth distribution  $F_T$, and write  $\bTi=(T_{i1},\dots,T_{in_i})^T$ to denote  the vector of sampling time points for the $i$th subject. 
The following condition is required: 
\begin{enumerate}[label=(X\arabic*)]
	\item \label{a:fbelow} $\{ T_{ij}: i=1, \dots, n,\, j=1, \dots, n_i \}$ are i.i.d.\ copies of a random variable $T$ defined on $\cT$, and $n_i$ are non-random. The density $f(\cdot)$ of $T$ is bounded below, $\min_{\tinT} f(t)\ge  m_f  >0$.
\end{enumerate}
Assumption \ref{a:fbelow} is standard  \citep{zhan:16,dai:16:1} to ensure there are no systematic sampling gaps. The measurement errors $\epsilon_{ij}$ are assumed to be  i.i.d.\ with mean 0 and variance $\sigma^2$, and independent of the underlying process $X_i(\cdot)$. 
The proposed method is motivated by Gaussian distributions, but as we will show shortly, the distributional assumption can be relaxed in our key results in this section. % We do not require Gaussianity for the process $X$.
Throughout, our analysis is conditional on the random number of observations per subject $n_i$ \citep{zhan:16}. 

\bco

Denote the auto-covariance function of the process $X$ by
\begin{align*}
	\Gamma(s, t)& = \cov(X(s), X(t)    ) = \sum_{k=1}^\infty \lambda_k \phi_k(s) \phi_k(t),\quad  s,t\in\cT ,
\end{align*}
where $\lambda_1> \lambda_2 > \dots \ge 0$ are the ordered eigenvalues, where  $\sum_{k=1}^\infty \lambda_k < \infty$, and $\phi_k$, $k\ge 1$, are the orthonormal eigenfunctions associated with the Hilbert--Schmidt operator $\Xi(g)=\int_\cT \Gamma(\cdot,t) g(t) dt$, $g\in L^2(\cT)$. Define eigengaps $\delta_k=\min(\lambda_{k-1}-\lambda_{k},\lambda_{k}-\lambda_{k+1})$, $k=1,2,\dots$, and denote by $\mu(t) = E(X_i(t))$ the mean function,  by $\Xic(t) = X_i(t) -\mu(t)$ the centered process, and by $\xi_{ik} = \int_\cT \Xic(t) \phi_k(t) d{t}$ the $k$th functional principal component, $k=1,2,\dots$, which satisfies  $E(\xi_{ik})=0$, $E(\xi_{ik}^2)=\lambda_k$ and $E(\xi_{ik}\xi_{il})=0$ for $k,l=1,2,\dots,\, l\neq k$. This can be interpreted as an extension of classical  principal component analysis \cp{mull:00:6, hsin:15}.

Trajectories can then be represented through the Karhunen--Lo\`eve decomposition $X_i(t)=\mu(t)+\sum_{k=1}^\infty \xi_{ik}\phi_k(t)$, where in practice it is often useful to consider a truncated expansion using the first $K>0$ components that explain most of the variation, for example through the fraction of variance explained or FVE criterion \citep{mull:05:4}. \fi

Denoting by $\bTi=(T_{i1},\dots,T_{in_i})^T$ the sampling time points for the $i$th subject and  
writing $\bXi = (\tX_{i1}, \dots, \tX_{in_i})^T$ and conditional on $\bTi$, it follows from \eqref{auto} and \eqref{kl} that  $\cov(\tX_{ij}, \xi_{ik}\vert \bTi) = \lambda_k\phi_k(T_{ij})$, $j = 1,\dots, n_i, \, k=1,\dots, K$. Define
\[
\bPhiiK=
\begin{pmatrix}
	\phi_1(T_{i1})&\dots&\phi_K(T_{i1})\\
	\vdots & \vdots & \vdots \\
	\phi_1(T_{in_i})&\dots&\phi_K(T_{in_i})
\end{pmatrix},
\]
$\bmu_i = \E(\bXi\vert \bTi)=(\mu(T_{i1}),\dots,\mu(T_{in_i}))^T$ and the $n_i\times n_i$ conditional covariance matrix $\bSigma_i = \cov(\bXi\vert \bTi)$, for which the $(j,l)$ entry is given by $\sigma^2\delta_{jl} + \Gamma(T_{ij}, T_{il})$, where $\delta_{jl} = 1$ if $j = l$ and 0 otherwise. To predict the functional principal components  $\bxi_{iK} = (\xi_{i1}, \xi_{i2}, \dots, \xi_{iK})^T$, we utilize  best linear unbiased predictors \citep{rice:01} of $\bxi_{iK}$ given $\bXi$ and $\bT_i$, which are given by \begin{equation} \label{blp} \tbxiiK =\bLambda_K\bPhiiK^T\bSigma_i^{-1}(\bXi-\bmu_i) \,\text{ with }\,  \bLambda_K = \diag(\lambda_1, \dots, \lambda_K). \end{equation}
%, without any distributional assumptions on the process $X$.}
As  the number of observations for an individual increases as the functional sampling gets denser, the predicted functional principal components  $\tbxiiK$ converge to  their targets  $\bxi_{iK}$  under 
%, which is not the case  for the sparse situation, where 
%only best predictors can be consistently recovered \citep{mull:05:4}. Our analysis requires  
the following assumptions.
\begin{enumerate}[label=(X\arabic*)]
	\setcounter{enumi}{1}
	\item \label{a:Diff} The process $X(t)$ is continuously differentiable a.s.\  for  $t\in\cT$.
	\item \label{a:GammaDiff} $\partial\Gamma(s, t)/\partial s $ exists and is continuous, for $s,t\in\cT$.
\end{enumerate}

Assumptions \ref{a:Diff}--\ref{a:GammaDiff} are requirements for the smoothness of the original process and the covariance function, respectively.
The following result does not require Gaussian assumptions.
\begin{prop} \label{thm:xiEst}
	Suppose that \ref{a:fbelow}--\ref{a:GammaDiff} hold and the number of observations $n_i$ for the $i$th subject satisfies  $n_i=m\to\infty$, $i=1,\dots,n$. Then, for any fixed $K\ge 1$, $k=1,\dots, K$, and $i=1,\dots,n$, as $m \rightarrow \infty$, 
	\begin{align}
		|\txiik - \xiik| & = O_p(m^{-1/2}). \label{xi-tilde} 
	\end{align}
\end{prop}
Note that this result is for increasingly dense sampling across all subjects and indicates how this leads to better approximation of $\xiik$. The rate of convergence is the same as derived previously in \cite{dai:16:1}  for the functional principal components  of the derivative process $X'(t)$ under Gaussian assumptions. This previous analysis utilized  convergence results for nonparametric posterior distributions \citep{shen:02} that are tied to the Gaussian assumption, whereas this paper presents  a  novel direct approach that does not require distributional assumptions on $X$. 
%Result (\ref{xi-tilde})  is for hyothetical quantities  population values for mean and covariance functions as well as the measurement error variance, all of which  are unknown.

Next, we study scenarios where the unknown population  quantities are estimated from  the available data, and the subjects are assumed to be observed either on dense designs, with $ n_i=m\to\infty$, or on sparse designs, 
with  $2 \le n_i\le N_0<\infty$ for a fixed number $N_0<\infty$. %, reflecting few and irregularly timed  observations per subject. To  simplify notations, we will 
Consider   sequences   %abbreviations for rates of convergence 
%for the mean and covariance of the underlying stochastic process $X$,
\begin{align}
	a_{n1}&=h_\mu^2 + \left\{\frac{\log(n)}{nh_\mu}\right\}^{1/2}, \quad b_{n1} = h_G^2 + \left\{\frac{\log(n)}{nh_G^2}\right\}^{1/2} ,\nonumber \\
	a_{n2}&=h_\mu^2 + \left\{\left(1 + \frac{1}{mh_\mu}\right)\frac{\log(n)}{n}\right\}^{1/2}, \quad b_{n2} = h_G^2 + \left(1 + \frac{1}{mh_G}\right)\left\{\frac{\log(n)}{n}\right\}^{1/2} \label{eq:bwMuCov}
\end{align}
with bandwidths $h_\mu$ and $h_G$.  %Quantities $a_n$ and $b_n$ will be used in the following in dependence on the design setting as follows: 
For sparse designs, define sequences  $a_n=a_{n1}$ and $b_n=b_{n1}$, while for dense designs these sequences will be defined as    $a_n=a_{n2}$ and $b_n=b_{n2}$. Note  that for dense designs the rates $a_n$ and $b_n$ also depend on $m$. %while for sparse designs the dependency on $m$ is not relevant for asymptotic analysis as $m$ is always small.} %  For ease of presentation, we suppress the dependency of these rates on $m$ throughout.}

The estimation of mean function $\mu$ and covariance surface $\Gamma$ utilizes local linear smoothers,  in analogy  to \cite{zhan:16}, with further details in the Supplement Section~\ref{AA1}. For the covariance smoothing step  $n_i\ge 2$ is assumed throughout. % as in \cite{zhan:16}. 
The estimation of remaining population quantities such as $\sigma^2$ and eigenpairs $(\lambda_k,\phi_k)$, $k\ge 1$, is carried out analogously as in equations $(2)$ and $(3)$ in \cite{mull:05:4}.
Denote by $\hXi$ the estimated counterpart of the Hilbert--Schmidt integral operator $\Xi$ with eigenpairs $(\hlambda_k, \hphi_k)$ such that $\langle \hphi_k,\phi_k\rangle_{L^2}\ge 0$, where $\langle \cdot,\cdot \rangle_{L^2}$ denotes the $L^2$ inner product and $k\ge 1$. 

Consider the estimated functional principal components  in \eqref{kl} for a new independent subject $i^*$ that is not part of the training data sample ($i=1,\dots,n$) and  for which measurements are available over a dense but possibly irregular grid. Then as the design gets denser,  these estimates converge to the true functional principal components, irrespective of whether  the subjects in the training set are observed under  sparse or  dense designs. Specifically, for  a realization $X^*$ of the process $X$ that is  independent of $X_1,\dots,X_n$, 
assume  one has  measurements of the process $X^*$ at times $T_j^*$ ($j=1,\dots,m^*)$ with added noise $\bX^*=(X^*(T_1^*)+\epsilon_1^*,\dots,X^*(T_{m^*}^*)+\epsilon_{m^*}^*)$. Here $m^*\to\infty$ and the  errors $\epsilon_j^*$ have mean zero and variance $\sigma^2$ and are independent of all other random quantities. %The new independent subject is observed over $m^*\to\infty$ time points. % which may differ from the number $m$ of individual observations that is available under dense design settings for the training data.  In the 
Consider  % the Karhunen--Lo\`eve decomposition $X^*(t)=\mu(t)+\sum_{k=1}^\infty \xi_k^* \phi_k(t)$ and the FPC score 
estimates $\hxi_{k}^*=\hlambda_k \hphi_{k}(\bT^*)^T \hbSigma^{*-1}(\bX^*-\hat{\bmu}^*)$, where $\hat{\bmu}^*=\hat{\mu}(\bT^*):=(\hat{\mu}(T_1^{*}),\dots,\hat{\mu}(T_{m^*}^{*}))^T$, $\hphi_{k}(\bT^*)=(\hphi_{k}(T_1^{*}),  \dots,\hphi_{k}(T_{m^{*}}^{*})   )^T$, $\bT^*=(T_1^*,\dots,T_{m^*}^*)^T$, and $\hbSigma^{*-1}$ is analogous to $\bSigma_i^{-1}$ but replacing the $T_{ij}$ with $T_{j}^*$ and the population quantities by their estimated counterparts. A requirement is that 
\begin{enumerate}[label=(B\arabic*)]
	\item \label{a:eigendecay} The eigenvalues $\lambda_1 > \lambda_2 > \dots > 0$ are all distinct. 
\end{enumerate}
The following result does not require Gaussianity of $X$.
\begin{thm}  \label{thm:xiEstHat}
	Suppose that assumptions \ref{a:Diff}, \ref{a:eigendecay} and \ref{a:K}--\ref{a:Ubeta} in the Supplement Section~\ref{AA1} are satisfied. Consider either a sparse design setting where $2 \le  n_i\le N_0<\infty$ or a dense design where  $n_i=m\to\infty$, $i=1,\dots,n$. Setting  $a_n=a_{n1}$ and $b_n=b_{n1}$ for the sparse case, and $a_n=a_{n2}$ and $b_n=b_{n2}$ for the dense case, for  a new independent subject $i^*$ and $k\ge 1$, if $m^{*}(a_n+b_n)=o(1)$ as $\ntoinf$, where $m^*=m^*(n)\to\infty$, 
	\begin{equation*}
		|\hxi_{k}^* - \xi_{k}^*| = O_p(m^{*-1/2}+m^{*} (a_n+b_n)).	
	\end{equation*}
\end{thm}
{ This  result concerns  the transition of sparse to dense sampling specifically for a new subject. A related result was obtained previously in \cite{dai:16:1} in the Gaussian case. The present result is more general as it does not require Gaussian or any other distributional assumptions.
 
 Write for two sequences $\theta_n$ and $\gamma_n$ that  $\theta_n\asymp \gamma_n$ whenever $c_1\theta_n\le \gamma_n \le c_2 \theta_n$ holds for some constants $c_1,c_2>0$ as $\ntoinf$. 
For dense designs, if the number of individual observations $m=m(n)$ satisfies $m\asymp (n/\log n)^{q}$ for some $q\in [1/4,\infty)$, $h_\mu\asymp (\log n/n)^{1/4}$, $h_G\asymp (\log n/n)^{\rho}$ with $\rho \in (0,1/4)$, $\alpha$ defined in \ref{a:Ualpha} satisfies $\alpha>4$, $\beta_\gamma$ defined in \ref{a:Ubeta} is such that $\beta_\gamma> 2/(1-4\rho)$, where the assumptions are introduced in the Supplement Section~\ref{AA1}, then $a_n+b_n\asymp(\log n / n)^{2\rho}$. 

A larger value of $\rho \in (0,1/4)$ along with the existence of a suitable $\beta_\gamma=\beta_\gamma(\rho)$ as before leads to a rate $a_n+b_n$ closer to $(\log n / n)^{1/2}$. Here the choice $0<\rho<1/4$, which entails the rate for the covariance smoothing bandwidth $h_G$, is required in order to satisfy condition \ref{a:Ubeta}.
 If $m^*\asymp (a_n+b_n)^{-\rho_1}$ for some $\rho_1\in(0,1)$,  the condition $m^{*}(a_n+b_n)=o(1)$  is satisfied and the rate in Theorem \ref{thm:xiEstHat} becomes $O_p((\log n / n)^{\rho_1 \rho} + (\log n / n)^{2\rho (1-\rho_1)} )$.  Hence, larger values of $\rho\in(0,1/4)$ along with the optimal choice $\rho_1=2/3$  lead to an optimal rate arbitrarily close to $O_p((\log n / n)^{1/6})$. 

For  sparse designs, choosing bandwidths $h_\mu\asymp (\log n/n)^{1/5}$ and $h_G\asymp (\log n/n)^{1/6}$ leads to $a_n+b_n \asymp (\log n/n)^{1/3}$. Taking $m^*\asymp (a_n+b_n)^{-\rho_1}$ for some $\rho_1\in(0,1)$, the condition $m^{*}(a_n+b_n)=o(1)$ is satisfied and the rate in Theorem \ref{thm:xiEstHat} becomes $O_p((\log n / n)^{\rho_1/6} + (\log n / n)^{(1-\rho_1)/3})$. The optimal rate  then becomes $O_p((\log n / n)^{1/9})$, which is achieved when $\rho_1=2/3$.

\section{Predictive Distributions for Gaussian Processes}\label{S:preliminariesGaussian}

Using Gaussianity, for any positive integer $K$, $\bxi_{iK} = (\xi_{i1}, \xi_{i2}, \dots, \xi_{iK})^T \sim \textit{N}(0, \bLambda_K)$, where as above $\bLambda_K = \diag(\lambda_1, \dots, \lambda_K)$ and $\lambda_k=E(\xi_{ik}^2)$.  Conditional on $\bTi$, it follows that $\bxi_{iK}$ and $\bXi$ are jointly normal 
\[
\left( \begin{array}{c}
	\bXi \\ 
	\bxi_{iK}
\end{array}  \right) 
\sim N\left( \left(
\begin{array}{c}
	\bmu_i \\ 
	0
\end{array}  
\right), \left(
\begin{array}{cc}
	\bSigma_i & \bPhiiK\bLambda_K \\ 
	\bLambda_K\bPhiiK^T & \bLambda_K
\end{array}  
\right) \right).
\]
By a well-known property of multivariate normal distributions (see e.g. \cite{mard:79}),
\begin{equation} \label{eq:cond}
	\bxi_{iK}|\bXi,\bTi \sim N_K(\tbxiiK , \bSigmaiK),
\end{equation}
where $\tbxiiK = E(\bxi_{iK}|\bXi,\bTi)$ given in \eqref{blp} %=\bLambda_K\bPhiiK^T\bSigma_i^{-1}(\bXi-\bmu_i)$ 
is the best linear unbiased predictor of $\bxi_{iK}$ %given $\bXi$ and $\bT_i$,
 and $\bSigmaiK = \bLambda_K - \bLambda_K \bPhiiK^T \bSigma_i^{-1} \bPhiiK\bLambda_K$ is the conditional variance. The relation in \eqref{eq:cond} was previously exploited, for example in \cite{mull:05:4}, to construct
simultaneous confidence bands for estimated trajectories; compare also \cite{shi:14}. We refer to the conditional distribution in \eqref{eq:cond} as  {\it $K$-truncated predictive distribution} since it is a distributional representation  for the subject's truncated true but unobserved scores $\bxi_{iK}$. % and moreover can be utilized to make predictions, e.g. through its center $\tbxiiK$ \citep{mull:05:4}.

Note that (\ref{xi-tilde}) implies that the center of the $K$-truncated predictive distribution converges to the true FPCs $\bxi_{iK}$ as the design gets denser, i.e. as $m\to\infty$. 
Next, it will be shown that the entire  $K$-truncated predictive distribution shrinks to a point mass located at its true $K$-truncated FPCs. 
Recall that $\bSigmaiK$ is the conditional covariance as in \eqref{eq:cond} and for a matrix $A\in \bbR^{p\times q}$ denote by $\lVert A\rVert_{\text{op},2}=\sup_{\rVert v\rVert_2 = 1} \lVert Av\rVert_2$ the $2$-matrix norm, where $\lVert \cdot \rVert_2$ is the Euclidean norm in $\bbR^p$, $p,q>0$.  For the following, Gaussianity will be required, i.e.,  
\begin{enumerate}[label=(X\arabic*)]
	\setcounter{enumi}{3}
	\item \label{a:GaussProcess} The process $X(t)$, $t\in\cT$, and the measurement errors are jointly Gaussian.
\end{enumerate}
\begin{prop} \label{thm:xiEst2}
	Suppose that \ref{a:fbelow}--\ref{a:GaussProcess} hold and the number of observations for the $i$th subject diverges, i.e. $n_i=m\to\infty$, $i=1,\dots,n$. Then for any fixed $K\ge 1$ %and $i=1,\dots,n$,   %as $m \rightarrow \infty$
	\begin{align*}
		\lVert \bSigmaiK \rVert_{\text{op},2} & = O_p(m^{-1}).
	\end{align*}
\end{prop}
Note that Gaussianity is used only to derive the explicit form of the conditional covariance $\bSigmaiK$ of the FPCs given the data $(\bXi,\bTi)$. We are   not aware of any  other results in the literature studying the shrinkage of conditional covariance in the dense sampling case. %, which in this case does not depend on $\bXi$. 

If Gaussianity does not hold, using the explicit form of $E(\bxi_{iK}|\bXi,\bTi)$ in Section \ref{S:setup} and the relation $ \var(\bxi_{iK} \mid \bTi)=\bLambda_K$, by a conditioning argument  $\bSigmaiK \coloneqq E[\var(\bxi_{iK}|\bXi,\bTi) \mid \bTi] = \var(\bxi_{iK} \mid \bTi) - \var( E(\bxi_{iK}|\bXi,\bTi) \mid \bTi)$ share the same definition as in the Gaussian case, and therefore  Proposition~\ref{thm:xiEst2} continues to hold for this  $\bSigmaiK$}.  %under this interpretation of $\bSigmaiK$. % without Gaussianity.
Propositions~\ref{thm:xiEst} and \ref{thm:xiEst2} demonstrate  that the $K$-truncated  predictive distribution of a given subject shrinks to the true $K$-truncated FPCs $\bxi_{iK}$ at a root-$m$ rate as the number of observations per subject diverges. The size of a $K$-truncated  predictive distribution defined through the covariance norm corresponding to the Gaussian distribution \eqref{eq:cond} implicitly reflects the number of available observations.

To discuss this further, consider an independent densely measured subject $i^*$ as in Section \ref{S:setup}. The next result quantifies the shrinkage of the conditional variance corresponding to the $K$-truncated distribution as the number of 
observations for the subject $i^*$ increases.
\begin{thm} \label{thm:xiEstHat2}
	Suppose that \ref{a:Diff}, \ref{a:GaussProcess}, \ref{a:eigendecay} and \ref{a:K}--\ref{a:Ubeta} in the Supplement Section~\ref{AA1} hold. Let $K>0$ be fixed and consider either a sparse design setting when $n_i\le N_0<\infty$ or a dense design when $n_i=m\to\infty$, $i=1,\dots,n$. Set $a_n=a_{n1}$ and $b_n=b_{n1}$ for the sparse case, and $a_n=a_{n2}$ and $b_n=b_{n2}$ for the dense design. For a new independent subject $i^*$, if $m^*(a_n+b_n)=o(1)$ as $\ntoinf$, where $m^*=m^*(n)\to\infty$, 
	\begin{align*}
		\lVert  \hbSigma_{K}^{*} -\bSigma_{K}^{*}\rVert_{\text{op},2}&=O_p(a_n+b_n).
	\end{align*}
\end{thm}

As outlined in  Section \ref{S:setup},    the estimated covariance $\hbSigma_{K}^{*}$ for a new  subject $i^*$  and thus its  $K$-truncated predictive distribution can be consistently recovered. The shrinkage effect for predictive distributions from sparse to dense is illustrated in  Figure  \ref{fig:sim_shrinkage}.  % Since the entire trajectory can be recovered from the FPC scores, a distributional representation via predictive distributions for $\bxi_{iK}$ naturally leads to a corresponding predictive distribution for the latent trajectory.
\single 
\begin{figure}[t]
	\centering
	\includegraphics[width=1.89 in]{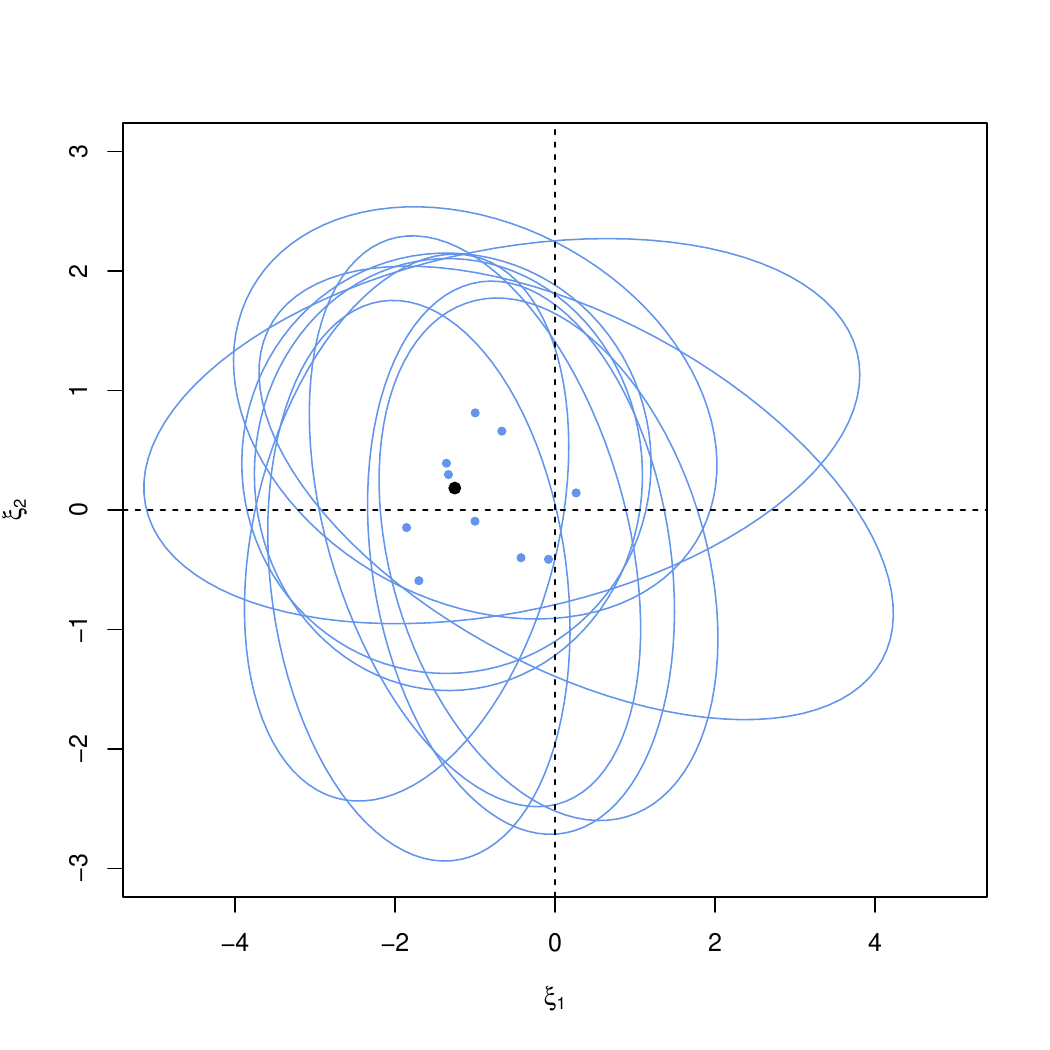}%
	\includegraphics[width=1.89 in]{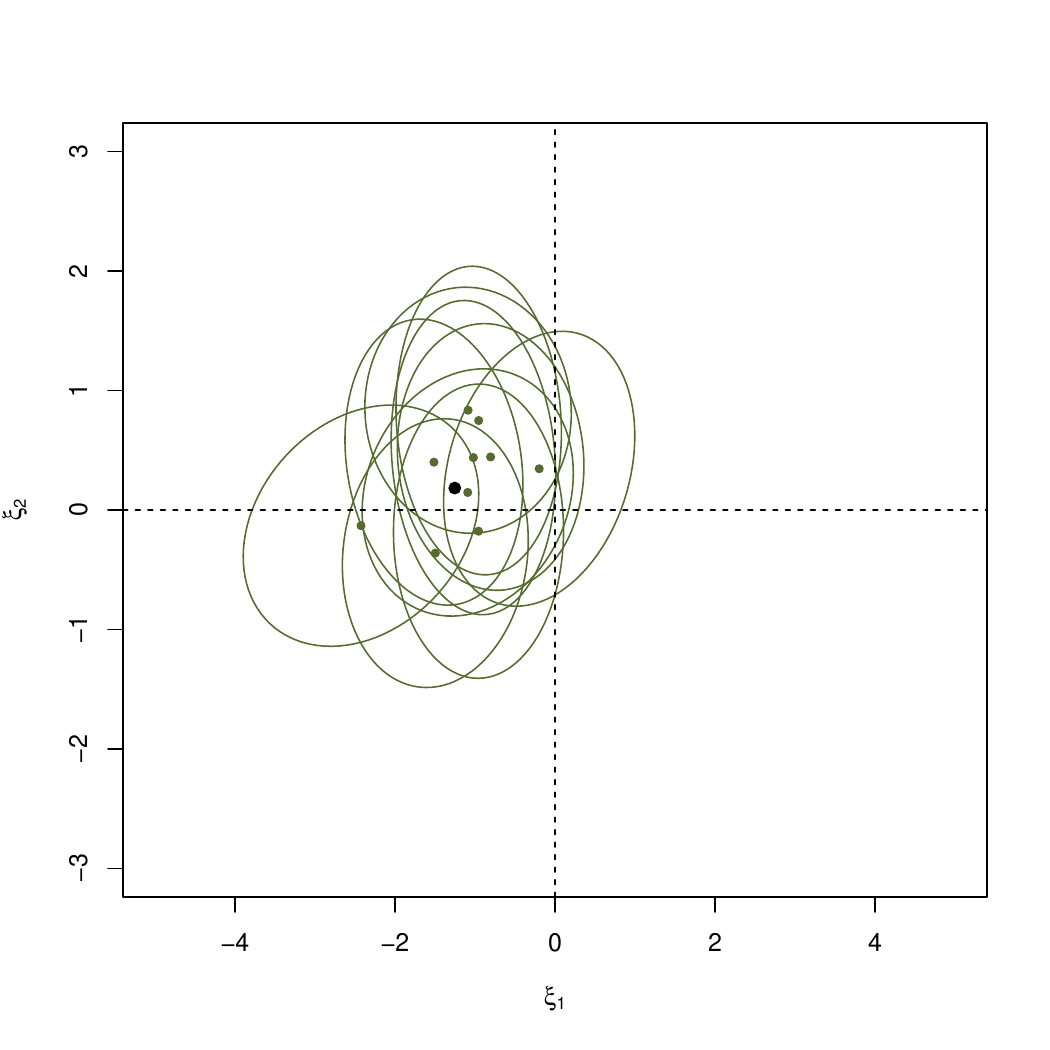}%
	\includegraphics[width=1.89 in]{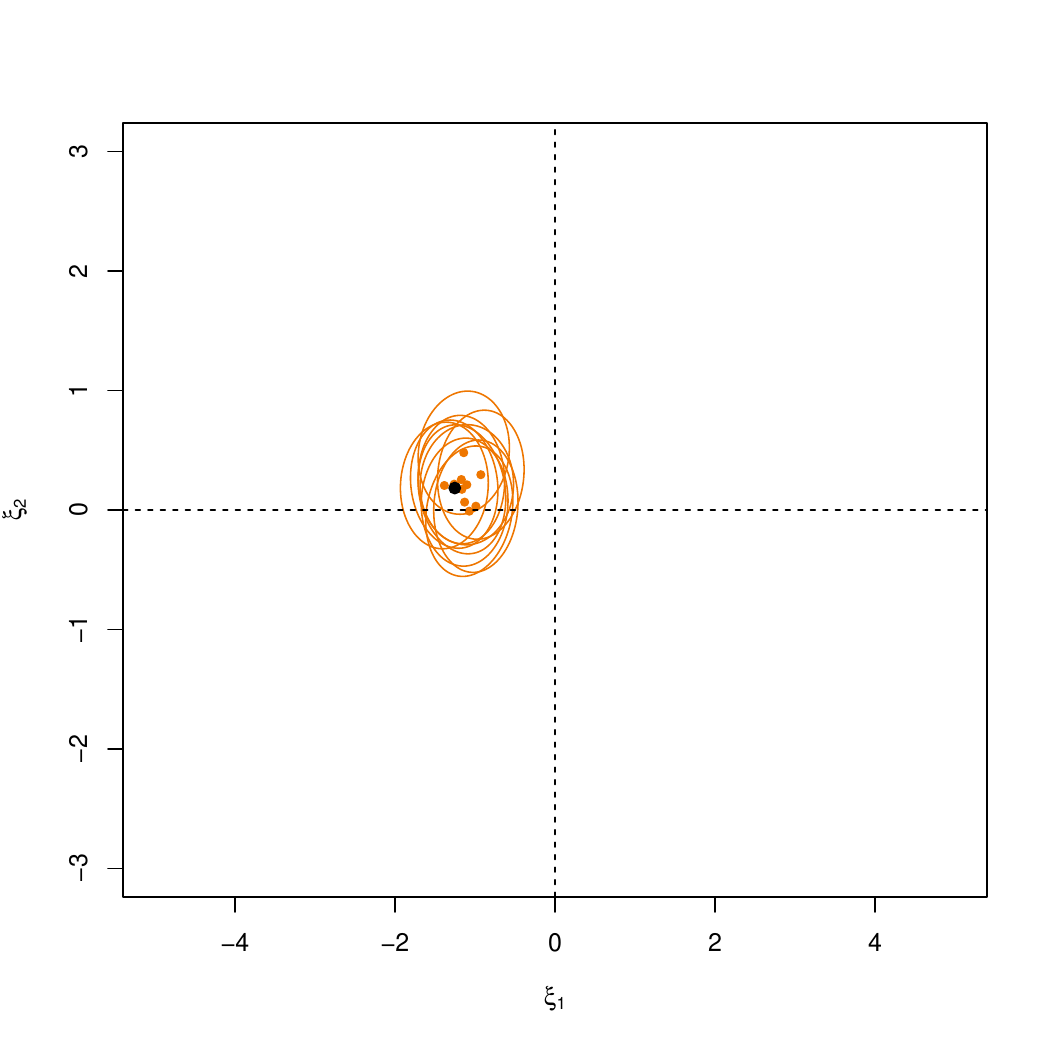}
	\caption{The $95\%$ contours for $10$ predictive distributions for the joint distribution of the first two functional principal components with $K=2$ obtained by random sampling of the data of a new subject when varying the  number of observations  $n_i$ per subject in the transition from sparse to dense, for 
	 $n_i=2$ (very sparse; left panel), $n_i=10$ (medium sparse; middle panel), and $n_i=50$ (dense; right panel), for error variance $\sigma=0.5^2$ and eigenfunctions $\phi_1(t)=-\cos(\pi t/10)/\sqrt{5}$, $\phi_2(t)=\sin(\pi t/10)/\sqrt{5}$, $\mu(t)=t+\sin(t)$, $t\in\cT=[0,10]$. The time points are sampled from a uniform distribution on $\cT$. As expected, the predictive distributions  shrink towards a point mass located at the true unobserved functional principal components (black dot) as the data gets denser. The colored dots  correspond to the centers of the simulated predictive distributions.} \label{fig:sim_shrinkage}
\end{figure}
\double
The following theoretical framework is a direct consequence of the theory of square integrable Gaussian processes. For the separable real Hilbert space $\mathcal{H}=L^2(\cT)$ with inner product $\langle \ ,\  \rangle:=\langle \ ,\  \rangle_{L^2(\cT)}$, a probability measure $\nu$ defined over the Borel sets $\mathcal{B}(\mathcal{H})$ is Gaussian if for any $h\in \mathcal{H}^*$, where $\mathcal{H}^*$ denotes the dual space consisting of continuous and linear functionals on $\mathcal{H}$, $\mu \circ h$ is a Gaussian measure on $\mathbb{R}$ \citep{gelb:90}. Such measures $\nu$ are characterized by their mean $m_\nu\in \mathcal{H}$ and covariance operator $\Xi_\nu:\mathcal{H}\to \mathcal{H}$ \citep{hkuo:75},  defined through 
\begin{align*}
	\langle m_\nu , a \rangle &=\int_{\mathcal{H}} \langle x , a \rangle \nu(dx), \quad a\in \mathcal{H},\\
	\langle \Xi_\nu(a), b \rangle &=\int_{\mathcal{H}} \langle x-m_\nu , a \rangle \langle x-m_\nu , b \rangle \nu(dx),  \quad a,b\in \mathcal{H}.
\end{align*}

Denote the Gaussian measure $\nu$ (depending on the context, in  $\mathbb{R}^p$ or $L^2$) by $\mathbb{G}(m_\nu,\Xi_\nu)$. The  $K$-truncated predictive distribution of the centered process $X_i^c(\cdot)$ given  $(\bXi,\bTi)$ is 
defined as 
\begin{equation*}
	\cGKi = (\text{The conditional distribution of } \bxi_{iK}^T\bPhi_K \mid  \bXi,\bTi) = \mathbb{G}(\tmu_{iK}, \XiiK),
\end{equation*}
where $\tmu_{iK}={\tbxi_{iK}}^T \bPhi_K$, $\bPhi_K = (\phi_1, \dots, \phi_K)^T$ are the first $K$ eigenfunctions, and $\XiiK:L^2(\cT) \rightarrow L^2(\cT)$ is the integral operator associated with the covariance function
\newline  $\Gamma_{iK}(s,t) \coloneqq \sum_{1\le j,l\le K} [\bSigmaiK]_{jl} \phi_j(s)\phi_l(t)$,  with  $[\bA]_{ij}$ denoting  the $(i,j)$th entry of  a matrix $\bA$.  This  is the functional
 %(and finite-dimensional) 
counterpart of the $K$-truncated predictive distribution in \eqref{eq:cond}.   %as it involves the first $K$ eigenfunctions $\bPhi_K$. 
We refer to $\cGKi $ as the {\it $K$-truncated predictive distribution} of the $i$th subject's  latent trajectory. %Since functional data is inherently infinite-dimensional, we may regard the truncated predictive distributions as a form of dimension reduction, where 
The $K$-truncated predictive distribution $\cGKi$ approximates the  {\it true infinite-dimensional predictive distribution},
\begin{equation} 
	\cGi = (\text{The conditional distribution of } (X - \mu) \mid \bXi,\bTi) = \mathbb{G}(\tmu_i, \Xii), \label{eq:true}
\end{equation}
where $\tmu_i = \Gamma(\cdot, \bTi)\bSigma_i^{-1}(\bXi - \bmui)$, $t\in\cT$ and $\Xii$ is the integral operator associated with the covariance function  $\Gamma_i(s,t) = \Gamma(s,t) - \Gamma(s, \bTi) \bSigma_i^{-1} \Gamma(\bTi, t)$, $s,t\in\cT$, under the convention that $\Gamma(s, \bTi)$ and $\Gamma(\bTi, t)$ are row and column vectors containing the evaluations of $\Gamma$, respectively. 

Studied next %in which way the $K$-truncated predictive distributions provide an 
is the approximation to the true latent trajectory as the truncation point $K$ increases  %in the transition from sparse to dense sampling, 
where estimated versions are obtained by replacing population quantities by their estimates, leading to the estimate  $\hcGKi = \mathbb{G}(\hmu_{iK},\hXiiK)$  of the $K$-truncated   predictive distribution $\cGKi$. Here 
$\hmu_{iK}=\hbxiiK^T \hbPhi_K$, and $\hXiiK$ is the integral operator associated with the covariance function $\hGammaiK(s,t) \coloneqq  \sum_{1\le j,l\le K} [\hbSigmaiK]_{jl} \hphi_j(s)\hphi_l(t)$. 
The corresponding infinite-dimensional version is 
\begin{equation} 
	\hcGi = \mathbb{G}(\hmu_i, \hXii), \label{eq:hattrue}
\end{equation}
where $\hmu_i(t) = \hGamma(t, \bTi)\hbSigma_i^{-1}(\bXi - \hbmui)$, $t\in\cT$ and $\hXii$ is the integral operator with kernel   $\hGamma_i(s,t) \coloneqq \hGamma(s,t) - \hGamma(s, \bTi) \hbSigma_i^{-1} \hGamma(\bTi, t)$, $s,t\in\cT$. 

To quantify the discrepancy between estimated and true %$K$-truncated 
 predictive distributions, % requires a  suitable metric for probability measures with support in the Hilbert space $L^2$. We 
 we adopt the $2$-Wasserstein distance $\cW_2$  %due to its straightforward interpretation inherited from its connection to the optimal transport problem 
 \citep{vill:03}, which  %, and since it admits a simple form for Gaussian processes on a Hilbert space as well as for the distance between a Gaussian process  and an atomic point mass. F
 for two measures $\nu$ and $\tau$ is % the $2$-Wasserstein metric is 
\begin{equation}
	\cW_2(\nu, \tau) = \left\{\inf_{A\sim \nu, B\sim \tau} E(\lVert A - B\rVert^{2})\right\}^{1/2},  \label{W2}
\end{equation}
where the norm $\lVert \cdot\rVert$ is either the Euclidean norm for measures supported on $\bbR^d$, $d\geq 1$, or $L^2$-norm for measures on the $L^2$ space, and the infimum is taken over all pairs of random variables $A$ and $B$ with marginal distribution $\nu$ and $\tau$, respectively. % \citep{vill:03}.
%The $2$-Wasserstein distance between two Gaussian measures $\mathbb{G}(m_{\mu_1}, \Xi_{\mu_1})$ and $\mathbb{G}(m_{\mu_2}, \Xi_{\mu_2})$ over the infinite-dimensional Hilbert space $L^2(\cT)$ has a  simple form \citep{gelb:90}, %given by
%	$\cW_2^2(\mathbb{G}(m_{\mu_1}, \Xi_{\mu_1}), \mathbb{G}(m_{\mu_2}, \Xi_{\mu_2})) = \lVert m_{\mu_1} - m_{\mu_2}\rVert_{L^2}^2 + \tr(\Xi_{\mu_1} + \Xi_{\mu_2} - 2(\Xi_{\mu_1}^{\half} \Xi_{\mu_2} \Xi_{\mu_1}^{\half})\to-half),$
%where
%for a positive, self-adjoint and compact operator $R:L^2(\cT)\rightarrow L^2(\cT)$, the square root operator $R^{1/2}$ is defined through its spectral decomposition \citep{hsin:15}.  % In terms of  the $L^2$-Wasserstein distance on the space of $K$-truncated predictive distributions $\cGKi$,  % defined conditionally on both the measurements $\bXi$ and time observations $\bTi$. %\ref{thm:functionalShrinkage} below shows
The shrinkage  of the % $K$-truncated  predictive 
distributions $\cGKi$ towards an atomic point mass measure $\mathcal{A}_{X_i^c}$ located at the unobserved latent centered process $X_i^c$  when the number of observations $n_i=m$ diverges and the truncation point $K=K(m)$ suitably grows with $m$ can then be characterized as follows. 
\begin{thm} \label{thm:functionalShrinkage}
	Suppose that \ref{a:fbelow}--\ref{a:GaussProcess} and \ref{a:eigendecay} hold. Consider a given subject $i\in\{1,\dots,n\}$ for which one has  $m$ measurements with $m\toinf$. If $K=K(m)\toinf$ is chosen such that $\sumkK \lambda_k\inv \asymp m^{1-\delta}$ for some  $\delta\in(1/2,1)$, then %as $m\rightarrow \infty$,
	\begin{align}
		\cW_2^2(\cGKi,\mathcal{A}_{X_i^c})&=O_p\left( m^{-(2\delta-1)} + \sum_{k=K(m)+1}^\infty \lambda_k \right). \label{Ws} 
	\end{align}
\end{thm}

The expectation that implicitly appears in the definition \eqref{W2} of the $2$-Wasserstein distance is taken here  conditionally on the data for the $i$th subject $(\bXi,\bTi)$ and the unobserved latent trajectory $X_i^c$ so that the point mass $\mathcal{A}_{X_i^c}$ is well defined. This also explains why we have an $O_p$ rather than an $O$ convergence. Shrinkage of the $K$-truncated predictive distribution towards the latent centered process is tied to the eigenvalue decay. The rate of convergence in (\ref{Ws}) can be illustrated under  polynomial and exponential eigenvalue decay:  
\begin{enumerate}[label=(D\arabic*)]
	\item \label{D:polynomialEigenDecay} $\lambda_k=k^{-\alpha_0}$ for a constant $\alpha_0>1$ and all $k\ge 1$,
\end{enumerate}
%and exponential decay rates,
\begin{enumerate}[label=(D\arabic*)]
	\setcounter{enumi}{1}
	\item \label{D:expoEigenDecay}$\lambda_k=\exp(-\alpha_1 k)$ for a constant $\alpha_1>0$ and all $k\ge 1$.
\end{enumerate}

Under polynomial decay  \ref{D:polynomialEigenDecay}, it follows that $\sumkK \lambda_k\inv \asymp K^{1+\alpha_0}$ and also $\sum_{k=K+1}^\infty \lambda_k \asymp K^{1-\alpha_0}$, so the condition in Theorem \ref{thm:functionalShrinkage} 
implies that $K\asymp m^{(1-\delta)/(1+\alpha_0)}$ and the optimal rate in (\ref{Ws}) %\ref{thm:functionalShrinkage} 
is given by $m^{(1-\alpha_0)/(1+3\alpha_0)}$. This is achieved by choosing $\delta=2\alpha_0/(1+3\alpha_0)$ and $K\asymp m^{1/(1+3\alpha_0)}$. Faster eigenvalue decay rates for larger 
$\alpha_0$ are associated with slower growth rates for $K=K(m)$ as $\delta$ approaches $2/3$. In this case the optimal rate approaches $m^{-1/3}$, which is slower than $m^{-1/2}$. The latter rate can be achieved for a finite-dimensional process,  where $\lambda_k=0$ for all $k\ge k_0$ and some $k_0>0$. Under exponential eigenvalue  \ref{D:expoEigenDecay}, the optimal rate in Theorem \ref{thm:functionalShrinkage} is again $m^{-1/3}$, which is obtained by selecting $\delta=2/3$ and $K\asymp \log(m^{1/3})$. {Note  that the result in Theorem 3 is at the population level and does not involve estimation or the sample size $n$.}

The bound \eqref {Ws} utilizes the population level $K$-truncated  predictive distribution $\cGKi$, which depends upon unknown quantities that must be estimated in practice,  which introduces additional  errors.  % that need to be taken into account. 
The following result establishes  consistency of the estimated $K$-truncated  predictive distribution counterpart $\hcGKs$ under this scenario for a new  subject as described in Section \ref{S:setup}. Let $\gamma_K(p,q)=\sumkK \lambda_k^{-p}\delta_k^{-q}$, where $p,q$ are non-negative integers and $\delta_k$ are the eigengaps. As in the previous result, $K$ is allowed to diverge.
\begin{thm} \label{thm:functionalShrinkage_estimated}
	Suppose that assumptions \ref{a:Diff}, \ref{a:GaussProcess}, \ref{a:eigendecay} and \ref{a:K}--\ref{a:Ubeta} in the Supplement Section~\ref{AA1} are satisfied. Consider either a sparse design setting with  $2 \le n_i\le N_0<\infty$ or a dense design when $n_i=m\to\infty$, $i=1,\dots,n$. Set $a_n=a_{n1}$ and $b_n=b_{n1}$ for the sparse case, and $a_n=a_{n2}$ and $b_n=b_{n2}$ for the dense case. For a new subject $i^*$, suppose that $m^*=m^*(n)\to\infty$ is such that $m^* (a_n+b_n)=o(1)$ as $\ntoinf$. If $K=K(m^*)$ satisfies $(a_n+b_n) \gamma_K(1/2,1)=o(1)$, $m^* (a_n+b_n)^2 \gamma_K(2,2)=o(1)$, $m^{*2} (a_n+b_n)^2 \gamma_K(2,0)=o(1)$, $m^{*4} (a_n+b_n)^4 \gamma_K(2,2)=o(1)$, $(a_n+b_n) \gamma_K(2,1)=o(1)$ as $\ntoinf$, and  $\sumkK \lambda_k\inv \asymp m^{*(1-\delta)}$ for some $\delta\in (1/2,1)$, then
	\begin{align*}
		\cW_2^2(\hcGKs,\mathcal{A}_{X^{*c}})
		&=o_p(1).
	\end{align*}
\end{thm}

Under polynomial eigenvalue decay  \ref{D:polynomialEigenDecay} and taking $m^*=m^*(n)\asymp (a_n+b_n)^{-q}$ for some $q\in(0,2/3)$, it follows from the proof of Theorem \ref{thm:functionalShrinkage_estimated} that the optimal rate is  $(a_n+b_n)^{q(\alpha_0-1)/(3\alpha_0+1)}$, which is achieved by taking $\delta=2\alpha_0/(3\alpha_0+1)\in (1/2,1)$ and $K\asymp m^{*(1-\delta)/(1+\alpha_0)}$. Thus the optimal rate can be
arbitrarily close to $(a_n+b_n)^{2(\alpha_0-1)/(3(1+3\alpha_0))}$ by assuming faster growth rates of $m^*$ with $q\uparrow 2/3$. Faster eigenvalue decay rates, i.e. larger values of $\alpha_0$, lead to 
a rate closer to $(a_n+b_n)^{2/9}$. If the eigenvalues exhibit exponential decay   \ref{D:expoEigenDecay} and $m^*=m^*(n)\asymp (a_n+b_n)^{-q}$, $q\in(0,1)$, the optimal rate is $(a_n+b_n)^{2/9}$.

\section{Predictive Distributions in the Functional Linear Model}\label{S:FLM}

The concept of predictive distributions is also sensible in the functional linear model \eqref{flm}  when predictor processes are sparsely sampled. Suppose one has an infinite-dimensional Gaussian predictor process $X(t)$, $t\in\cT$ with a Euclidean response $Y\in\mathbb{R}$. Utilizing the Karhunen--Lo\`eve representation 
\eqref{kl} of predictor processes $X$  and a representation of the slope function  $\beta(t)=\sum_{j=1}^\infty \beta_j \phi_j(t)$ in the eigenbasis with $\beta_j=\int_\cT \beta(t) \phi_j(t)$, $j=1,2,\dots,$ leads to  % Then the functional principal component formulation 
%which corresponds to the expansion of $\beta$ in the eigenbasis of the predictor process $X$, the FLM 
%(\ref{flm})  may be equivalently formulated as  %by connecting the response to the predictor FPCs through
\begin{equation}
	E(Y\vert X^c)=\beta_0+ \sum_{j=1}^\infty \xi_j \beta_j=:\eta. \label{eq:FLM}
\end{equation}
Here $\beta_0=E(Y)$ is the intercept and $\eta$ is the linear predictor with  responses $Y=\beta_0+ \sum_{j=1}^\infty \xi_j \beta_j+\epsilon_Y$, where $\epsilon_Y\sim N(0,\sigma_Y^2)$ is independent of all other random quantities \cp{mull:05:5,cai:06,mull:25}.

Predicting the scalar response $Y$ \citep{hall:07:1} based on a sparsely observed predictor process $X$ is clearly of  interest. % but has remained a challenging and unresolved issue. % as the functional principal components  of the predictor trajectory $X$ cannot be recovered consistently in the sparse case.  % \citep{mull:05:4}. Indeed, for a new subject whose predictor process $X$ is sparsely observed at irregular times and moreover contaminated with measurement error, e
%Even knowledge of the true $\beta$ does not make it possible to consistently estimate the part of  $Y$ that one expects to be consistently predictable, which is $\eta=E(Y\vert X^c)=\int_\cT \beta(s) X^c(s) ds$ in  the FLM (\ref{flm}). %  notwithstanding   % One might try to circumvent this basic  problem by instead targeting the expected value of the response conditional on the sparse and noisy measurements of the predictor process $X$  \citep{mull:05:5}. In contrast, 
%the substantial literature on the consistency  of estimates of the slope function $\beta$ in the case of fully observed  \citep{cai:06}
%or sparsely sampled \citep{mull:05:5} functional predictors.
When shifting  the  focus from point prediction  to 
predictive distributions of the linear predictor $\eta$, %, moving the  
%target from constructing a point prediction to that of obtaining a predictive distribution for the  response given the data available for  the subject. 
instead of targeting  the  distribution for the observed response $Y$, which   contains the additional error $\epsilon_Y$ that is independent of all other random quantities and thus is aleatoric and inherently unpredictable,  
the focus is on the distribution of  the predictable part of the response $Y$. Consider $\eta_K=\beta_0+\bbeta_K^T \bxi_K$, the truncated real-valued predictor employing the first $K$ principal components, where in practice $K$ can be chosen by a suitable criterion and $\bbeta_K=(\beta_1,\dots,\beta_K)^T$.  Thus $\eta=\eta_K + \cR_K$, where $\cR_K=\sum_{j\ge K+1}\xi_j\beta_j$ corresponds to a term  that remains unexplained by $\eta_K$. This term  decreases asymptotically as $E(\cR_K)=0$ and $\Var(\cR_K)=\sum_{j\ge K+1} \lambda_j \beta_j^2=o(1)$ as $K$ increases, where the latter rate can be further specified and can be made arbitrarily fast under additional assumptions \citep{hall:07:1}.  %For responses $Y=\eta_K+\cR_K+\epsilon_Y$, a natural predictive distribution $\cP_{iK}$ is  thus the conditional distribution of $\eta_{iK}$ given the data $\bXi$.  %and $\bTi$. 

Since $X$ is a Gaussian process, given $\bbeta_K$ one  obtains predictive distributions  \begin{equation} \label{pd} \cP_{iK} \overset{d}{=} N(\beta_0 +\bbeta_K^T \tbxiiK, \bbeta_K^T \bSigmaiK \bbeta_K) \end{equation} as before.  % which corresponds to a projection of the $K$-truncated predictive distribution of the FPC scores as derived above. 
According to Theorems 1 and 2, % and utilizes all data available for a subject. We remark that in the transition from sparse to dense sampling, the results in \ref{thm:xiEst} and \ref{thm:xiEst2} reveal that 
these predictive distributions collapse into a point mass located at the true but unobserved  predictable part $\eta_{iK}$ of the response $Y_i$ %so that one can recover the predictable component up to the truncation point
 in the transition from sparse to dense sampling.  %We formalize this result in \ref{thm:predictiveShrinkage_FLM} below.
To quantify the performance of the predictive distribution $\cP_{iK}$ \eqref{pd} in the sparse case, it is sensible to employ the 2-Wasserstein distance between two probability measures $\nu_1, \nu_2$, which for multivariate distributions is as previously given in \eqref{W2}. 
For our current purpose the predictive distributions are one-dimensional, and for this case \eqref{W2} greatly simplifies and can be expressed as  \citep{vill:03}  
\begin{align} \label{eq:WassUnif}
	\cW_2^2(\nu_1,\nu_2)&= \int_0^1 (Q_1(p)-Q_2(p))^2 dp,
\end{align}
where $Q_j(p)=\inf\{s\in\bbR \colon F_j(s)\ge p\}$, $p\in(0,1)$, is the quantile function of $\nu_j$, $j=1,2$. 

To quantify the discrepancy of this predictive distribution, it makes sense to utilize %Since the true predictive distribution is unknown and the predictable part of the response $Y$ also is not observable,  as a practical tool we utilize a Wasserstein discrepancy  $\cD_{nK}$,  
the average Wasserstein distance between $\cP_{iK}$ and the atomic measure $\cA_{Y_i}$ located at $Y_i$. Formally,
\begin{align}
	\cD_{nK}:=n\inv \sumin \cW_2^2(\cA_{Y_i}, \mathcal{P}_{iK})&=n\inv \sumin (Y_i - \teta_{iK})^2 + n\inv \sumin \bbeta_K^T \bSigmaiK \bbeta_K, \label{14} 
\end{align}
where $\teta_{iK}=E(\eta_{iK}\vert \bXi)=\beta_0+\bbeta_K^T \tbxiiK$ is the best prediction of the truncated linear predictor. % or equivalently the center of $\cP_{iK}$. 
Note that \eqref{14}  follows from \eqref{eq:WassUnif} and similar ideas as in \cite{amar:21} when computing the Wasserstein distance between the predictive distribution and an atomic measure.

If the number of observations $n_i=m_0<N_0$ is common across subjects, so that the $\bSigmaiK$ form an  i.i.d. sequence of random positive definite matrices,  the proof of Theorem 5 below %\ref{thm:WassDiscrep}) 
shows that $\cD_{nK} $ converges to the population-level Wasserstein discrepancy %measure given by 
\begin{align}
	\cD_K=2\bbeta_K^T E(\bSigma_{1K})\bbeta_K+\sigma_Y^2 %\nonumber\\ & \quad\quad 
	+\sum_{k\ge K+1} \lambda_k \beta_k^2  -2\bbeta_K^T E\Big[ \bLambda_K\bPhi_{1K}^T\bSigma_{1}\inv  \sum_{k\ge K+1} \phi_k(\bT_1)\lambda_k \beta_k \Big] \label{eq:populationWassDiscrep}.
\end{align}

The first term in \eqref{eq:populationWassDiscrep} reflects  both the number of observations and the time locations, where increased values of $m_0$ %(with the constraint of being upper bounded by $N_0$) 
lead to smaller $\bbeta_K^T E(\bSigma_{1K})\bbeta_K$ %(see \ref{thm:xiEst2}) 
and thus lower discrepancies,  i.e. increased predictability. Similarly, increased predictor and response noise levels  $\sigma^2$ and  $\sigma_Y^2$ are associated with %, respectively, are related to higher discrepancy, i.e. 
worse predictability. The last two terms come from the unexplained linear predictor part $\cR_K$ %due to the truncation at $K$ components 
and become smaller as $K$ increases. 

Consider an example with eigenbasis $\phi_k(t)=\sin(k\pi t)/\sqrt{2}$, $t\in\cT$. If the Fourier coefficients $\beta_k$ and eigenvalues $\lambda_k$ exhibit polynomial decay $\lvert \beta_k\rvert =O(k^{-\alpha_1})$ and $\lambda_k=O(k^{-\alpha_2})$,  $\alpha_1,\alpha_2>1$, the Cauchy--Schwarz inequality implies   $\sum_{k\ge K+1} \lambda_k \beta_k^2=O(K^{1-2\alpha_1-\alpha_2})$ and similarly \newline $\bbeta_K^T E[\bLambda_K\bPhi_{1K}^T\bSigma_1\inv  \sum_{k\ge K+1} \phi_k(\bT_1)\lambda_k \beta_k ]=O(K^{1-\alpha_1-\alpha_2})$ with $K^{1-\alpha_1-\alpha_2}\le K\inv$, where one uses  that $\lVert \bbeta_K\rVert_2 \le \lVert \beta \rVert_{L^2}$ and a uniform bound on the remaining quantities; see  the proof of Supplement Lemma~ \ref{lem:eigenFuncK}. 
In practice, the predictive distributions $\cP_{iK}$ and therefore also the $\cD_{nK}$ are unknown as they depend on unknown  population quantities; substituting estimates for these quantities results in estimates $\hcP_{iK}$ and $\hcD_{nK}$. %obtained by  replacing population quantities by their estimated counterparts, where  intercept
To obtain estimates of 
$\beta_0$ and slope coefficients $\bbeta_K$, one can  adopt a standard approach under the following assumption (B2), 
\begin{enumerate}[label=(B\arabic*)]
	\setcounter{enumi}{1}
	\item \label{a:betaSeries} $\lVert \beta \rVert_{L^2}^2=\sum_{m=1}^\infty \sigma_m^2/\lambda_m^2<\infty$. 
\end{enumerate}

With $C(t)=\Cov(X(t),Y)=\sumkinf E(Y \xi_k) \phi_k(t)$ denoting  the cross-covariance function between the process $X$ and response $Y$ and   $\sigma_k=\int_\cT C(t) \phi_k(t)dt =E(Y \xi_k)$, $k=1,2,\dots$, one can  estimate $C(t)$ using a local linear smoother on the raw covariances $C_i(T_{ij})=(\tX_{ij}-\hmu(T_{ij}))Y_i$ \citep{mull:05:5}, leading to an estimate $\hC(t)$ that depends  on a bandwidth $h$; see Lemma~\ref{lem:auxLemma_beta_4} in the Supplementary Material for  details. Since $\sigma_k=\lambda_k \beta_k$, under (B2), 
it holds that $\beta(t)=\sum_{m=1}^\infty \sigma_m\phi_m(t)/\lambda_m$, $t\in\cT$. %where the equality is understood in the $L^2$ sense. This 
This motivates to estimate  $\beta$  by
\begin{align*}
	\hbeta_M(t)=\sum_{m=1}^M \frac{\hsigma_m}{\hlambda_m} \hphi_m(t), \quad \tinT,
\end{align*}
where $\hsigma_k=\int_\cT \hC(t) \hphi_k(t)dt$ is an estimate of $\sigma_k$ and $M=M(n)$ is a positive integer sequence that diverges as $\ntoinf$. The intercept $\beta_0=E(Y)$ is estimated by $\hbeta_0=n\inv\sumin Y_i$. Convergence of $\hbeta_M$ towards $\beta$ is  tied to the eigengaps of $X$ \citep{cai:06,mull:10}.

With estimates $\hbeta_M$ of $\beta$ in hand, one  readily obtains estimates of the predictive distributions $\hcP_{iK}$. 
%For the following and in the sparse case, we 
We assume for simplicity that the optimal asymptotic tuning parameters  are used for estimating the mean, covariance and cross-covariance,  $h_\mu\asymp(\log n/n)^{1/5}$, $h_G\asymp(\log n/n)^{1/6}$ \citep{dai:16:1} and $h\asymp n^{-1/3}$ in the sparse design situation; in particular, this implies $c_n:=\max(a_n,b_n)\asymp(\log n/n)^{1/3}$. Defining sequences  $\upsilon_M=\sumM \delta_m\inv$, $\tau_M=\sumM \lambda_m\inv$ and a remainder term $\Theta_M=\lVert \sum_{m\ge M+1} (\sigma_m/\lambda_m)\phi_m \rVert_{L^2}$, where $\delta_m$ are the eigengaps,  note that  $M=M(n)$ should not grow too fast with sample size $n:$, %which we formalize as follows, 
\begin{enumerate}[label=(B\arabic*)]
	\setcounter{enumi}{2}
	\item \label{a:Mrate} The integer sequence $M=M(n)\to\infty$ as $n\to\infty$ is such that \newline $\sum_{m=1}^M  \lambda_m^{-1/2} \delta_m\inv=O( c_n^{\rho-1})$ for some $\rho\in(1/3,1)$,
\end{enumerate}
with an additional regularity assumption to obtain uniform convergence, 
\begin{enumerate}[label=(C\arabic*)]
	\item \label{a:LM_uncond}  There exists a scalar $\kappa_0>0$ such that $\lambda_{\min}(\bSigmaiK) \ge \kappa_0$ almost surely, for all $i\ge 1$.
\end{enumerate}
\ref{a:LM_uncond} is a  mild assumption, as  $\bSigmaiK$ corresponds to the conditional variance of  $\bxi_{iK}-\tbxi_{iK}$ given $\bTi$, which is positive definite and does shrink to zero in the sparse case, due to  $n_i\le N_0<\infty$.  

Our next result demonstrates that $\hcP_{iK}$ is consistent for $\cP_{iK}$  in the $2$-Wasserstein metric, the Kolmogorov metric %between the corresponding distribution functions, 
and  in the $L^2$ metric. % between the corresponding predictive densities. 
Let $F_{iK}$  denote  the cumulative distribution function corresponding to $\cP_{iK}$ in \eqref{pd} % \overset{d}{=} N(\beta_0 +\bbeta_K^T \tbxiiK, \bbeta_K^T \bSigmaiK \bbeta_K)$
 and  $\hF_{iK}$  that obtained by replacing $\tbxiiK$ and $\bSigmaiK$ by $\hbxiiK$ and $\hbSigmaiK$, respectively, and $\beta_0$ and $\bbeta_K$ by the above  estimates. Denote the estimated and true predictive densities by $\hat{f}_i(t)=d\hat{F}_i(t)/dt$ and $f_i(t)=dF_i(t)/dt$. The $L^2$ norm of a   function $g\colon \cT \to \bbR$ is  $\lVert g \rVert_{L^2(\bbR)}=(\int_\bbR g^2(s)ds)^{1/2}$. % denote  its $L^2$ norm over $\bbR$. 

\begin{enumerate}[label=(B\arabic*)]
	\setcounter{enumi}{3}
	\item \label{a:rate} Let $c_n = \max(a_n, b_n) \tozero$ as $\ntoinf$, where $a_n$ and $b_n$ are defined in \eqref{eq:bwMuCov}.
\end{enumerate}

\begin{thm} \label{thm:predictiveConsistency}
	Suppose that \ref{a:GaussProcess}, \ref{a:eigendecay}--\ref{a:rate}, \ref{a:K}--\ref{a:Ubeta} in the Supplement Section~\ref{AA1} hold, and consider a sparse design with  $n_i\le N_0<\infty$.  For a fixed $K \ge 1$,  setting  $a_n=a_{n1}$ and $b_n=b_{n1}$,
	\begin{align}
		\cW_2(\hcP_{iK},\cP_{iK})&=O_p(\alpha_n) \label{eq:wassCons}, \\
		\sup_{t\in\mathbb{R}} \ \lvert  \hat{F}_{iK}(t)-F_{iK}(t)\rvert  &= O_p(\alpha_n) \label{eq:wassKS}, \\
		\lVert  \hat{f}_{iK}-f_{iK} \rVert_{L^2(\mathbb{R})}  &= O_p(\alpha_n) \label{eq:wassDensL2},
	\end{align}
	as $\ntoinf$, where $\alpha_n=c_n \upsilon_M+c_n^{\rho}\tau_M^{1/2}+\Theta_M$ and the $O_p(\alpha_n)$ terms are uniform in $i$.
\end{thm}

Under the conditions of Theorem \ref{thm:predictiveConsistency}, $\alpha_n \rightarrow 0$  is a consequence of  $\tau_M\le \upsilon_M=O(c_n^{\rho-1})$, which implies $\alpha_n\le O(c_n^{(3\rho-1)/2}+\Theta_M)$. %as shown in the proof of \ref{lem:auxLemma_beta_5}. 
There is a trade-off between how fast $M$ can grow and the rate of convergence for the estimates of the population quantities, where  a larger $M$ entails a lower remainder term $\Theta_M$ but affects the rate at which $\beta$ is recovered through $\hbeta_M$, which involves $M$ components, and vice versa. Since the former term is connected to the decay of the covariance terms $\sigma_m/\lambda_m$, the optimal growth rate of $M(n)$ is inherently tied to the decay rate of $\sigma_m$, $\lambda_m$ and the eigengaps $\delta_m$.

It is of interest to consider the special case where $X$ is a  Brownian motion, for which  the $\lambda_m$ and $\phi_m$ are well-known \citep{hsin:15}. Although Brownian motion does not satisfy the smoothness assumptions required, it still  provides insight into how the convergence rate is related to the eigenvalue decay of the process. by Lemma \ref{lem:BrownianExample} in the Supplement, if $M=M(n)\asymp ( \log n/n)^{(\rho-1)/15}$, then $M$ satisfies \ref{a:Mrate} with $\sumM \lambda_m^{-1/2}\delta_m\inv \asymp c_n^{\rho-1}$. Moreover, if the decay of $\sigma_m$ is such that $\sigma_m^2 \le C m^{-(8+\delta)}$ for some constant $C>0$ and $\delta>0$, then  \ref{a:betaSeries} is satisfied, the remainder $\Theta_M=O\left( M^{-(1+\delta/2)}\right)$ and the rate $\alpha_n$ satisfies the following conditions as stated in Lemma \ref{lem:BrownianExample}: $\,\,$ If $\rho\le (5+\delta)/(15+\delta)$, then $\alpha_n=O((\log n/n)^{(13\rho-3)/30})$ while if $\rho> (5+\delta)/(15+\delta)$ it holds that $\alpha_n=O((\log n/n)^{(1-\rho)(1+\delta/2)/15})$. The optimal rate is achieved when $\rho= (5+\delta)/(15+\delta)$ and leads to $\alpha_n=O((\log n/n)^q)$, where $q=((2+\delta)/(15+\delta))/3$. A sufficiently large $\delta$ implies that $q$ is closer to $1/3$ so that the rate $\alpha_n$ approaches $c_n=(\log n/n)^{1/3}$, which is the rate at which population quantities such as the covariance function $\Gamma$ are uniformly recovered (see e.g. Theorem $5.2$ in \cite{zhan:16}).

Regarding the Wasserstein discrepancy $\cD_{nK}$,  the proposed predictability measure and the response measurement error variance $\sigma_Y^2$ can be consistently estimated in the sparse case. Consider the special case where  the number of observations $n_i=m_0<N_0$ is common across subjects. Then the estimated Wasserstein discrepancy measure $\hcD_{nK}$ converges to the population target $\cD_K$. %, which is inherently related to the predictability of the response by the $K$-truncated predictive distribution.
\begin{thm} \label{thm:WassDiscrep}
	Suppose that \ref{a:GaussProcess}, \ref{a:eigendecay}--\ref{a:rate}, \ref{a:LM_uncond}, \ref{a:K}--\ref{a:Ubeta} in the Supplement Section~\ref{AA1} hold and consider a sparse design with  $n_i=m_0\le N_0<\infty$, setting  $a_n=a_{n1}$ and $b_n=b_{n1}$. For  $K\ge 1$, 
	\begin{align}\label{eq:WassDiscrepancy1}
		\hcD_{nK}&=\cD_{K}+O_p(\alpha_n), \quad \alpha_n=c_n \upsilon_M+c_n^{\rho}\tau_M^{1/2}+\Theta_M,  	\end{align}	and furthermore 
	\begin{align}\label{eq:WassDiscrepancy2}
		n\inv\sumin (Y_i-\bar{Y}_n)^2-	\sumM \hlambda_j \hbeta_j^2&= \sigma_Y^2 +O_p(\alpha_n) + \sum_{m\ge M+1} \lambda_m \beta_m^2\end{align} with $\bar{Y}_n=n\inv\sumin Y_i.$
\end{thm}

Of interest is also the  behavior of the estimated predictive distributions under the  transition from sparse to dense sampling for  a new independent subject $i^*$. % as in Section \ref{S:setup}.
\begin{thm} \label{thm:predictiveShrinkage_FLM}
	Suppose that assumptions \ref{a:Diff}, \ref{a:GaussProcess}, \ref{a:eigendecay}--\ref{a:rate} and \ref{a:K}--\ref{a:Ubeta} in the Supplement Section S2 are satisfied. Consider either a sparse design setting when $2 \le n_i\le N_0<\infty$ or a dense design when $n_i=m\to\infty$, $i=1,\dots,n$. Set $a_n=a_{n1}$ and $b_n=b_{n1}$ for the sparse case and $a_n=a_{n2}$ and $b_n=b_{n2}$ for the dense case. Let $K>0$ be fixed and take $h=n^{-1/3}$. For a new independent subject $i^*$, suppose that $m^*=m^*(n)\to\infty$ is such that $m^* (a_n+b_n)=o(1)$ as $\ntoinf$. Then
	\begin{align*}
		&\cW_2^2(\cP_{K}^*,\cA_{\beta_0 +\bbeta_K^T \bxi_K^*})=O_p(m^{*-1}),\\
		&\cW_2^2(\hcP_{K}^*,\cA_{\beta_0 +\bbeta_K^T \bxi_K^*})
		=
		O_p\left( m^{*2} (a_n+b_n)^2 + m^{*-1} + a_n+b_n +r_n^{*2}\right),
	\end{align*}
	where $r_n^*=c_n \upsilon_M+c_n^{\rho}\tau_M^{1/2} + \tau_M  \Big[ n^{-1/3}+a_n\Big]+ \Theta_M$.
\end{thm}

\section{Simulations}\label{S:Simulations}

To illustrate the theoretical results in  Propositions \autoref{thm:xiEst} and \autoref{thm:xiEst2} pertaining to the convergence of the best linear unbiased predictors $\txiik$ to the FPC scores $\xiik$ and the shrinkage of the conditional variance $\bSigmaiK$ in the transition from sparse to dense sampling designs,  consider a finite-dimensional Gaussian process $X(t)$, $t\in \cT=[0,10]$, generated from four principal components with population quantities given by $\phi_1(t)=-\cos(\pi t/10)/\sqrt{5}$, $\phi_k(t)=\sin((2k-3)\pi t/10)/\sqrt{5}$, $k=2,\dots, 4$, $\mu(t)=t + \sin(t)$, $\lambda_k=5-k$, $k=1,\dots,4$, and $\sigma=0.5$. It is of interest to  consider a range of sampling designs including very sparse ($n_i=m=2$), medium sparse ($n_i=m=10$), and dense ($n_i=m=50$) designs. The time points are selected at random without replacement from an equispaced grid of $2,000$ points over $\cT$.

\autoref{fig:sim_proposition1_2} shows the boxplot for $||\tbxi_{iK} - \bxi_{iK} ||_2$ and  $||\bSigma_{iK} ||_\text{op,2}$ in the transition from sparse to dense sampling across $200$ simulations and for a truncation parameter $K=2$. Clearly, both errors terms shrink towards zero as the sampling design gets denser, indicating the convergence of each $\txiik$ to the FPC score $\xiik$, $k=1,2$, since $|\txiik - \xiik | \leq  ||\tbxi_{iK} - \bxi_{iK} ||_2$, but also the shrinkage of the entire conditional distribution.
To demonstrate how the  finite sample results  conform with  the theory  for the FLM model, consider the following population quantities: $\mu(t)=t/2$, $\lambda_k=4/(1+k)^2$, $k=1,\dots,4$, $K=4$, and the intercept and slope coefficients  $\beta_0=0.5$, $\beta_1=1$, $\beta_2=-1$, $\beta_3=0.5$ and $\beta_4=-0.5$. We investigate various noise levels for  the predictor process $X$ and response $Y$ as well as a variety of  sparse settings, where we generate $n_i=m_0$ random time points  for the $i$th subject, $i=1,\dots,n$. Here  $m_0=2$ reflects a very sparse design, $m_0=8$ a medium sparse and $m_0=20$ a dense design. 
We then select the time points at random and without replacement from an equi-spaced grid of $100$ points over $\cT$. Finally, we performed $2,000$ simulations with Julia, interfacing with  R and the  fdapace package \citep{fdapace}.

Table \ref{table:predictionWassUnif} presents the results for the Wasserstein discrepancy $\hcD_{nK}$ under different sparsity designs and noise levels in both the functional predictor  and scalar response $Y$. The discrepancy $\hcD_{nK}$  reflects the improvements in predictability for lower noise levels and under increasingly denser designs and increases monotonically in both $\sigma$ and $\sigma_Y$ and  decreases monotonically as the design becomes denser when keeping the noise level $\sigma$ and $\sigma_Y$ fixed. As an additional  measure of performance for $\cP_{iK}$, we computed the estimated $2$-Wasserstein distance  between the empirical distribution of  $\hF_{iK}(\beta_0+\int_\cT \beta(s) (X_i(s)-\mu(s))ds)$, $i=1,\dots,n$ and a uniform distribution on $(0,1)$. Further results and discussion can be found in the Supplement Section~\ref{s:AdditionalSimulation}. 

Figure \ref{fig:sim_discrepancy_boxplots} displays  boxplots for the true underlying discrepancy measure $\cD_{nK}$ as the sampling design gets denser for different noise levels in the predictor process and response, which further demonstrates the improvement in predictability for lower noise levels and denser sampling designs as observed before for the estimated discrepancy measure.
In addition, Figure \ref{fig:sim_predictive_distr} illustrates  predictive distributions for sparse to dense sampling design scenarios for  the underlying predictor process $X$ for noise level $\sigma=\sigma_Y=0.5$. Clearly, as the sampling design becomes less sparse, the predictive distributions shrink towards the predictable part of the response.

\single

\begin{figure}[H]
	\centering
	\includegraphics[scale=0.33]{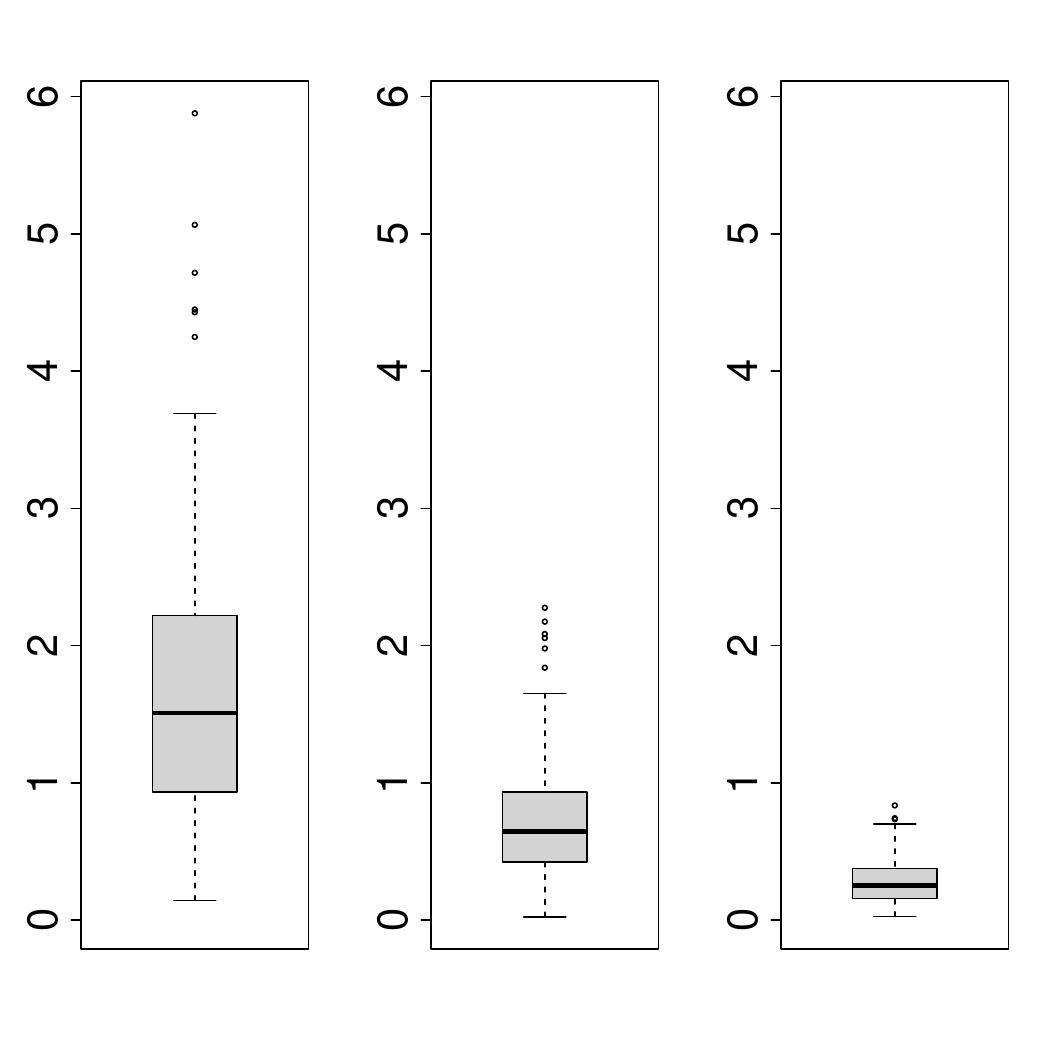}

	\includegraphics[scale=0.33]{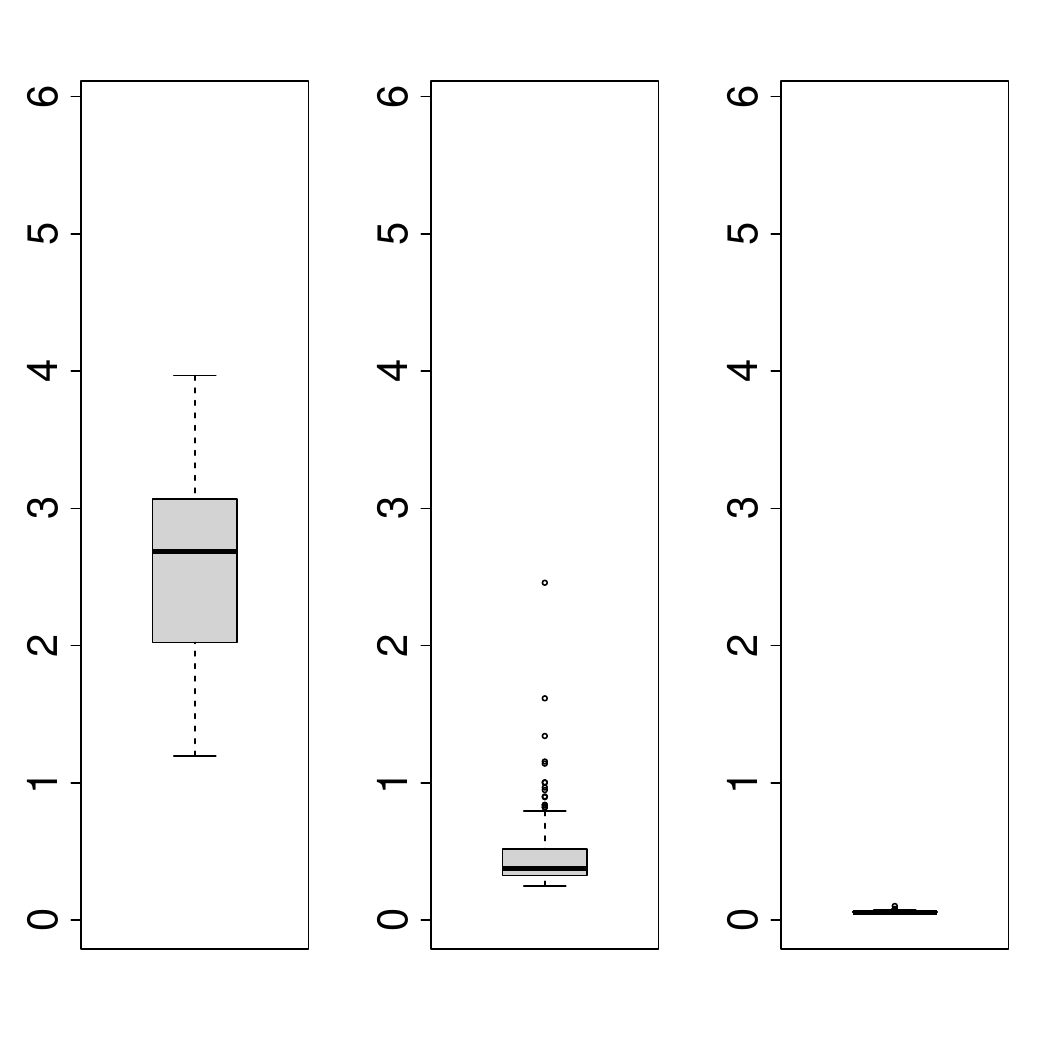}
	\caption{Simulation results illustrating  Propositions  \autoref{thm:xiEst} and \autoref{thm:xiEst2} with $K=2$. The upper panel shows  boxplots across $200$ simulations of the error term $||\tbxi_{iK} - \bxi_{iK} ||_2$ for very sparse ($m=2$, left), less  sparse ($m=10$, middle)  and  more dense ($m=50$, right) designs.  The lower panel shows the  corresponding results for $||\bSigma_{iK} ||_\text{op,2}$.} \label{fig:sim_proposition1_2}
\end{figure}%\vspace{.5cm} 

\double

\single 

\begin{table}
	\aboverulesep=0.0ex
	\belowrulesep=0.0ex
	\caption{Averages of the Wasserstein discrepancy $\hcD_{nK}$ \eqref{14},  which measures the predictability of the responses $Y_i$ in the functional linear model by the predictive distribution $\cP_{iK}$, obtained for   2000 simulation runs. The true regression parameters are $\beta_0=0.5$ and  $\bbeta_K=(1,-1,0.5,-0.5)^T$. Results are for various  predictor and response measurement error  and sparsity levels.  The very sparse case corresponds to $m=2$, the  sparse case to $m=8$ and the less sparse design case to $m=20$ observations per subject.} 
	\centering
	\begin{tabular}[c]{|c|c|c|c|c|c|c|c|}
		\hline
		\multicolumn{2}{|c|}{Measurement Error Noise level} & \multicolumn{6}{c|}{Sparsity setting} \\
		\hline
		\multicolumn{1}{|c|}{Predictor} & \multicolumn{1}{c|}{Response} & \multicolumn{2}{c|}{Very Sparse} & \multicolumn{2}{c|}{Sparse} & \multicolumn{2}{c|}{Less Sparse} \\
		\hline
		$\sigma$ & $\sigma_Y$ & $n=500$ & $n=2000$ & $n=500$ & $n=2000$ & $n=500$ & $n=2000$\\
		\hline
		& 0.5 & 3.008 & 2.645 & 1.492 & 1.477 & 0.863 & 0.853\\
		\cmidrule{2-8}
		\multirow{-2}{*}{\centering 0.5} & 1.0 & 3.863 & 3.421 & 2.255 & 2.237 & 1.612 & 1.606\\
		\cmidrule{1-8}
		1.0 & 0.5 & 3.639 & 3.449 & 2.540 & 2.418 & 1.729 & 1.715\\
		\hline
	\end{tabular}
		\label{table:predictionWassUnif}
\end{table}

\double

\section{Data Illustration}\label{S:DataApplication}

The concept of predictive distributions for longitudinal data in the context of functional linear regression models is demonstrated for the body mass index (BMI) and systolic blood pressure (SBP) data in the Baltimore Longitudinal Study of Aging \citep[BLSA,][]{shoc:84}, with sparse longitudinal measurements for each subject. This dataset has been analyzed previously in \cite{mull:05:5}, where one can find further details.  We consider a sample of 713 male subjects aged between  50 and 80 years  for which their SBP and BMI measurements are within the corresponding $1\%$ and $99\%$ quantiles across all subjects.  For the estimation of population quantities the fdapace R package \citep{fdapace} was used and  estimated predictive distributions were constructed as described in Section \ref{S:FLM}, regressing SBP (in mm Hg)  at the last age where it is measured as  scalar response against  the sparsely observed functional predictor (BMI in kg/$m^2$).  We utilize the first $K=3$ functional principal component scores of the BMI trajectory, which  explain more than $98\%$ of the variation, and choose $M=K$ components and the cross-covariance bandwidth $h$ by leave-one-out cross-validation. 

The estimated eigenfunctions are  in \autoref{fig:data_application_eigenfunctions}. They reflect the  modes of variation in the sample of functional data, where the bandwidths  used for the mean and covariance smoothing steps are $1.3$ and $2.6$, respectively. The first eigenfunction reflects a variation in the overall BMI base level across all ages,  whereas the higher order eigenfunctions reflect different BMI contrasts between younger and older ages. For example, the second eigenfunction reflects a mode of variation that differentiates higher BMI levels at ages below 62 years old from lower BMI levels afterwards.

\single 
\begin{figure}[H]
	\centering
	\includegraphics[scale=0.45]{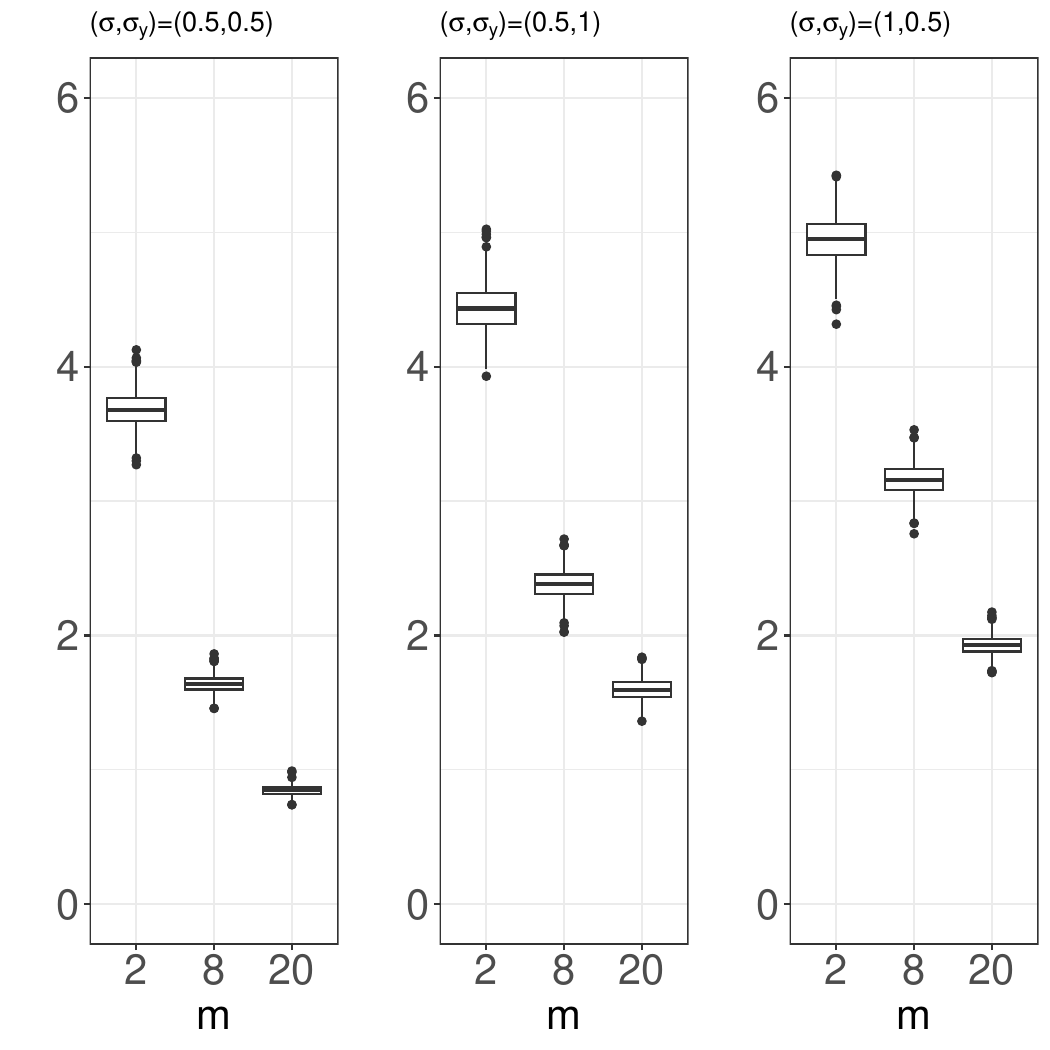}

	\caption{Boxplots of the true underlying Wasserstein discrepancy measure $\cD_{nK}$ \eqref{14} in the functional linear model  for $1000$ simulations and  sample size $n=500$, for  increasingly less sparse  sampling designs and various  noise levels for the predictor process $X$ and response $Y$.} \label{fig:sim_discrepancy_boxplots}
\end{figure}%\vspace{.5cm} 
\begin{figure}[H]
	\centering
	\includegraphics[scale=0.45]{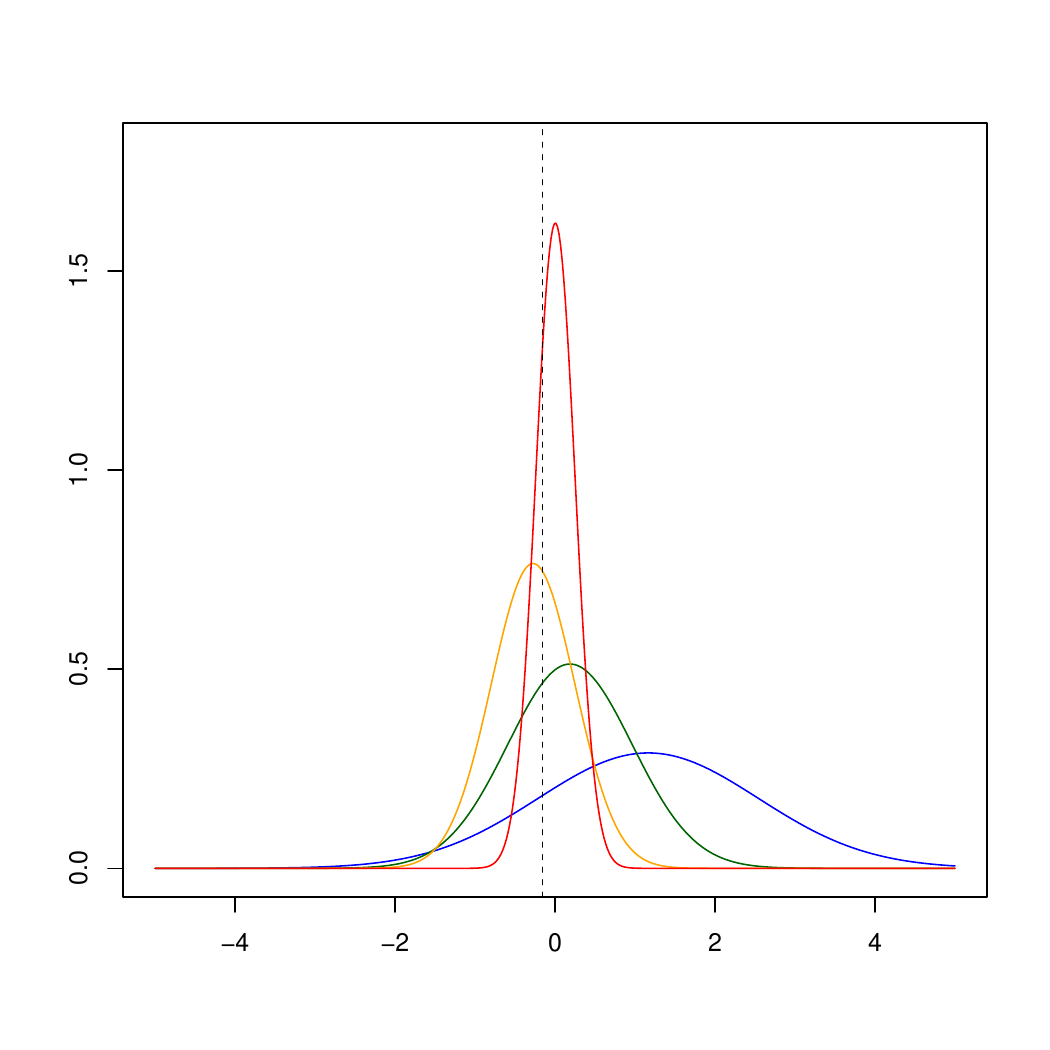}
	
	\caption{Predictive distributions $\cP_{K}$ for the response in the functional linear model obtained  by simulating different sampling design scenarios for a given realization of the  predictor process $X$, for very sparse $m=2$ (blue), sparse $m=8$ (green), less sparse  $m=20$ (orange) and dense design $m=100$ (red), with $\sigma=\sigma_Y=0.5$. The vertical line corresponds to the (unobserved) predictable part $\eta_K$ of the response.}\label{fig:sim_predictive_distr}
\end{figure}
\double 

Figure \ref{fig:predDistr_blsa} illustrates  predictive distribution intervals constructed from the $5\%$ and $95\%$ quantiles of the predictive  distribution $\hcP_{iK}$ for $20$  subjects, where we order them from lowest to largest mean of the predictive distribution. Here it is necessary to emphasize that these intervals are for the prediction intervals for $E(Y|X)$ in the functional linear model and  not for the responses $Y$, which for SBP are known to have a large variance, which means  that $E(Y|X)$ will usually be far from $Y$.

\single

\begin{figure}[H]
	\centering
	\includegraphics[width=0.3\textwidth]{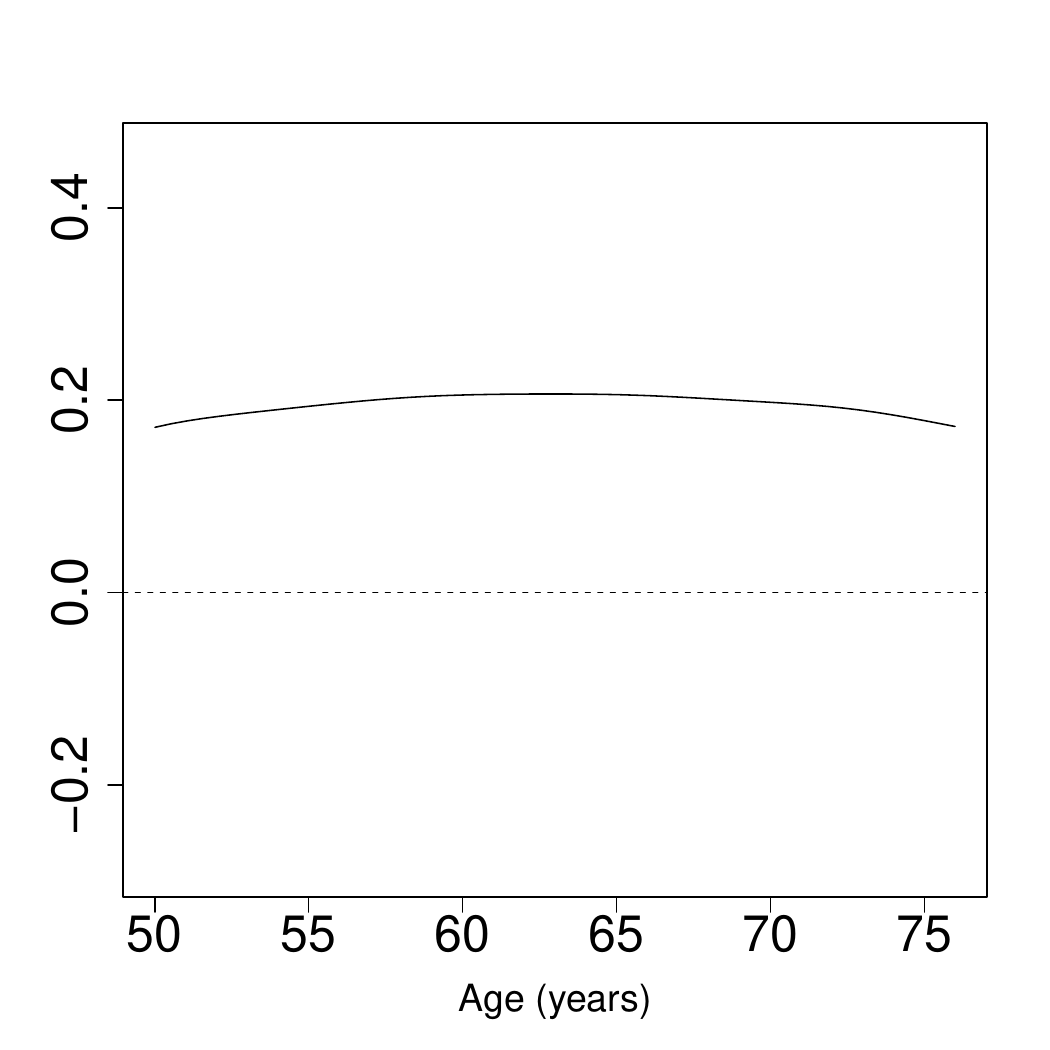}%
	\includegraphics[width=0.3\textwidth]{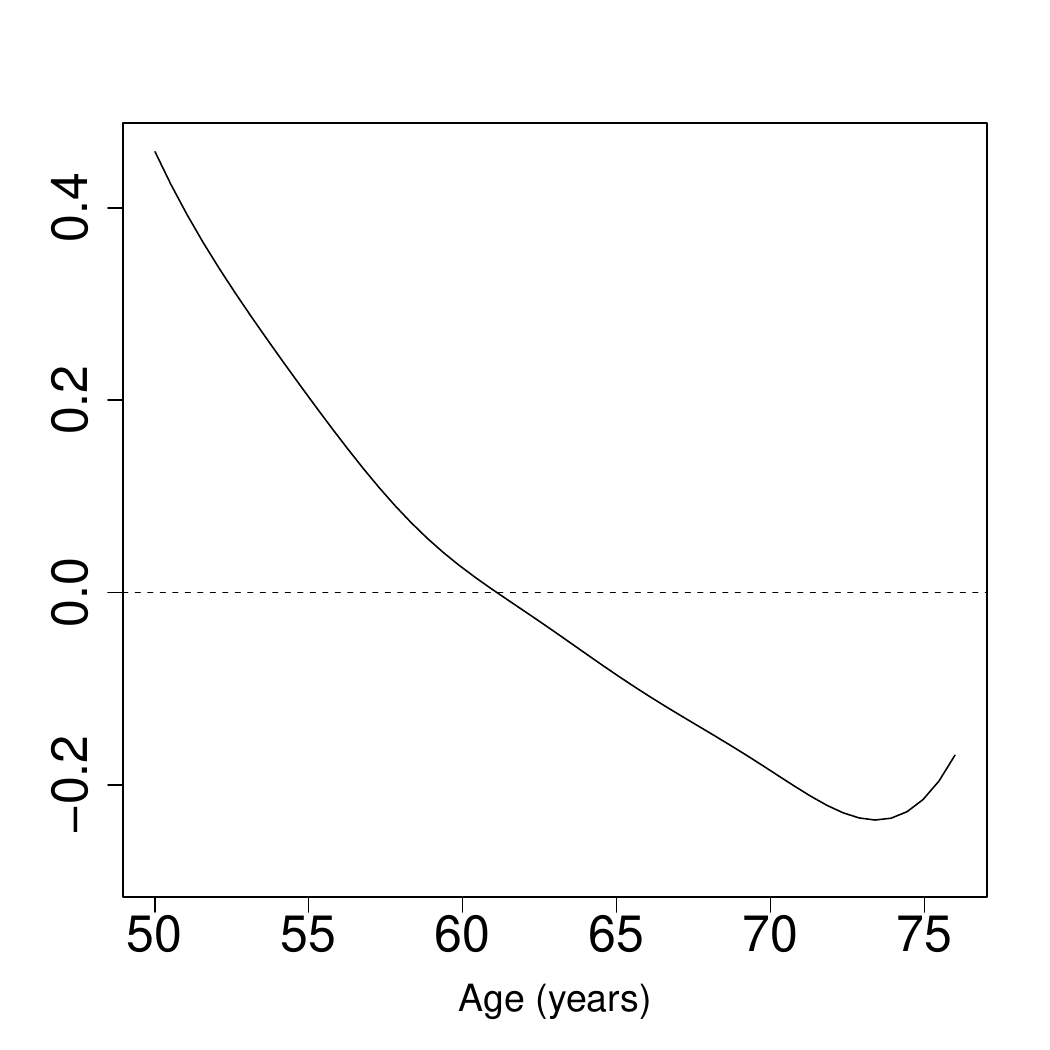}%
	\includegraphics[width=0.3\textwidth]{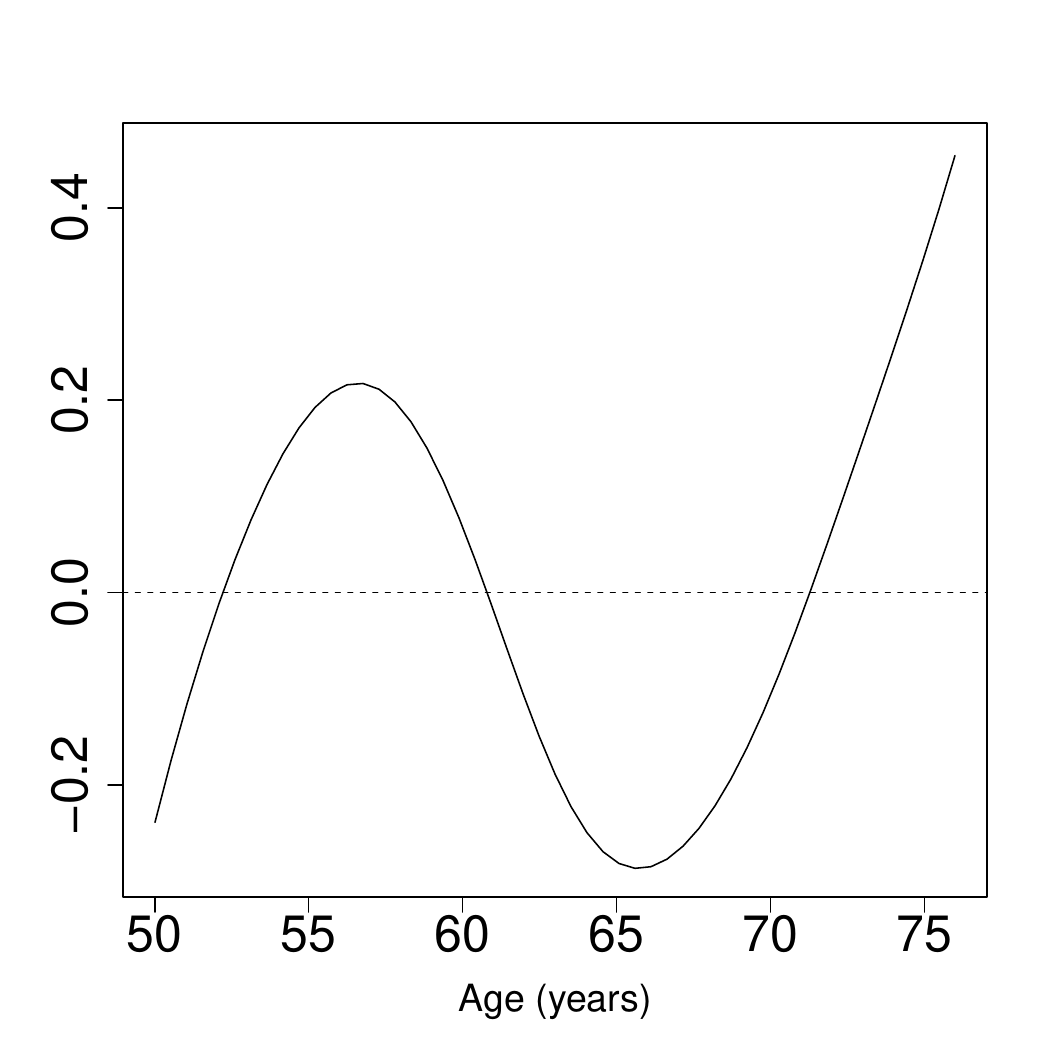}
	\caption{The first three estimated  eigenfunctions reflecting the main modes of variation in the sample of sparsely observed BMI functional data from the Baltimore Longitudinal Study of Aging.} \label{fig:data_application_eigenfunctions}
\end{figure}%\vspace{.5cm} 
\double

\section{Discussion and Concluding Remarks}\label{S:Discussion}

The main message of our paper concerns   the common scenario of sparsely observed functional data, which  covers many longitudinal designs, where one has only few time points at which noisy measurements of the function are available.  In these situations 
a point estimation perspective is not productive for the functional principal components that figure in the Karhunen--Lo\`eve expansion of FPCA  since consistent estimators are unavailable.  We advocate not to target  point estimates of functional principal components and responses given sparse measurements of functional predictors  but instead to target  predictive distributions, for which consistent estimators are available, leading to prediction regions as the targets of interest.  

The inherent uncertainty caused by the sparsity of the measurement times is even present for more densely sampled functional  data but to a lesser extent and may then be ignored.   This paper provides  a formal analysis and precise characterization of the decline in uncertainty as designs get denser.   The increasing information content in a design as it gets denser is accurately reflected in the shrinkage of the conditional distributions as delineated in  Propositions $1$ and $2$.

When one aims at a response in a functional linear regression model the predictive distribution targets the (truncated)  predictable part of the response $Y_i$, which is the part that is not contaminated with unpredictable measurement error $\epsilon_{iY}$. Therefore, the observed  $Y_i$, which includes measurement error,  is not necessarily located within the prediction region  constructed from $\cP_{iK}$.  Instead, the predictive intervals target the true truncated predictable part $\eta_{iK}=\beta_0+\bbeta_K^T \bxi_{iK}$ of the observed response $Y_i$, which is close to the linear predictor  $\eta_i=\beta_0+\sumkinf \beta_k \xi_{ik}$  for a large enough truncation point $K$.

\single
\begin{figure}[H]
	\centering
	\includegraphics[width=3.5in]{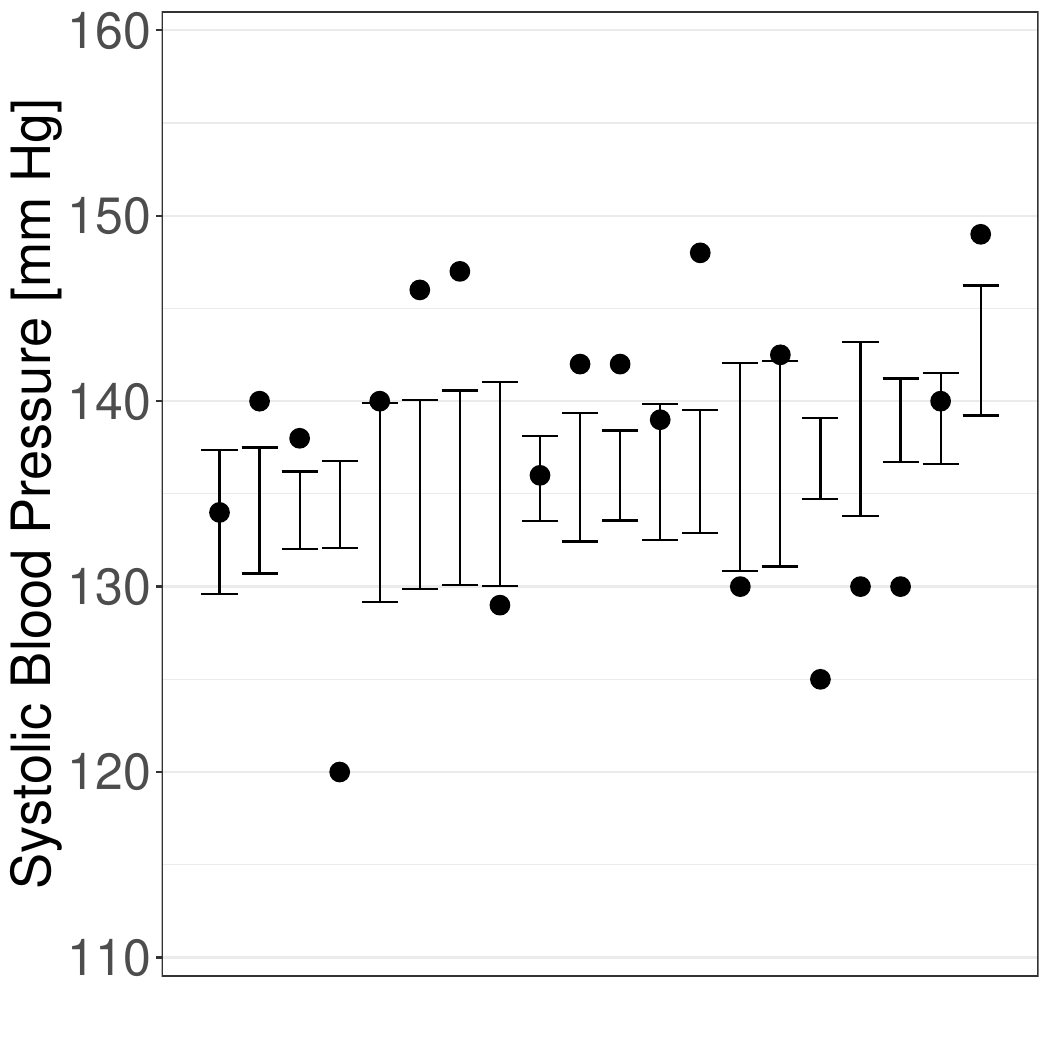}
	\caption{Predictive distribution intervals for $E(Y|X)$ where $X$ are sparsely observed BMI trajectories and $Y$ is the last observed systolic blood pressure. The intervals are ordered from left to right  by 
	size of the mean of the predictive distribution, where the interval for the smallest mean is at the left end and the interval with the largest mean at the right end. The dots are the observed responses $Y$ which carry a large random component that is unpredictable.}\label{fig:predDistr_blsa}
\end{figure}%\vspace{.5cm} 
\double

From a  practical perspective, the main implication  of the predictive distribution approach that we advocate here  is  to abandon  inconsistent point estimates of functional principal components and their associated trajectories and of predicted responses in the presence of sparsely sampled functional predictors. Instead one should focus on obtaining and using predictive distributions. Since  under Gaussian assumptions consistent estimation of these predictive distributions is feasible and theoretically supported with convergence rates, this  approach provides for valid uncertainty quantification of functional trajectories and the predictable part of the response  whenever  predictors are sparsely sampled, as is common in longitudinal designs. Obtaining these predictive distributions
is straightforward  and can also be used to simulate the effects of different sampling schemes on the uncertainty of the resulting prediction of trajectories and responses.

%%%%%%%%%%%%%%%%%%%%%%%%%%%%%%%%%%%%%%%%%%%%%%%%%%%%%%%%%%%%%%%%%%%%%%%%%%%%%%%%%%%%%%%%%%%%%%%%%%%%%%%%%%%%%%%%%%%%%%%%%%%%
\section*{Acknowledgments}

The research of XD has been supported by NSF grant DMS-2329879 and the research of 
HGM  by NSF grant  DMS-2310450. We express our thanks to the reviewers for helpful comments that led to numerous improvements.

\newpage 
%%%%%%%%%%%%%%%%%%%%%%%%%%%%%%%%%%%%%%%%%%%%%%%%%%%%%%%%%%%%%%%%%%%%%%%%%%%%%%%%%%%%%%%%%%

\centerline{\large\bf Supplement:  Additional  Results and Proofs}

\section{Additional Simulation Results} \label{s:AdditionalSimulation}

 Here we report simulation results for an additional  measure of performance for $\cP_{iK}$, where we computed the estimated $2$-Wasserstein distance  between the empirical distribution of  $\hF_{iK}(\beta_0+\int_\cT \beta(s) (X_i(s)-\mu(s))ds)$, $i=1,\dots,n$, and a uniform distribution on $(0,1)$.  This 
 is of interest, observing that \newline  $F_{1K}(\eta_{1K}),\dots,F_{nK}(\eta_{nK})$ constitute an i.i.d.\ sample from a uniform random variable $U$ in $(0,1)$. A conditioning argument gives $P(F_{iK}(\eta_{iK})\le p)=E(P(F_{iK}(\eta_{iK})\le p\vert \bXi))=E(P(\eta_{iK}\le F_{iK}\inv(p) \vert \bXi))=p$, $p\in(0,1)$. 
Thus, if we denote by $F_K(\eta_K)$ a generic probability transformation of the linear response $\eta_K$ % through the cdf corresponding to $\eta_K \vert \bX$, then 
one would expect the random variable $F_K(\eta_K)$ to be close to a uniform distribution over $(0,1)$, in terms of %the $2$-Wasserstein distance. % to measure the discrepancy between these distributions,
\begin{align} \label{eq:Wass_againstUnif}
	\cW_2^2(F_K(\eta_K),U)&= \int_0^1 (Q_K(p)-p)^2 dp,
\end{align}
where $Q_K$ is the quantile function of the random variable $F_K(\eta_K)$. Since the quantities $F_{1K}(\eta_{1K}),\dots,F_{nK}(\eta_{nK})$  are i.i.d. and share the same distribution with $F_K(\eta_K)$, we may estimate $Q_K$ by the empirical quantile of the $F_{iK}(\eta_{iK})$. %Increasing  $\sigma$ or $\sigma_Y$ is seen to entail larger values of  $\hcD_{nK}$, which means decreasing  predictability of the response.

Defining $Z_i$ to be the $i$th order statistic of the $F_{jK}(\eta_{jK})$, $j=1,\dots,n$, a natural estimate $U_\mathcal{W}$ of $W_2^2(F_K(\eta_K),U)$ in \eqref{eq:Wass_againstUnif}  is \citep{amar:21}  
\begin{align*}
	U_\mathcal{W}&=\sumin \frac{z_i^2}{n}-z_i\left(\frac{i^2}{n^2}-\frac{(i-1)^2}{n^2}\right)+\frac{1}{3}\left(\frac{i^3}{n^3}-\frac{(i-1)^3}{n^3}\right),
\end{align*}
and we define $\hat{U}_\mathcal{W}$ analogously after replacing population quantities by their estimated versions. The simulation results are in Table~\ref{tab:WassersteinDiscrepancy}. One finds that as $n$ increases, the distance $\hat{U}_\mathcal{W}$ diminishes,  which reflects better performance of the predictive distributions $\cP_{iK}$. Higher noise levels lead to worse performance as it becomes harder to estimate population quantities with the same sample size. Similarly, denser designs have a lower average value of  $\hat{U}_\mathcal{W}$ as expected. %more observations are available.

\single 

\begin{table}
	\aboverulesep=0.0ex
	\belowrulesep=0.0ex
	\caption{Simulation results for the Wasserstein discrepancy  against a uniform distribution $\hat{U}_\mathcal{W}$ defined through \eqref{eq:Wass_againstUnif} for the same settings as in Table 1, displaying the averages of $\hat{U}_\mathcal{W}$ based on 2000 simulation runs.  Averages are  scaled by a factor $1,000$. Smaller discrepancies indicate improved estimation of predictive distributions.} \label{tab:WassersteinDiscrepancy}
	
	\centering
	\begin{tabular}[c]{|c|c|c|c|c|c|c|c|}
		\hline
		\multicolumn{2}{|c|}{Measurement Error Noise level} & \multicolumn{6}{c|}{Sparsity setting} \\
		\hline
		\multicolumn{1}{|c|}{Predictor} & \multicolumn{1}{c|}{Response} & \multicolumn{2}{c|}{Very Sparse} & \multicolumn{2}{c|}{Medium Sparse} & \multicolumn{2}{c|}{Dense} \\
		\hline
		$\sigma$ & $\sigma_Y$ & $n=500$ & $n=2000$ & $n=500$ & $n=2000$ & $n=500$ & $n=2000$\\
		\hline
		& 0.5 & 1.74 & 0.62 & 0.85 & 0.46 & 0.76 & 0.37\\
		\cmidrule{2-8}
		\multirow{-2}{*}{\centering 0.5} & 1.0 & 2.18 & 0.75 & 1.22 & 0.58 & 1.25 & 0.52\\
		\cmidrule{1-8}
		1.0 & 0.5 & 2.95 & 1.54 & 1.05 & 0.44 & 0.82 & 0.45\\
		\hline
	\end{tabular}
\end{table}

\double

\vspace{1cm} 

\section{Assumptions and Main Proofs} \label{s:AssumptionsAndProofs}
\subsection{Assumptions}\label{AA1}

We assume the following regularity conditions \ref{a:K}--\ref{a:Ubeta}, which are similar to those in \cite{zhan:16} and \cite{dai:16:1}, and are compiled  here in one place to facilitate reading.  Recall that $w_i=\left(\sum_{j=1}^n n_j\right)\inv$ and $v_i=\left(\sum_{j=1}^n n_j(n_j-1)\right)\inv$.
\begin{enumerate}[label=(A\arabic*)]
	\item \label{a:K} $K(\cdot)$ is a symmetric probability density function on $[-1, 1]$ and  is Lipschitz continuous: There exists $0 < L < \infty$ such that $|K(u) - K(v)| \le L|u - v|$ for any $u, v \in [0, 1]$. 
	\item \label{a:T} $\{ T_{ij}: i=1, \dots, n,\, j=1, \dots, n_i \}$ are i.i.d.\ copies of a random variable $T$ defined on $\cT$, and $n_i$ are regarded as fixed. The density $f(\cdot)$ of $T$ is bounded below and above,
	\[ 0 < m_f \le \min_{\tinT} f(t) \le \max_{\tinT} f(t) \le M_f < \infty. \]
	Furthermore $f^{(2)}$, the second derivative of $f(\cdot)$, is bounded.
	\item \label{a:indep} $X$, $\epsilon$, and $T$ are independent.
	\item \label{a:muGDerBounded} $\mu^{(2)}(t)$ and $\partial^2 \Gamma(s, t)/\partial s^p\partial t^{2-p}$ exist and are bounded on $\cT$ and $\cT \times \cT$, respectively, for $p = 0, \dots, 2$. 
	\item \label{a:hmu} $h_\mu \tozero$, $\log(n) \sumin n_i w_i^2/h_\mu \tozero$ and $\log(n) \sumin n_i (n_i-1) w_i^2 \tozero$ . 
	\item \label{a:Ualpha} For some $\alpha > 2$, $E(\supt|X(t) - \mu(t)|^\alpha) < \infty$, $E(|\epsilon|^\alpha) < \infty$, and 
	\[ n\left[ \sumin n_iw_i^2h_\mu + \sumin n_i(n_i - 1)w_i^2h_\mu^2 \right] \left[\frac{\log(n)}{n}\right]^{2/\alpha - 1} \toinf. \]
	\item \label{a:hG} $h_G \tozero$, $\log(n)\sumin n_i(n_i - 1) v_i^2 / h_G^2 \tozero$ and $\log(n)\sumin n_i(n_i - 1)(n_i-2) v_i^2/h_G \tozero$. 
	\item \label{a:Ubeta} For some $\beta_\gamma > 2$, $E (\supt |X(t) - \mu(t)|^{2\beta_\gamma}) < \infty$, $E(|\epsilon|^{2\beta_\gamma}) < \infty$, and
	\begin{align*}
		n\bigg[ & \sumin n_i(n_i - 1)v_i^2h_G^2 + \sumin n_i(n_i-1)(n_i-2)v_i^2h_G^3 \\
		& + \sumin n_i(n_i-1)(n_i-2)(n_i-3)v_i^2h_G^4 \bigg]\left[ \frac{\log(n)}{n} \right]^{2/\beta_\gamma - 1} \toinf.
	\end{align*}
\end{enumerate}
We remark that assumption \ref{a:T} implies \ref{a:fbelow} in the main text %where the latter  is used in \ref{thm:xiEst}.
%Also we note that 
and assumption \ref{a:muGDerBounded} implies %the weaker assumption 
\ref{a:GammaDiff}.

\subsection{Additional Details for Mean and Covariance Estimation}\label{AA2}

For notational simplicity, for a function $g_1:\cT \to \bbR$ and a vector \newline $\boldsymbol{z}=(z_1,\dots,z_p)^T\in \bbR^p$, $p>0$, denote by $g_1(\boldsymbol{z})=(g_1(z_1),\dots,g_1(z_p))^T$  the application of $g_1$ to $\boldsymbol{z}$  entry-wise. Similarly, for a function $g_2:\cT\times \cT \to \bbR$ and a second vector $\boldsymbol{r}=(r_1,\dots,r_q)^T\in \bbR^q$, $q>0$, denote by $g_2(\boldsymbol{z},\boldsymbol{r}^T)$ the $p\times q$ matrix, for which the $(l,k)$ element is given by $g_2(z_l,r_k)$, where $1\le l \le p$ and $1\le k \le q $. Also, for two scalar sequences $\theta_n$ and $\gamma_n$, write $\theta_n\lesssim \gamma_n$ if there exists a constant $c_0>0$ such that $\theta_n \le c_0 \gamma_n$ holds for large enough $n$.

%\subsubsection{A.2.1 Mean and Covariance Estimation}
For the mean function estimate, set $\hmu(t)=\hat{\gamma}_0$, where
\begin{align*}
	(\hat{\gamma}_0,\hat{\gamma}= \argmin_{\gamma_0,\gamma_1} \sumin w_i \sum_{j=1}^{n_i} (X_{ij}-\gamma_0-\gamma_1(T_{ij}-t))^2 K_{h_\mu}(T_{ij}-t),
\end{align*}
where $w_i=(\sumjn n_j)\inv$ are equal subject weights, $K$ is a kernel function corresponding to a density function with compact support $[-1,1]$ and $K_{h_\mu}(\cdot)=K(\cdot/h_\mu)/h_\mu$. For the covariance surface estimate, writing  $\hat{C}_{ijl}=(X_{ij}-\hmu(T_{ij}))(X_{il}-\hmu(T_{il}))$ for  the raw covariances \citep{mull:05:4}, set $\hGamma(s,t)=\hat{\gamma}_0$, where
\begin{align*}
	&(\hat{\gamma}_0,\hat{\gamma}_1,\hat{\gamma}_2)\\
	&= \argmin_{\gamma_0,\gamma_1,\gamma_2} \sumin v_i \sum_{1\le j\neq l \le n_i} (C_{ijl}-\gamma_0-\gamma_1(T_{ij}-s)-\gamma_2(T_{il}-t))^2\\  & \times K_{h_G}(T_{ij}-s) K_{h_G}(T_{il}-t).
\end{align*}
Here  $v_i=(\sumjn n_j(n_j-1))\inv$ and $n_i\ge 2$ is assumed throughout for the covariance estimation step.

For the cross-covariance smoothing step, using  the raw covariances $C_i(T_{ij})=(\tX_{ij}-\hmu(T_{ij}))Y_i$, the local linear estimate of $C(t)$ is given by $\hC(t)=\hbeta_0^X$, where
\begin{align}
	(\hbeta_0^X,\hbeta_1^X)=\argmin_{\beta_0^X,\beta_1^X \in\bbR} \ \sumin \sum_{j=1}^{n_i} w_i K_h(T_{ij}-t) (C_i(T_{ij})-\beta_0^X -\beta_1^X (t-T_{ij}))^2\label{eq:localLinear_C},
\end{align}
with  $w_i=(\sumin n_i)\inv$.

\subsection{Proofs of Main Results in Section \ref{S:setup} and  \ref{S:preliminariesGaussian}}\label{AA3}

The proofs in this and the  following sections rely on various auxiliary results and lemmas included in Section~\ref{s:AuxiliaryResults}. 
%Additional proofs of theorems and propositions are also included in the Supplement.

\begin{proof}[Proof of Proposition \ref{thm:xiEst}]
	Fix $i\in \{1,\dots,n\}$ and $k\in\mathbb{N}$, and recall that
	\begin{equation} \label{eq:phikInner}
		\txi_{ik} = \lambda_k \bphi_{ik}^T \bSigma_i^{-1}(\bXi-\bmu_i),
	\end{equation}
	where $\bphi_{ik}=(\phi_k(T_{i1}),\dots,\phi_k(T_{im}))^T$. Define $\bW = \diag(w_l)$, where $w_l$ are quadrature weights chosen according to the left endpoint rule, i.e. $w_l = T_{il} - \max_{j: T_{ij} < T_{il}} T_{ij}$ for $l = 1, \dots, m$, and we set $\max_{j: T_{ij} < T_{il}} T_{ij}=0$ whenever $\{j: T_{ij} < T_{il}\}=\emptyset$. Let $g_m$ be the size of the maximal gap between $\{0, T_{i1}, \dots, T_{im}, 1\}$ for $\cT = [0, 1]$ and consider the quadrature approximation errors
	\[
	\be_k = \int_\cT \Gamma(\bTi, t) \phi_k(t) dt - \bSigma_i \bW \bphi_{ik},
	\]
	where $\Gamma(\bTi, t) = (\Gamma(T_{i1}, t), \dots, \Gamma(T_{im}, t))^T$. Here note that since $\bSigma_i=\sigma^2 I_m + \Gamma(\bTi,\bTi^T)$, where $\Gamma(\bTi,\bTi^T)$ corresponds to the matrix with elements $[\Gamma(\bTi,\bTi^T)]_{jl}=\Gamma(T_{ij},T_{il})$, $j,l\in\{1,\dots,m\}$,  we have $\bSigma_i \bW \bphi_{ik}= \sigma^2 \bW \bphi_{ik} +\Gamma(\bTi,\bTi^T) \bW \bphi_{ik} $ where the second term in the previous expression corresponds to the numerical quadrature approximation to $\int_\cT \Gamma(\bTi, t) \phi_k(t) dt$ and the first term will be shown to be negligible as $m\to\infty$. 
	
	From the quadrature approximation error for integrating a continuously differentiable function $g$ over $[0,1]$ under the left-endpoint rule and denoting $T_{i}^{(m)}:=\max_{1\leq j\leq m} T_{ij}$ we have
	\begin{align}
		&\left|\int_{0}^{1} g(t) dt - \sumlm g(T_{il}) w_l\right|  \nonumber 
		\\ & \le  \frac{\sup_{\tinT}|g'(t)|}{2} \left( \sumlm w_l^2 + (1-T_{i}^{(m)})^2   \right) + |(1-T_{i}^{(m)}) g(1)| \label{eq:intErrBound} \\
		& = O_p(m^{-1}), \label{eq:intErr}
	\end{align}
	where \eqref{eq:intErr} follows from Lemma \ref{lem:unifGaps}. Denoting by $\norm{\cdot}_2$ the Euclidean norm in $\mathbb{R}^m$, we have
	\begin{align}\label{eq:ek}
		\norm{\be_k}_2 & \le \norm{\int_\cT \Gamma(\bTi, t) \phi_k(t) dt - \Gamma(\bTi,\bTi^T)\bW\bphi_{ik}}_2 + \norm{\sigma^2 \bW \bphi_{ik}}_2  = O_p(m^{-1/2}),
	\end{align}
	which follows by noting that the integration error rates for all entries in $\be_k$ are uniform due to Condition~\ref{a:GammaDiff} in the main text and \eqref{eq:intErrBound}, and that 
	\begin{equation} \label{eq:Wphi}
		\norm{\bW\bphiik}_2^2 \le \sumlm w_l^2 \supt \phi_k^2(t) = O_p(m^{-1}). 
	\end{equation}
	Since
	\begin{equation}
		\lambda_k\bphi_{ik} = \bSigma_i \bW \bphi_{ik} + \be_k, \label{eq:phiSigmaInv}
	\end{equation}
	we have
	\begin{align}
		\lambda_k \bphi_{ik}^T \bSigma_i^{-1}(\bXi-\bmu_i) & = \bphi_{ik}^T\bW(\bXi-\bmu_i) + \be_k^T \bSigma_i^{-1} (\bXi-\bmu_i)  \nonumber \\ 
		& = \bphi_{ik}^T\bW(\bYi- \bmui) + \bphi_{ik}^T\bW\bepsi + \be_k^T\bSigma_i^{-1}(\bXi-\bmui), \label{eq:phi3terms}
	\end{align}
	where $\bYi=(X_{i1},\dots,X_{im})^T$ and $\bepsi=(\epsilon_{i1},\dots,\epsilon_{im})^T$. Let $g_k(t)=\phi_k(t) (X_i(t)-\mu(t))$. Then, from Condition~\ref{a:GammaDiff} and since the process $X_i(t)$ is assumed continuously differentiable almost surely, we have $g_k(t)$ is continuously differentiable a.s.  over the compact set $\cT=[0,1]$ so that $\sup_{\tinT}|g_k'(t)|=O_p(1)$. Thus, using \eqref{eq:intErrBound} and the fact that $\int_{0}^{1} \phi_k(t) (X_i(t)-\mu(t)) dt = \xi_{ik}$, we obtain
	\begin{gather*}
		\xi_{ik} - \sumlm  \phi_k(T_{il}) (X_i(T_{il})-\mu(T_{il})) w_l =O_p(m\inv),
	\end{gather*}
	whence
	\begin{align}
		\bphi_{ik}^T\bW(\bYi- \bmui)& = \xi_{ik} + O_p(m^{-1}). \label{eq:phi1} 
	\end{align}
	By conditioning and using the independence between $\bepsi$ and $\bT_i$, 
	$E(\bphi_{ik}^T\bW\bepsi)^2=E[E (\bphi_{ik}^T\bW\bepsi \bepsi^T \bW \bphi_{ik} \vert \bTi)]$
	
	$\quad =E[\bphi_{ik}^T\bW E (\bepsi \bepsi^T) \bW \bphi_{ik}]=\sigma^2 E(\norm{\bW\bphiik}_2^2).$  
	
	Hence, from \eqref{eq:Wphi} it follows that $E(\bphi_{ik}^T\bW\bepsi)^2=O(m\inv)$ and thus
	\begin{align}
		\bphi_{ik}^T\bW\bepsi & = O_p(m^{-1/2}) \label{eq:phi2}.
	\end{align}
	We now show that $Z_m:=\be_k^T\bSigma_i^{-1}(\bXi-\bmui)=O_p(m^{-1/2})$. Note that for any $M>0$ 
	\begin{align}
		P\left(\sqrt{m}\ \lvert Z_m \rvert>M \vert \bTi\right) \le \frac{m}{M^2} \lVert \be_k \rVert_2^2 \lVert \bSigma_i^{-1/2} \rVert_{\text{op},2}^2 \le \frac{m}{M^2 \sigma^2} \lVert \be_k \rVert_2^2 \label{eq:ineqZm},
	\end{align}
	where the last inequality follows since $\lVert \bSigma_i^{-1/2} \rVert_{\text{op},2} \leq \sigma\inv$. From \eqref{eq:ek},  $m \norm{\be_k}_2^2=O_p(1)$ and thus for any $\epsilon>0$ there exist $M_0=M_0(\epsilon)>0$ and $m_0=m_0(\epsilon)\in \mathbb{N}^+$ such that
	\begin{align}
		P\left( m \norm{\be_k}_2^2 >M_0 \right) \le \epsilon, \quad \forall m\ge m_0 \label{eq:mek}.
	\end{align}
	Hence, by choosing $M=M_\epsilon:=\sqrt{M_0/(\epsilon \sigma^2)}$ and defining \newline  $u_{im}:=P\left(\sqrt{m}\ \lvert Z_m \rvert>M \vert \bTi\right)$, 
	\begin{align}
		& P\left( \sqrt{m}\ \lvert Z_m \rvert>M_\epsilon \right)=E[u_{im}] \nonumber \\& \quad = E[u_{im}1_{\{u_{im}\le \epsilon\}}+u_{im} 1_{\{u_{im} >\epsilon\}}]\le  \epsilon+ P(u_{im}>\epsilon)\label{eq:ineqCondTi},
	\end{align}
	where the last inequality follows since $u_{im}\le 1$. Now  \eqref{eq:ineqZm} and \eqref{eq:mek} imply $P(u_{im}>\epsilon) \le \epsilon$ for $m\ge m_0$, whence 
	\begin{align}
		P\left( \sqrt{m}\ \lvert \be_k^T\bSigma_i\inv (\bXi-\bmui) \rvert>M_\epsilon \right)\le 2\epsilon, \quad \forall m\ge m_0\label{eq:phi3},
	\end{align}
	which shows that $\be_k^T\bSigma_i\inv (\bXi-\bmui)=O_p(m^{-1/2})$. The result follows by combining \eqref{eq:phi3terms}, \eqref{eq:phi1}, \eqref{eq:phi2} and \eqref{eq:phi3}.
\end{proof}

\begin{proof}[Proof of Theorem \ref{thm:xiEstHat}]
	Let $K_0\ge k$ be any fixed integer and consider the constant sequence $K=K(n)=K_0$, for all $n\ge 1$. Thus $(a_n+b_n) \sumkK \lambda_k\inv =o(1)$ as $\ntoinf$ and similar arguments as the ones outlined in the proof of Lemma \ref{lem:secondMomentFPCscoresNorm} leads to
	\begin{align*}
		(\hxi_{k}^*-\txi_{k}^*)^2
		&\lesssim
		(\hbe_k^{*T} \hbSigma^{*-1} (\bX^*-\hbmu^*))^2
		+
		(\hbphi_{k}^{*T} \bW^* (\bX^*-\hbmu^*)
		-
		\bphi_{k}^{*T} \bW^* (\bX^*-\bmu^*))^2\\
		&+
		(\be_k^{*T} \bSigma^{*-1} (\bX^*-\bmu^*))^2\\
		&=
		O_p\left(
		m^{*-1}+m^{*2} (a_n+b_n)^2 \right),
	\end{align*}
	where $\hbe_k^*$ is defined as in \eqref{eq:hbe_k*definition}. The result follows.
\end{proof}

%\subsection{A4. Proofs of Main Results in Section \ref{S:preliminariesGaussian}}\label{AA4}

\begin{proof}[Proof of Theorem \ref{thm:functionalShrinkage}]
	Recall that $\tmu_{iK}=\tbxiiK^T \bPhi_{K}$ and $K=K(m)$ satisfies $\sumkK \lambda_k\inv \asymp m^{1-\delta}$, where $\delta\in(1/2,1)$ and $\tbxiiK=\bLambda_K\bPhiiK^T\bSigma_i^{-1}(\bXi-\bmu_i)$. We first show shrinkage of $\lVert \tmu_{iK}-\sumkinf \xi_{ik}\phi_k\rVert_{L^2}$. Also, for any $k\ge 1$ define 
	\begin{align*}
		\be_k = \int_\cT \Gamma(\bTi, t) \phi_k(t) dt - \bSigma_i \bW \bphi_{ik}.
	\end{align*}
	From \eqref{eq:phi3terms} and the triangle inequality, we have
	\begin{align}
		\lVert \tmu_{iK}-\sumkinf \xi_{ik}\phi_k\rVert_{L^2}
		&= \lVert \sumkK \lambda_k \bphi_{ik}^T\bSigma_i^{-1}(\bXi-\bmu_i)\phi_k  -\sumkinf \xi_{ik}\phi_k\rVert_{L^2} \nonumber\\
		&\le
		\lVert \sumkK  \bphi_{ik}^T\bW(\bYi- \bmui)  \phi_k  -\sumkinf \xi_{ik}\phi_k\rVert_{L^2}
		+
		\lVert \sumkK  \bphi_{ik}^T\bW\bepsi \phi_k \rVert_{L^2}\nonumber\\
		&+
		\lVert \sumkK  \be_k^T\bSigma_i^{-1}(\bXi-\bmui) \phi_k \rVert_{L^2}\nonumber\\
		&= \lVert  A\rVert_{L^2}+ \lVert  B \rVert_{L^2}+ \lVert  C\rVert_{L^2} \label{eq:shrinkageFuncSubject_decomposition},
	\end{align}
	where the functions $A=A(t)$, $B=B(t)$ and $C=C(t)$ are defined through the last equation. By Fubini's theorem and orthogonality of the $\phi_k$, we have
	\begin{align*}
		&E(\lVert  B \rVert_{L^2}^2)\\
		&=
		\int_\cT E\left[\left( \sumkK \phi_k(t) \bphiik^T \bW \bepsi\right)^2 \right]dt
		= \sumkK E[(\bphiik^T \bW \bepsi)^2 ]=\sumkK  \sigma^2 E(\norm{\bW\bphiik}_2^2),
	\end{align*}
	where the last equality follows from the proof of Theorem \ref{thm:xiEst}. Thus, from \eqref{eq:Wphi} and Lemma \ref{lem:unifGaps} we obtain
	\begin{align*}
		&E(\lVert  B \rVert_{L^2}^2)\\
		&\le 
		\sumkK  \sigma^2 m\inv \lVert \phi_k \rVert_\infty^2=O\left(  m\inv \sumkK  \lambda_k^{-2} \right) 
		=O\left(  m\inv \Big[\sumkK  \lambda_k^{-1}\Big]^2 \right) =O(m^{1-2\delta}),
	\end{align*}
	where the first equality is due to $\lVert \phi_k \rVert_\infty=O(\lambda_k\inv)$. This follows from the relation 
	\begin{align*}
		\lambda_k \phi_k(t) &= \int_\cT \Gamma(t,s)\phi_k(s)ds \le \lVert \Gamma(t,\cdot)\rVert_{L^2} <\infty
	\end{align*}
	uniformly over $t$, which is a consequence of the Cauchy--Schwarz inequality and continuity of $\Gamma$ over the compact set $\cT^2$. Therefore
	\begin{align}
		\lVert  B \rVert_{L^2}&=O_p(m^{1/2-\delta}) \label{eq:shrinkageFuncSubject_B}.
	\end{align}
	Observe
	\begin{align*}
		A(t)&=
		\sumkK  (\bphi_{ik}^T\bW(\bYi- \bmui) -\xi_{ik})\phi_k(t)  - \sumkKinf \xiik \phi_k(t)
		=A_1(t) - A_2(t),
	\end{align*}
	where $A_1(t)$ and $A_2(t)$ are defined through the last equation. By Fubini's theorem along with the orthonormality of the $\phi_k$, we have
	\begin{align*}
		E(\lVert A_2 \rVert_{L^2}^2)&= \sumkKinf \lambda_k,
	\end{align*}
	and then
	\begin{align}
		\lVert A_2 \rVert_{L^2}&=O_p\left( \left(\sumkKinf \lambda_k\right)^{1/2} \right)\label{eq:shrinkageFuncSubject_A2}.
	\end{align}
	Define $g_k(t)=\phi_k(t) (X_i(t)-\mu(t))$, $t\in\cT$. By the dominated convergence theorem along with the Cauchy--Schwarz inequality,
	\begin{align*}
		\lambda_k \lvert \phi_k'(t)\rvert &= \Big\lvert \int_\cT \Gamma^{(1,0)} (t,s) \phi_k(s)ds \Big\rvert
		\le 
		\lVert \Gamma^{(1,0)} \rVert_\infty<\infty,
	\end{align*}
	where $\Gamma^{(1,0)}(t,s)=\partial \Gamma(t,s)/\partial t$. This shows that $\lVert \phi_k' \rVert_\infty = O(\lambda_k\inv)$ which combined with the fact that $\lVert \phi_k \rVert_\infty=O(\lambda_k\inv)$ and Condition \ref{a:Diff} leads to $\lVert g_k'\rVert_\infty = O(\lambda_k\inv)$ and $\lVert g_k\rVert_\infty = O(\lambda_k\inv)$. Hence, from the Riemann sum approximation error bound in \eqref{eq:intErrBound} applied to the function $g_k(t)=\phi_k(t) (X_i(t)-\mu(t))$, we obtain
	\begin{align*}
		\lvert \bphi_{ik}^T\bW(\bYi- \bmui) -\xi_{ik}\rvert 
		&\lesssim \lambda_k\inv
		\left( \sumlm w_l^2 + (1-T_{i}^{(m)})^2  + (1-T_{i}^{(m)}) \right).
	\end{align*}
	Therefore
	\begin{align*}
		E(\lVert A_1 \rVert_{L^2})
		&\le
		\sumkK  E(\lvert \bphi_{ik}^T\bW(\bYi- \bmui) -\xi_{ik}\rvert)
		\lesssim
		\sumkK \lambda_k\inv m\inv
		=O(m^{-\delta}),
	\end{align*}
	where we use the condition $\sumkK \lambda_k\inv \asymp m^{1-\delta}$. This shows that $\lVert A_1 \rVert_{L^2}=O_p(m^{-\delta})$, which combined with \eqref{eq:shrinkageFuncSubject_A2} leads to
	\begin{align}
		\lVert A \rVert_{L^2}&=O_p\left(m^{-\delta} + \left(\sumkKinf \lambda_k\right)^{1/2}  \right)\label{eq:shrinkageFuncSubject_A}.
	\end{align}
	From \eqref{eq:ek}, \eqref{eq:Wphi}, the Riemann sum approximation error bound \eqref{eq:intErrBound}, and using that $\lVert \phi_k' \rVert_\infty = O(\lambda_k\inv)$ along with $\lVert \phi_k \rVert_\infty = O(\lambda_k\inv)$, we obtain
	\begin{align}
		\lVert \be_k\rVert_2 
		&\lesssim \sqrt{m} \lambda_k\inv
		\left( \sumlm w_l^2 + (1-T_{i}^{(m)})^2  + (1-T_{i}^{(m)}) \right)
		+ \lambda_k\inv \left(\sumlm w_l^2 \right)^{1/2}\label{eq:shrinkageFuncSubject_bound_ek},
	\end{align}
	Thus, using the inequality $(x_0+x_1)^2\le 2 x_0^2+2x_1^2$, which is valid for all $x_0,x_1\in\bbR$, along with Lemma \ref{lem:unifGaps} leads to
	\begin{align}
		E(\lVert \be_k\rVert_2^2 )
		&\lesssim
		E \left( m \lambda_k^{-2} \left(\left(\sumlm w_l^2\right)^2+ (1-T_{i}^{(m)})^4 + (1-T_{i}^{(m)})^2 \right)+ \lambda_k^{-2}\sumlm w_l^2\right)\nonumber\\
		&=O(m\inv \lambda_k^{-2})\label{eq:shrinkageFuncSubject_expectation_ek}. 
	\end{align}
	Therefore
	\begin{align*}
		E(\lVert  C\rVert_{L^2})
		\le
		\sumkK  E(\lvert \be_k^T\bSigma_i^{-1}(\bXi-\bmui) \rvert)
		&\le
		\sumkK  (E\{E[(\be_k^T\bSigma_i^{-1}(\bXi-\bmui))^2\vert \bTi]\})^{1/2}\\
		&\le
		\sigma\inv \sumkK  (E(\lVert \be_k\rVert_2^2)  )^{1/2}
		\lesssim
		m^{1/2-\delta},
	\end{align*}
	where last inequality uses that $\sumkK \lambda_k\inv \asymp m^{1-\delta}$. Hence
	\begin{align}
		\lVert  C\rVert_{L^2} &=O_p(m^{1/2-\delta} )\label{eq:shrinkageFuncSubject_C}.
	\end{align}
	Combining \eqref{eq:shrinkageFuncSubject_decomposition}, \eqref{eq:shrinkageFuncSubject_B}, \eqref{eq:shrinkageFuncSubject_A}, and \eqref{eq:shrinkageFuncSubject_C} leads to
	\begin{align}
		\lVert \tmu_{iK}-\sumkinf \xi_{ik}\phi_k\rVert_{L^2}
		&=
		O_p\left( m^{1/2-\delta} + \left(\sumkKinf \lambda_k\right)^{1/2} \right)\label{eq:shrinkageFuncSubject_keyDiff1}.
	\end{align}
	%Concerning the shrinkage of $\int_\cT \Gamma_{iK}(t,t) dt$, 
	By  orthonormality of the $\phi_k$ and since $\bSigmaiK = \bLambda_K - \bLambda_K \bPhiiK^T \bSigma_i\inv  \bPhiiK\bLambda_K$,
	\begin{align}
		\int_\cT \Gamma_{iK}(t,t) dt &= \text{trace}(\bSigma_{iK})
		= \sum_{k=1}^K  \left( \lambda_k-\lambda_k \bphi_{ik}^T \bSigma_i\inv  \lambda_k \bphi_{ik} \right)\label{eq:shrinkageFuncSubject_varianceDecomp}.
	\end{align}
	From \eqref{eq:shrinkageFuncSubject_expectation_ek} and using the condition $\sumkK \lambda_k\inv \asymp m^{1-\delta}$, we obtain $\sumkK \lambda_k^{-2}=O(m^{2-2\delta})$ and
	\begin{align*}
		E\left(\sumkK  \be_k^T \bSigma_i\inv  \be_k \right)&\le \sigma^{-2} \sumkK E(\lVert \be_k\rVert_2^2)=O(m^{1-2\delta}).
	\end{align*}
	Thus
	\begin{align}
		\sumkK  \be_k^T \bSigma_i\inv  \be_k &=O_p(m^{1-2\delta})\label{eq:shrinkageFuncSubject_bound1}.
	\end{align}
	Since $\lVert \phi_k \rVert_\infty=O(\lambda_k\inv)$ and $\sumkK \lambda_k^{-2}=O(m^{2-2\delta})$, 
	\begin{align*}
		&\sumkK \lambda_k\inv  \lVert \be_k\rVert_2 
		\lesssim
		m^{5/2-2\delta} \left( \sumlm w_l^2 + (1-T_{i}^{(m)})^2   + (1-T_{i}^{(m)}) \right)   \\ &\quad		+
		m^{2-2\delta} \left(\sumlm w_l^2 \right)^{1/2}
		=O_p\left( m^{3/2-2\delta} \right),
	\end{align*}
	where the first inequality is due to \eqref{eq:shrinkageFuncSubject_bound_ek} and the last  equality is due to Lemma \ref{lem:unifGaps}. Thus
	\begin{align}
		\sumkK \lvert \be_k^T \bW \bphiik \rvert
		\le 
		\sumkK \lVert \be_k\rVert_2 \lVert \bW \bphiik \rVert_2
		&\le
		\left(\sumlm w_l^2\right)^{1/2} \sumkK \lVert \be_k\rVert_2   \lVert \phi_k \rVert_\infty %\nonumber
		=
		O_p\left( m^{1-2\delta} \right)\label{eq:shrinkageFuncSubject_bound2},
	\end{align}
	where the second inequality is due to \eqref{eq:Wphi}. Also,
	\begin{align}
		\sumkK \sigma^2 \lvert \bphiik^T \bW \bW \bphiik \rvert
		&\le
		\sigma^2 \sumkK \lVert  \bW \bphiik \rVert_2^2
		\le
		\sigma^2  \sumkK \lVert \phi_k \rVert_\infty^2 \left(\sumlm w_l^2 \right)
		=O_p\left( m^{1-2\delta}\right)\label{eq:shrinkageFuncSubject_bound3}.
	\end{align}
	From the Riemann sum approximation error bound \eqref{eq:intErrBound} applied to the function  $g_k(t)=\lambda_k \phi_k^2(t)$, 
	and using that $\lVert g_k \rVert_\infty=O(\lambda_k\inv)$ and $\lVert g_k' \rVert_\infty=O(\lambda_k\inv)$, we have
	\begin{align*}
		\lvert \lambda_k \bphiik^T \bW \bphiik  - \lambda_k \rvert 
		&= O\left( \lambda_k\inv \left( \sumlm w_l^2 + (1-T_{i}^{(m)})^2  + (1-T_{i}^{(m)}) \right)\right).
	\end{align*}
	Thus
	\begin{align*}
		E\left( \sumkK \lvert \lambda_k \bphiik^T \bW \bphiik  - \lambda_k \rvert \right)
		&=O(m^{-\delta}),		
	\end{align*}
	which implies
	\begin{align}
		\sumkK \lvert \lambda_k \bphiik^T \bW \bphiik  - \lambda_k\rvert &=O_p\left( m^{-\delta} \right)\label{eq:shrinkageFuncSubject_bound4}.
	\end{align}
	Also, from \eqref{eq:intErrBound} and \eqref{eq:Wphi} we have
	\begin{align*}
		&\sumkK \lvert \bphiik^T \bW \left( \Gamma(\bTi,\bTi^T) \bW \bphi_{ik} - \lambda_k \bphiik \right)\rvert \\
		&\le
		\sumkK \lVert \bphiik^T \bW \rVert_2  \lVert  \Gamma(\bTi,\bTi^T) \bW \bphi_{ik} - \lambda_k \bphiik \rVert_2 \\
		&\lesssim
		\sumkK \lambda_k^{-2}  m^{1/2}  \left(\sumlm w_l^2 \right)^{1/2}
		\left( \sumlm w_l^2 + (1-T_{i}^{(m)})^2  + (1-T_{i}^{(m)}) \right),
	\end{align*}
	which along with Lemma \ref{lem:unifGaps} leads to
	\begin{align*}
		E\left( \sumkK \lvert\bphiik^T \bW \left( \Gamma(\bTi,\bTi^T) \bW \bphi_{ik} - \lambda_k \bphiik \right)\rvert \right)
		&=O(m^{1-2\delta}).
	\end{align*}
	This shows that
	\begin{align}
		\sumkK [\bphiik^T \bW \left( \Gamma(\bTi,\bTi^T) \bW \bphi_{ik} - \lambda_k \bphiik \right)]
		&=O_p(m^{1-2\delta})\label{eq:shrinkageFuncSubject_bound5}.
	\end{align}
	From \eqref{eq:shrinkageFuncSubject_bound1}, \eqref{eq:shrinkageFuncSubject_bound2}, \eqref{eq:shrinkageFuncSubject_bound3}, 
	\eqref{eq:shrinkageFuncSubject_bound4}, \eqref{eq:shrinkageFuncSubject_bound5}, and observing
	\begin{align*}
		\bphi_{ik}^T \bW \bSigma_i \bW \bphi_{ik}=\sigma^2 \bphi_{ik}^T \bW \bW \bphi_{ik}+\bphi_{ik}^T \bW  \Gamma(\bTi,\bTi^T) \bW \bphi_{ik},
	\end{align*}
	leads to
	\begin{align*}
		&\Big\lvert \sumkK (\lambda_k-\lambda_k \bphi_{ik}^T \bSigma_i\inv  \lambda_k \bphi_{ik}) \Big\rvert\\
		&= 
		\Big\lvert \sumkK (\lambda_k - \be_k^T \bSigma_i\inv  \be_k - 2 \be_k^T \bW \bphi_{ik} - \bphi_{ik}^T \bW \bSigma_i \bW \bphi_{ik}) \Big\rvert\\
		&\le
		\sumkK  \be_k^T \bSigma_i\inv  \be_k + 2 \sumkK \lvert \be_k^T \bW \bphiik \rvert
		+\sigma^2 \sumkK \lvert  \bphi_{ik}^T \bW \bW \bphi_{ik} \rvert \\
		&+ \Big\lvert \sumkK [\bphiik^T \bW \left( \Gamma(\bTi,\bTi^T) \bW \bphi_{ik} - \lambda_k \bphiik \right)] \Big\rvert
		+ \sumkK \lvert \lambda_k \bphiik^T \bW \bphiik  - \lambda_k \rvert\\
		&=O_p\left( m^{1-2\delta}\right),
	\end{align*}
	where the first equality uses \eqref{eq:phiSigmaInv}. This along with \eqref{eq:shrinkageFuncSubject_varianceDecomp} implies
	\begin{align}
		\int_\cT \Gamma_{iK}(t,t) dt 
		&= 
		O_p\left( m^{1-2\delta}\right)\label{eq:shrinkFunc_var}.
	\end{align}
	Combining \eqref{eq:shrinkFunc_var} with \eqref{eq:shrinkageFuncSubject_keyDiff1} leads to the result.
\end{proof}

\section{Auxiliary Results and Proofs} \label{s:AuxiliaryResults}
We provide the proofs of Propositions~\ref{thm:xiEst2} and Theorems  \ref{thm:xiEstHat2}, 
\ref{thm:functionalShrinkage_estimated}--\ref{thm:predictiveShrinkage_FLM} in the main text, 
followed by a sequence of auxiliary lemmas and their proofs. These auxiliary results  are used to derive the main results.\\ % in section \ref{S:preliminariesGaussian}. 

\begin{proof}[Proof of Proposition \ref{thm:xiEst2}]\vspace{-.2cm} 

%\noindent{\bf Proposition \ref{thm:xiEst2}} \vspace{-.2cm} 
%\begin{proof}
	Recalling that $\bSigmaiK = \bLambda_K - \bLambda_K \bPhiiK^T \bSigma_i\inv  \bPhiiK\bLambda_K$ we have
	\begin{equation}
		\norm{\bSigmaiK}_{\text{op},2}\leq \text{trace}(\bSigmaiK)
		=
		\sum_{k=1}^K  \left( \lambda_k-\lambda_k \bphi_{ik}^T \bSigma_i\inv  \lambda_k \bphi_{ik} \right)\label{eq:maxnorm}.
	\end{equation}
	Moreover, since $\lambda_k \bphi_{ik}=\be_k+\bSigma_i \bW \bphi_{ik}$, where $\be_k$ is defined as in the proof of Proposition~\ref{thm:xiEst}, it follows that
	\begin{align}
		\lambda_k \bphi_{ik}^T \bSigma_i\inv  \lambda_k \bphi_{ik}= \be_k^T \bSigma_i\inv  \be_k+2 \be_k^T \bW \bphi_{ik} + \bphi_{ik}^T \bW \bSigma_i \bW \bphi_{ik}
		\label{eq:expandNorm}.
	\end{align}
	%Next, from the proof of Theorem \ref{thm:xiEst}, 
	From, \eqref{eq:ek},  
	\begin{align*}
		\normtwo{\lambda_k\bphi_{ik} -\Gamma(\bTi,\bTi^T) \bW \bphi_{ik}}=O_p(m^{-1/2}),
	\end{align*}
	and using \eqref{eq:Wphi}, 
	\begin{gather*}
		\bphi_{ik}^T \bW \bSigma_i \bW \bphi_{ik}=\sigma^2 \bphi_{ik}^T \bW \bW \bphi_{ik}+\bphi_{ik}^T \bW  \Gamma(\bTi,\bTi^T) \bW \bphi_{ik}\\
		=
		O_p(m\inv)+ \bphi_{ik}^T \bW \left(\lambda_k\bphi_{ik} -O_p(m^{-1/2})  \right)
		=\lambda_k \bphi_{ik}^T \bW \bphi_{ik} +O_p(m\inv),
	\end{gather*}
	where $\lambda_k \bphi_{ik}^T \bW \bphi_{ik}= \lambda_k +O_p(m\inv)$. This follows from the quadrature approximation error \eqref{eq:intErr}, observing  $\int_0^1 \phi_k^2(t)dt=1$,  and implies 
	\begin{align}
		\bphi_{ik}^T \bW \bSigma_i \bW \bphi_{ik}=\lambda_k +O_p(m\inv).
		\label{eq:expandNorm3}
	\end{align}
	The result then follows by combining \eqref{eq:maxnorm}, \eqref{eq:expandNorm}, \eqref{eq:expandNorm3}, \eqref{eq:ek}, \eqref{eq:Wphi}, and the fact that $\norm{\bSigma_i\inv}_{\text{op},2}\le \sigma^{-2}$.
\end{proof}\vspace{1cm}

%\begin{proof}[Proof of Proposition \ref{thm:xiEstHat2}]  %\vspace{-.2cm} 
%\noindent{\bf Theorem \ref{thm:xiEstHat2}} \vspace{-.2cm} 
\begin{proof}[Proof of Theorem \ref{thm:xiEstHat2}]
	Recall that $\hat{\bmu}^*=\hat{\mu}(\bT^*)$, $\bT^*=(T_1^*,\dots,T_{m^*}^*)^T$, the estimated FPCs $\hxi_{k}^*=\hlambda_k \hphi_{k}(\bT^*)^T \hbSigma^{*-1}(\bX^*-\hat{\bmu}^*)$, $\hbPhi_{K}^{*}$ is analogous to $\hbPhi_{iK}$ while replacing the $T_{ij}$ with $T_{j}^*$, and similarly for quantities such as $\bPhi_{K}^{*}$, $\hbSigma^{*-1}$, and $\bSigma^{*-1}$. Note that
	\begin{align}
		\bSigma_{K}^{*} -\hbSigma_{K}^{*}&= \bLambda_K-\hbLambda_K  +  \hbLambda_K \hbPhi_{K}^{*T} \hbSigma^{*-1} \hbPhi_{K}^{*}\hbLambda_K- \bLambda_K \bPhi_{K}^{*T} \bSigma^{*-1} \bPhi_{K}^{*}\bLambda_K\label{eq:thm2_1},
	\end{align}
	where $ \lVert \bLambda_K-\hbLambda_K  \rVert_{\text{op},2} =O_p(a_n+b_n)$ follows from Theorem $5.2$ in \cite{zhan:16} along with perturbation results \citep{bosq:00} and the fact that  $\lVert \bLambda_K-\hbLambda_K  \rVert_{\text{op},2}\le \sqrt{K} \max_{1\le k\le K}\lvert \lambda_k-\hlambda_k\rvert$. Since $\hlambda_k\hbphi_{k}^*= \int_\cT \hGamma(\bT^*, t) \hphi_k(t) dt $ and writing  $\hbe_k^{*}=\int_\cT \hGamma(\bT^*, t) \hphi_k(t) dt - \hbSigma^* \bW^* \hbphi_{k}^*$, we have that  the $(j,l)$ entry of $\hbLambda_K \hbPhi_{K}^{*T} \hbSigma^{*-1} \hbPhi_{K}^{*}\hbLambda_K$ is given by
	\begin{align}
		[\hbLambda_K \hbPhi_{K}^{*T} \hbSigma^{*-1} \hbPhi_{K}^{*}\hbLambda_K]_{j,l}&= (\hbe_j^{*T}\hbSigma^{*-1}+ \hbphi_{j}^{*T} \bW^{*} ) (\hbe_l^{*}+\hbSigma^{*} \bW^{*} \hbphi_{l}^{*})\nonumber \\
		&= \hbe_j^{*T}\hbSigma^{*-1} \hbe_l^{*}+ \hbe_j^{*T} \bW^{*} \hbphi_{l}^{*}+ \hbphi_{j}^{*T} \bW^{*} \hbe_l^{*}+ \hbphi_{j}^{*T} \bW^{*} \hbSigma^{*} \bW^{*} \hbphi_{l}^{*}\label{eq:thm2_2},
	\end{align}
	where $1\le j,l \le K$. Denote by $\hGamma(\bT^*,\bT^{*T})$ the matrix whose $(i,j)$ element is  $\hGamma(T_i^*,T_j^*)$, $1\le i,j\le m^*$, and similarly define $\Gamma(\bT^*,\bT^{*T})$. Also note that $\hbSigma^{*}=\hsigma^2 I_{m^*}+\hGamma(\bT^*,\bT^{*T})$, where $I_{m^*}\in\bbR^{m^*\times m^*}$ is the identity matrix. 
	From \eqref{eq:shrinkageFuncSubject_expectation_ek}, \eqref{eq:shrinkEstFunc_DeltaGamma}, \eqref{eq:secondMomentFPCscoresNorm_WDeltaPhiBound}, \eqref{eq:traceSigmaK_Deltaj}, Lemma \ref{lem:unifGaps}, and using that $ \lVert \hbSigma^{*-1} -\bSigma^{*-1} \rVert_{\text{op},2}=O_p(m^*(a_n+b_n))$  along with the condition $m^*(a_n+b_n)=o(1)$ as $\ntoinf$, it follows that
	$\lVert \hGamma(\bT^*,\bT^{*T})-\Gamma(\bT^*,\bT^{*T})  \rVert_2=O_p(m^*(a_n+b_n))$, $\lVert  \bW^{*}(\hbphi_{p}^{*}-\bphi_{p}^{*}) \rVert_2=O_p(m^{*-1/2}(a_n+b_n))$, $p=j,l$, $\lVert \Gamma(\bT^*,\bT^{*T})  \rVert_{\text{op},2}=O(m^*)$, $\lVert \bSigma^{*}\rVert_{\text{op},2}=O(m^*)$, $\lVert \hbSigma^{*}\rVert_{\text{op},2}=O_p(m^*)$, $\lVert \bW^{*}\bphi_{p}^{*} \rVert_2=O_p(m^{*-1/2})$, $p=j,l$, $ \lVert \hbSigma^{*} -\bSigma^{*} \rVert_{\text{op},2}=O_p(m^*(a_n+b_n))$, $\lVert \bW^{*} \rVert_2=O_p(m^{*-1/2})$, $\lVert \be_p^{*}\rVert_2=O_p(m^{*-1/2})$ and $\lVert \hbe_p^{*}-\be_p^{*} \rVert_2=O_p(m^{*1/2}(a_n+b_n))$, $p=j,l$. These bounds imply
	\begin{align*}
		\hbphi_{j}^{*T} \bW^{*} \hbSigma^{*} \bW^{*} \hbphi_{l}^{*}-
		\bphi_{j}^{*T} \bW^{*} \bSigma^{*} \bW^{*} \bphi_{l}^{*}=O_p(a_n+b_n), \\
		\hbe_j^{*T}\hbSigma^{*-1} \hbe_l^{*}-\be_j^{*T}\bSigma^{*-1} \be_l^{*}=O_p(a_n+b_n),\\
		\hbe_j^{*T} \bW^{*} \hbphi_{l}^{*}- \be_j^{*T} \bW^{*} \bphi_{l}^{*}=O_p(a_n+b_n),\\
		\hbphi_{j}^{*T} \bW^{*} \hbe_l^{*}- \bphi_{j}^{*T} \bW^{*} \be_l^{*}=O_p(a_n+b_n),
	\end{align*}
	which combined with \eqref{eq:thm2_2} leads to
	\begin{align*}
		[\hbLambda_K \hbPhi_{K}^{*T} \hbSigma^{*-1} \hbPhi_{K}^{*}\hbLambda_K]_{j,l}-
		[\bLambda_K \bPhi_{K}^{*T} \bSigma^{*-1} \bPhi_{K}^{*}\bLambda_K]_{j,l} =O_p(a_n+b_n).
	\end{align*}
	Hence $\lVert \hbLambda_K \hbPhi_{K}^{*T} \hbSigma^{*-1} \hbPhi_{K}^{*}\hbLambda_K- \bLambda_K \bPhi_{K}^{*T} \bSigma^{*-1} \bPhi_{K}^{*}\bLambda_K \rVert_F=O_p(a_n+b_n)$ and the result follows from \eqref{eq:thm2_1}.
\end{proof}

\vspace{1cm}

%\noindent{\bf Theorem \ref{thm:functionalShrinkage_estimated}} \vspace{-.2cm} 

\begin{proof}[Proof of Theorem \ref{thm:functionalShrinkage_estimated}]
	Let $\upsilon_K=\sumkK \lambda_k^{-1/2}\delta_k\inv$ and $\nu_K=\sumkK \lambda_k\inv$. Note that
	\begin{align}
		\lVert \hmu_{K}^*-\tmu_{K}^* \rVert_{L^2}
		&= 
		\lVert \hbxiKsT \hbPhi_{K} - \tbxiKsT \bPhi_{K}\rVert_{L^2}\nonumber\\
		&\le
		\lVert (\hbxiKs-\tbxiKs)^T (\hbPhi_{K} -\bPhi_{K})  \rVert_{L^2}+
		\lVert (\hbxiKs-\tbxiKs)^T \bPhi_{K}  \rVert_{L^2}\nonumber\\
		&+
		\lVert \tbxiKsT (\hbPhi_{K} -\bPhi_{K})  \rVert_{L^2}\label{eq:shrinkEstFunc_meanDecomp}.
	\end{align}
	Now, by the Cauchy--Schwarz inequality,
	\begin{align}
		\lVert (\hbxiKs-\tbxiKs)^T (\hbPhi_{K} -\bPhi_{K})  \rVert_{L^2}
		&\le
		\lVert \hbxiKs-\tbxiKs \rVert_2 
		\sumkK \lVert \hphi_k - \phi_k \rVert_{L^2}\nonumber\\
		&\lesssim 
		\left(\sumkK \delta_k\inv\right) \lVert \hbxiKs-\tbxiKs \rVert_2 \lVert \hXi -\Xi \rVert_{\text{op}} \label{eq:shrinkEstFunc_meanDecomp1},
	\end{align}
	and by orthonormality of the $\phi_k$,
	\begin{align}
		\lVert (\hbxiKs-\tbxiKs)^T \bPhi_{K} \rVert_{L^2}
		&\le \lVert \hbxiKs-\tbxiKs \rVert_2 \label{eq:shrinkEstFunc_meanDecomp2}.
	\end{align}
	Also note that
	\begin{align*}
		E(\lVert \tbxiKs\rVert_2^2)
		=
		\text{trace}(E[E( \tbxiKs  \tbxiKsT \vert \bT^*)]) 
		&=
		E(\text{trace}(\bLambda_K \bPhi_K^{*T} \bSigma^{*-1} \bPhi_K^* \bLambda_K))\\
		&= E\left( \sumkK \lambda_k^2 \bphi_{k}^{*T} \bSigma^{*-1} \bphi_{k}^*\right),
	\end{align*}
	and
	\begin{align*}
		\lambda_j^2 \bphi_{j}^{*T} \bSigma^{*-1} \bphi_{j}^*
		&=
		\bejsT \bSigma^{*-1} \bejs + 2 \bejsT \bW^* \bphi_{j}^* + \bphi_{j}^{*T}\bW^* \bSigma^* \bW^*\bphi_{j}^*,
	\end{align*}
	where $j=1,\dots,K$. Similar arguments as the ones outlined in the proof of Theorem \ref{thm:functionalShrinkage} then show that for large enough $n$
	\begin{align*}
		E(\lVert \tbxiKs\rVert_2^2)
		&=
		E\left( \sumkK \lambda_k^2 \bphi_{k}^{*T} \bSigma^{*-1} \bphi_{k}^*\right)
		\lesssim
		m^{*(1-2\delta)}+m^{* -\delta}+\sumkK \lambda_k
		\lesssim
		m^{*(1-2\delta)}+\sumkK \lambda_k.
	\end{align*}
	Since $\delta\in(1/2,1)$ and $\sumkinf \lambda_k <\infty$, this implies
	\begin{align}
		\lVert \tbxiKs\rVert_2&= O_p(1) \label{eq:shrinkEstFunc_estFPC}.
	\end{align}
	Observing
	\begin{align*}
		\lVert  \tbxiKsT (\hbPhi_{K} -\bPhi_{K})   \rVert_{L^2}
		&\le
		\lVert  \tbxiKs \rVert_2 \sumkK \lVert \hphi_k-\phi_k \rVert_{L^2}
		\lesssim	
		\left(\sumkK \delta_k\inv\right)  \lVert \hXi -\Xi \rVert_{\text{op}}  \lVert \tbxiKs \rVert_2,
	\end{align*}
	and using \eqref{eq:shrinkEstFunc_opNorm} along with \eqref{eq:shrinkEstFunc_estFPC} leads to
	\begin{align}
		\lVert  \tbxiKsT (\hbPhi_{K} -\bPhi_{K})   \rVert_{L^2}
		&=
		O_p\left( (a_n+b_n) \sumkK \delta_k\inv \right) \label{eq:shrinkEstFunc_meanDecomp3}.
	\end{align}
	In view of \eqref{eq:shrinkEstFunc_meanDecomp}, \eqref{eq:shrinkEstFunc_meanDecomp1}, \eqref{eq:shrinkEstFunc_meanDecomp2}, \eqref{eq:shrinkEstFunc_meanDecomp3}, the condition $\upsilon_K (a_n+b_n)=o(1)$ which implies $(a_n+b_n) \sumkK \delta_k\inv =o(1)$ as $\ntoinf$, and employing Lemma \ref{lem:secondMomentFPCscoresNorm} leads to
	\begin{align}
		&\lVert \hmu_{K}^*-\tmu_{K}^* \rVert_{L^2}\nonumber\\
		&= 
		O_p\Big(
		(a_n+b_n) \left(\sumkK \delta_k\inv \right)
		+
		m^{*1/2} (a_n+b_n) \left(\sumkK \delta_k^{-2}\lambda_k^{-2}\right)^{1/2} 
		+ 
		m^{*-1/2} \left(\sumkK \lambda_k^{-2}\right)^{1/2}\nonumber\\
		&\quad\quad
		+
		m^{*}(a_n+b_n) \left(\sumkK \lambda_k^{-2}\right)^{1/2}
		+ m^{*2}(a_n+b_n)^2 \left(\sumkK \delta_k^{-2} \lambda_k^{-2}\right)^{1/2}
		\Big) \label{eq:shrinkEstFunc_meanRate}.
	\end{align}
	Observe 
	\begin{align*}
		\cW_2^2(\hcGKs,\cA_{X^{*c}})
		&\le 
		E(\lVert g_1 - g_2 \rVert_{L^2}^2 \mid (\bX_j,\bT_j)_{j=0}^n),
	\end{align*}
	where $\bX_0:=\bX^*$ and $\bT_0:=\bT^*$, the random element $g_1\in L^2$ has conditional distribution $g_1\sim \hcGKs$ given $(\bX_j,\bT_j)_{j=0}^n$, and $g_2(\cdot)=X^{*c}(\cdot)$ almost surely. Since $E(g_1 \mid (\bX_j,\bT_j)_{j=0}^n)=\hmu_{K}^*$ and $\Var(g_1(t)\mid (\bX_j,\bT_j)_{j=0}^n)=\hGamma_K^*(t,t)$, $t\in\cT$, we obtain
	\begin{align*}
		\cW_2^2(\hcGKs,\cA_{X^{*c}})
		&\le
		E(\lVert g_1 - \hmu_{K}^* \rVert_{L^2}^2 \mid (\bX_j,\bT_j)_{j=0}^n) 
		+  \lVert \hmu_{K}^* - X^{*c} \rVert_{L^2}^2\\
		&=
		\int_\cT \hGamma_K^*(t,t) dt + \lVert \hmu_{K}^* - X^{*c} \rVert_{L^2}^2\\
		&\le
		\int_\cT (\hGamma_K^*(t,t)- \Gamma_{K}^*(t,t)) dt+
		\lVert \hmu_{K}^* - X^{*c} \rVert_{L^2}^2+ O_p(m^{*(1-2\delta)}),
	\end{align*}
	where the equality follows from Fubini's Theorem and the last inequality is due to $\int_\cT \Gamma_{K}^*(t,t) dt=O_p(m^{1-2\delta})$, which follows analogously as in \eqref{eq:shrinkFunc_var}. Combining \eqref{eq:shrinkEstFunc_meanRate} and arguments analogous to those in the proof of Theorem \ref{thm:functionalShrinkage} lead to 
	\begin{align*}
		&\lVert \hmu_{K}^* - X^{*c} \rVert_{L^2}\\
		&\le
		\lVert \tmu_{K}^* - X^{*c} \rVert_{L^2}
		+
		\lVert \hmu_{K}^* - \tmu_{K}^*\rVert_{L^2}\\
		&=
		O_p\Big[ m^{*(1/2-\delta)} + \left(\sumkKinf \lambda_k\right)^{1/2}
		+ (a_n+b_n) \left(\sumkK \delta_k\inv \right)\\
		&\quad\quad +
		m^{*1/2} (a_n+b_n) \left(\sumkK \delta_k^{-2}\lambda_k^{-2}\right)^{1/2}
		+ 
		m^{*-1/2} \left(\sumkK \lambda_k^{-2}\right)^{1/2}\\
		&\quad\quad 
		+
		m^{*}(a_n+b_n) \left(\sumkK \lambda_k^{-2}\right)^{1/2}
		+ m^{*2}(a_n+b_n)^2 \left(\sumkK \delta_k^{-2} \lambda_k^{-2}\right)^{1/2}
		\Big].
	\end{align*}
	From Lemma \ref{lem:traceSigmaK} we have
	\begin{align*}
		\int_\cT (\hGamma_K^*(t,t)- \Gamma_{K}^*(t,t)) dt 
		&=
		\text{trace}(\hbSigma_{K}^*-\bSigma_{K}^*)\\
		&=
		O_p\left( m^{*} (a_n+b_n)^2 \sumkK \lambda_k^{-2} \delta_k^{-2}
		+
		(a_n+b_n) \sumkK \lambda_k^{-2} \delta_k\inv
		\right).
	\end{align*}
	Therefore
	\begin{align*}
		&\cW_2^2(\hcGKs,\cA_{X^{*c}})\\
		&=
		O_p\Big[ m^{*(1-2 \delta)} + \sumkKinf \lambda_k
		+(a_n+b_n)^2 \left(\sumkK \delta_k\inv\right)^2
		+
		m^{*} (a_n+b_n)^2  \sumkK \delta_k^{-2}\lambda_k^{-2} \\
		&\quad\quad 
		+ m^{*-1} \sumkK \lambda_k^{-2}
		+
		m^{*2}(a_n+b_n)^2 \sumkK \lambda_k^{-2}
		+ 
		m^{*4}(a_n+b_n)^4 \sumkK \delta_k^{-2} \lambda_k^{-2}\\
		&\quad\quad+
		(a_n+b_n) \sumkK \lambda_k^{-2} \delta_k\inv
		\Big],
	\end{align*}
	and the result follows.
\end{proof}

\vspace{1cm}

%\noindent{\bf [Theorem \ref{thm:predictiveConsistency}]

%\vspace{-.2cm} 

\begin{proof}[Proof of Theorem \ref{thm:predictiveConsistency}]	We use the fact that  for a normal random variable $Z_1\sim N(\kappa_1,\kappa_2^2)$ and $t\in(0,1)$ it holds that $Q_1(t)=\kappa_2 q(t)+\kappa_1$, where $Q_1(\cdot)$ and $q(\cdot)$ are the quantile functions corresponding to $Z_1$ and a standard normal random variate, respectively. Note that since  $\lvert \lambda_{\min}(\hbSigmaiK) -\lambda_{\min}(\bSigmaiK)  \rvert \le \norm{\hbSigmaiK-\bSigmaiK}_{\text{op},2}=o_p(1)$, where the $o_p(1)$ term is uniform in $i$ (see the proof of Lemma \ref{lem:eigenFuncK}), and $\lambda_{\min}(\bSigmaiK)\ge \kappa_0$ a.s., we have
	\begin{align*}
		P\left(\norm{\hbSigmaiK-\bSigmaiK}_{\text{op},2}\le \kappa_0/2\right)&=P\left(\kappa_0-\norm{\hbSigmaiK-\bSigmaiK}_{\text{op},2}\ge \kappa_0/2\right)\\
		&\le P\left(\lambda_{\min}(\bSigmaiK)-\norm{\hbSigmaiK-\bSigmaiK}_{\text{op},2}\ge \kappa_0/2\right)\\
		&\le P\left(\lambda_{\min}(\hbSigmaiK)\ge \kappa_0/2\right),
	\end{align*}
	which implies $\lambda_{\min}(\hbSigmaiK)\ge \kappa_0/2$ with probability tending to $1$. For the remainder of the proof we work on this event. From the closed form expression for the $2$-Wasserstein distance between one-dimensional distributions with finite second moments,
	\begin{align}\label{eq:wassPredictive}
		\cW_2^2(\tcP_{iK},\cP_{iK})&=\int_0^1 \left(   [ (\bbeta_K^T  \hbSigmaiK \bbeta_K)^{1/2}  -(\bbeta_K^T \bSigmaiK \bbeta_K)^{1/2}]  q(t) + \boldsymbol{\beta}_K^T (\hbxiiK-\tbxiiK)  \right)^2  dt \nonumber\\
		&= [ (\bbeta_K^T  \hbSigmaiK \bbeta_K)^{1/2}  -(\bbeta_K^T \bSigmaiK \bbeta_K)^{1/2}]^2 \int_0^1 q^2(t)dt + (\boldsymbol{\beta}_K^T (\hbxiiK-\tbxiiK) )^2\nonumber \\
		&\quad +2 [ (\bbeta_K^T  \hbSigmaiK \bbeta_K)^{1/2}  -(\bbeta_K^T \bSigmaiK \bbeta_K)^{1/2}]\  \bbeta_K^T (\hbxiiK-\tbxiiK)  \int_0^1 q(t)dt\nonumber \\
		&\le  \frac{(\bbeta_K^T (  \hbSigmaiK  -\bSigmaiK )\bbeta_K)^2}{\bbeta_K^T \bSigmaiK \bbeta_K} \int_0^1 q^2(t)dt + (\bbeta_K^T (\hbxiiK-\tbxiiK) )^2,
	\end{align}
	where the last inequality follows from the fact that $\int_0^1 q(t)dt=E(Z)=0$, where $Z\sim N(0,1)$, and using the inequality $(\sqrt{x}-\sqrt{y})^2\le (x-y)^2/y$ which is valid for any scalars $x\ge 0$ and $y>0$. Since $\int_0^1 q^2(t)dt =E(Z^2)<\infty$, it then suffices to control the terms  $\bbeta_K^T (  \hbSigmaiK  -\bSigmaiK )\bbeta_K$ and $(\boldsymbol{\beta}_K^T (\hbxiiK-\tbxiiK) )^2$. From the proof of Lemma \ref{lem:eigenFuncK}, we have $\norm{\bSigma_{iK}-\hat{\bSigma}_{iK}}_F=O(a_n+b_n)$ a.s.  as $n\to\infty$, where the $O(a_n+b_n)$ term is uniform over $i$, and similar arguments as in the proof of Theorem 2 in \cite{dai:16:1} show that $\lvert \hat{\xi}_{ik}-\tilde{\xi}_{ik}\rvert= O(a_n+b_n) \lVert \bXi-\bmui \rVert_2 =O(a_n+b_n) O_p(1)=O_p(a_n+b_n)$, $k=1,\dots,K$. Thus,  
	$(\bbeta_K^T (\hbxiiK-\tbxiiK) )^2\le \normtwo{\bbeta_K}^2 \lVert\hbxiiK-\tbxiiK\rVert_2^2=O_p((a_n+b_n)^2)$ and properties of the operator norm show that $\lvert \bbeta_K^T (  \hbSigmaiK  -\bSigmaiK )\bbeta_K \rvert \le \normtwo{\bbeta_K}^2 \norm{\hbSigmaiK  -\bSigmaiK}_F=O(a_n+b_n)$ a.s. as $n\to\infty$. This along with \eqref{eq:wassPredictive} leads to
	\begin{align}
		\cW_2(\tcP_{iK},\cP_{iK})&=O_p(a_n+b_n).\label{eq:wassPred1}
	\end{align}
	Similar arguments show that
	\begin{align}
		\cW_2^2(\hcP_{iK},\tcP_{iK})&\le \frac{(\hbbeta_K^T \hbSigmaiK \hbbeta_K - \bbeta_K^T \hbSigmaiK \bbeta_K)^2}{\bbeta_K^T \hbSigmaiK \bbeta_K} \int_0^1 q^2(t)dt + ( (\hbbeta_K-\bbeta_K)^T \hbxiiK +\hbeta_0-\beta_0 )^2,\label{eq:wassPred2}
	\end{align}
	and
	\begin{align}
		&\lvert \hbbeta_K^T \hbSigmaiK \hbbeta_K - \bbeta_K^T \hbSigmaiK \bbeta_K\rvert \nonumber \\
		&= 
		\lvert (\hbbeta_K-\bbeta_K)^T \hbSigmaiK \hbbeta_K + \bbeta_K^T \hbSigmaiK (\hbbeta_K-\bbeta_K) \rvert \nonumber \\
		&\le \normtwo{\hbbeta_K-\bbeta_K}^2 \norm{\hbSigmaiK-\bSigmaiK}_{\text{op},2}+\normtwo{\hbbeta_K-\bbeta_K}\norm{\hbSigmaiK-\bSigmaiK}_{\text{op},2}  \normtwo{\bbeta_K}\nonumber  \\
		&\quad + \normtwo{\hbbeta_K-\bbeta_K}^2 \norm{\bSigmaiK}_{\text{op},2} +\normtwo{\hbbeta_K-\bbeta_K}\norm{\bSigmaiK}_{\text{op},2} \normtwo{\bbeta_K}\nonumber  \\
		&=O_p(\alpha_n),\label{eq:wassPred3}
	\end{align}
	where the first inequality follows from properties of the operator norm; the last equality is due to Lemma \ref{lem:Lemma_beta} and the facts that $h\asymp n^{-1/3}$ implies  that the rate $ \tau_M  \Big[\left(\frac{1}{nh}+h^2\right)^{1/2}+a_n\Big]$ is faster than $c_n \upsilon_M $,    $\norm{\hbSigmaiK  -\bSigmaiK}_F=O(a_n+b_n)$ a.s. as $n\to\infty$ and that $\norm{\bSigmaiK}_{\text{op},2}$ is uniformly bounded in $i$ in the sparse case. Since $\lvert \lambda_{\min}(\hbSigmaiK) -\lambda_{\min}(\bSigmaiK)  \rvert \le \norm{\hbSigmaiK-\bSigmaiK}_{\text{op},2}$, we have 
	\begin{align*}
		\bbeta_K^T \hbSigmaiK \bbeta_K 
		&\ge \bbeta_K^T\bbeta_K \lambda_{\min}(\hbSigmaiK) \\
		&\ge \bbeta_K^T\bbeta_K (\lambda_{\min}(\bSigmaiK) - \norm{\hbSigmaiK-\bSigmaiK}_{\text{op},2}) 1_{\{\lambda_{\min}(\bSigmaiK) \ge \norm{\hbSigmaiK-\bSigmaiK}_{\text{op},2} \}}.
	\end{align*}	
	Thus, using that $\norm{\hbSigmaiK  -\bSigmaiK}_{\text{op},2}=o_p(1)$, where the $o_p(1)$ term is uniform in $i$,  $\lambda_{\min}(\bSigmaiK) \ge \kappa_0$ a.s., and writing  $$p_0=P\left[\frac{1}{\bbeta_K^T \hbSigmaiK \bbeta_K} \le \frac{2}{\bbeta_K^T\bbeta_K \lambda_{\min}(\bSigmaiK)} \ \text{and} \ \lambda_{\min}(\hbSigmaiK)\ge \kappa_0/2 \right], $$ it follows that
	\begin{align*}
		p_0&\ge 
		P[ \bbeta_K^T  \bbeta_K \lambda_{\min}(\bSigmaiK) \le 2 \bbeta_K^T  \bbeta_K  \lambda_{\min}(\hbSigmaiK) \ 
		\text{and} \ \lambda_{\min}(\hbSigmaiK)\ge \kappa_0/2 ]\\
		&\ge 
		P[ \bbeta_K^T  \bbeta_K \lambda_{\min}(\bSigmaiK) \le 2 \bbeta_K^T  \bbeta_K  (\lambda_{\min}(\bSigmaiK) - \norm{\hbSigmaiK-\bSigmaiK}_{\text{op},2})  \ \text{and} \ \lambda_{\min}(\hbSigmaiK)\ge \kappa_0/2 ]\\
		&\ge
		P[ \kappa_0/2 \ge \norm{\hbSigmaiK-\bSigmaiK}_{\text{op},2} \ \text{and} \ \lambda_{\min}(\hbSigmaiK)\ge \kappa_0/2 ]\\
		&\ge 1- P\Big[ \norm{\hbSigmaiK-\bSigmaiK}_{\text{op},2} > \kappa_0/2\Big] - P[\lambda_{\min}(\hbSigmaiK)< \kappa_0/2].
	\end{align*}
	This implies $p_0\to 1$ as $n\to\infty$ and hence the event $(\bbeta_K^T \hbSigmaiK \bbeta_K)\inv \le 2 (\bbeta_K^T\bbeta_K \lambda_{\min}(\bSigmaiK))\inv$ with $\lambda_{\min}(\hbSigmaiK)\ge \kappa_0/2$ occurs with probability tending to $1$. It then suffices to work on this event in what follows. Combining with \eqref{eq:wassPred2}, \eqref{eq:wassPred3}, and 
	\begin{align*}
		\lvert (\hbbeta_K-\bbeta_K)^T \hbxiiK +\hbeta_0-\beta_0\rvert &\le \normtwo{\hbbeta_K-\bbeta_K} \left(\normtwo{\hbxiiK-\tbxiiK} + \normtwo{\tbxiiK}\right)+\lvert \hbeta_0-\beta_0\rvert\\
		&=O_p(\alpha_n),
	\end{align*}
	which follows from Lemma \ref{lem:Lemma_beta} and the facts that $\hbeta_0-\beta_0=\bar{Y}_n-E(Y)=O_p(n^{-1/2})$, $\normtwo{\hbxiiK-\tbxiiK}=O_p(a_n+b_n)$ and $\normtwo{\tbxiiK}=O_p(1)$ hold uniformly in $i$, then leads to
	\begin{align}
		\cW_2(\hcP_{iK},\tcP_{iK})&= O_p(\alpha_n).\label{eq:wassPred4}
	\end{align}
	The result in \eqref{eq:wassCons} then follows from \eqref{eq:wassPred1} and \eqref{eq:wassPred4}.

Denote by $\varphi$ and  $\varPhi$ the density and cdf of a standard normal random variable, and define the quantities $\tu_{in}=\beta_0 +\bbeta_K^T \hbxiiK$, $\tsigma_{in}=(\bbeta_K^T \hbSigmaiK \bbeta_K)^{1/2}$, $u_i=\beta_0 +\bbeta_K^T \tbxiiK$, $\sigma_i=(\bbeta_K^T \bSigmaiK \bbeta_K)^{1/2}$ and $\Delta_{in}(t)= (t-u_i)/\sigma_i - (t-\tu_{in})/\tsigma_{in}$,  $t\in\mathbb{R}$. Then
	\begin{align}
		\sup_{t\in\mathbb{R}} \lvert  \tF_{iK}(t)-F_{iK}(t)\rvert &= \sup_{t\in\mathbb{R}} \big\lvert \varPhi \left(\frac{t-\tu_{in}}{\tsigma_{in}}\right) - \varPhi \left(\frac{t-u_i}{\sigma_i}\right) \big\rvert
		= \sup_{t\in\mathbb{R}} \big\lvert  \varphi(\varepsilon_s) \Delta_{in}(t) \big \rvert,\label{eq:KSineq}
	\end{align}
	where the second equality follows by a Taylor expansion and $\varepsilon_s$ is between $(t-\tu_{in})/\tsigma_{in}$ and $(t-\mu_i)/\sigma_i$. Defining $r_{in}(t)=(t-\tu_{in})/\tsigma_{in}$, $r_i(t)=(t-u_i)/\sigma_i$ and setting $I_{in}=[\min\{u_i,\tu_{in}\},\max\{u_i,\tu_{in}\}]$,
	\begin{align}
		\big\lvert  \varphi(\varepsilon_s) \Delta_{in}(t) \big \rvert &\le
		\varphi(0) \big\lvert  \Delta_{in}(t) \big \rvert 1_{\{ t \in I_{in} \}} + \varphi(\min\{\lvert r_{in}(t)\rvert, \lvert r_i(t)\rvert\}) \big\lvert  \Delta_{in}(t) \big \rvert 1_{\{ t \in I_{in}^c \}}\nonumber \\
		&\le \varphi(0) \big\lvert  \Delta_{in}(t) \big \rvert 1_{\{ t \in I_{in} \}} + [\varphi(r_{in}(t))+\varphi(r_i(t))] \big\lvert  \Delta_{in}(t) \big \rvert. \label{eq:phiEpsIneq}
	\end{align}
	Since $\tu_{in}-u_i=O_p(a_n+b_n)$, $\lvert \tsigma_{in} -\sigma_i \rvert\le \lvert \tsigma_{in}^2 -\sigma_i^2 \rvert/\sigma_i =O_p(a_n+b_n)$, $\lvert \tsigma_{in}\inv - \sigma_i\inv \rvert\le \lvert \tsigma_{in} -\sigma_i \rvert / (\tsigma_{in} \sigma_i)\le \lvert \tsigma_{in} -\sigma_i \rvert \sqrt{2} (\bbeta_K^T\bbeta_K \lambda_{\min}(\bSigmaiK))^{-1/2} \sigma_i \inv$ and $\lambda_{\min}(\bSigmaiK)\ge \kappa_0$ a.s., it follows that
	\begin{align}
		\lvert \Delta_{in}(t)  \rvert &= \lvert (t-u_i)/\sigma_i - (t-\tu_{in})/\tsigma_{in}  \rvert \nonumber\\
		&\le \frac{1}{\sigma_i} \lvert \tu_{in}-u_i \rvert + \lvert t-u_i \rvert \ \Big\lvert \frac{1}{\tsigma_{in}} - \frac{1}{\sigma_i} \Big\rvert
		+ \lvert \tu_{in}-u_i \rvert \ \Big\lvert \frac{1}{\tsigma_{in}} - \frac{1}{\sigma_i} \Big\rvert \nonumber \\
		&=  O_p(a_n+b_n) + O_p(a_n+b_n)  \lvert t-u_i \rvert, \label{eq:unifBoundDeltan}
	\end{align}
	where both $O_p(a_n+b_n)$ terms are uniform in $t$. This implies
	\begin{align}
		\sup_{t\in\mathbb{R}} \big\lvert  \Delta_{in}(t) \big \rvert 1_{\{ t \in I_{in} \}} &\le   O_p(a_n+b_n) + O_p(a_n+b_n) \lvert \tu_{in}-u_i \rvert =O_p(a_n+b_n) \label{eq:DeltatFinite}.
	\end{align}
	Since $\normop{\bSigmaiK}$ is uniformly bounded above in the sparse case, it is easy to show that $\varphi(r_i(t)) \lvert t-u_i \rvert \le O(1)$, where the $O(1)$ term is uniform in both $t$ and $i$. This combined with \eqref{eq:unifBoundDeltan}  leads to
	\begin{align}
		\sup_{t\in\mathbb{R}}  \varphi(r_i(t)) \lvert \Delta_{in}(t)  \rvert&= O_p(a_n+b_n).\label{eq:phi_Deltat}
	\end{align}
	From \eqref{eq:unifBoundDeltan}, 
	\begin{align*}
		\sup_{t\in\mathbb{R}} \varphi(r_{in}(t)) \lvert \Delta_{in}(t) \rvert &\le    O_p(a_n+b_n) +  O_p(a_n+b_n) \sup_{t\in\mathbb{R}} \varphi(r_{in}(t))  \lvert t-u_i \rvert,
	\end{align*}
	and the result then follows from \eqref{eq:KSineq}, \eqref{eq:phiEpsIneq}, \eqref{eq:DeltatFinite} and \eqref{eq:phi_Deltat} if we can show that $\varphi(r_{in}(t))  \lvert t-u_i \rvert =O_p(1)$ uniformly in  $t$. It is easy to see that 
	\begin{align*}
		\varphi(r_{in}(t))  \lvert t-u_i \rvert &\le   \varphi(r_{in}(t_1^*)) ( t_1^*-u_i) 1_{\{t\ge u_i\}}+ \varphi(r_{in}(t_2^*)) ( u_i-t_2^*) 1_{\{t\le u_i\}}\\
		&\le \varphi(r_{in}(t_1^*)) ( t_1^*-u_i) + \varphi(r_{in}(t_2^*)) ( u_i-t_2^*)\\
		&\le \varphi(0) (t_1^*- t_2^*),
	\end{align*}
	where $t_1^*=(u_i+\tu_{in}+\sqrt{(u_i-\tu_{in})^2+4\tsigma_{in}^2})/2$ and $t_2^*=(u_i+\tu_{in}-\sqrt{(u_i-\tu_{in})^2+4\tsigma_{in}^2})/2$. Since $\tsigma_{in}$ is uniformly upper bounded in the sparse setting and $\tu_{in}-u_i=O_p(a_n+b_n)$, we obtain
	\begin{align*}
		\sup_{t\in\mathbb{R}} 	\varphi(r_{in}(t))  \lvert t-u_i \rvert &\le \varphi(0) \sqrt{(u_i-\tu_{in})^2+4\tsigma_{in}^2}=O_p(1).
	\end{align*}
	Therefore
	\begin{align}
		\sup_{t\in\mathbb{R}} \lvert  \tF_{iK}(t)-F_{iK}(t)\rvert &= O_p(a_n+b_n),\label{eq:KS1}
	\end{align}
	so that  it then remains to control the term $\sup_{t\in\mathbb{R}} \lvert  \hF_{iK}(t)-\tF_{iK}(t)\rvert $. For this purpose, define auxiliary quantities  $\hu_{in}=\hbeta_0 +\hbbeta_K^T \hbxiiK$, $\hsigma_{in}=(\hbbeta_K^T \hbSigmaiK \hbbeta_K)^{1/2}$  and $\hDelta_{in}(t)= (t-\hu_{in})/\hsigma_{in}- (t-\tu_{in})/\tsigma_{in}$,  $t\in\mathbb{R}$. From Lemma \ref{lem:Lemma_beta} it follows that $\hu_{in}-\tu_{in}=\hbeta_0-\beta_0 + (\hbbeta_K-\bbeta_K)^T(\hbxiiK-\bxi_{iK})+(\hbbeta_K-\bbeta_K)^T \bxi_{iK} =O_p(\alpha_n)$, $\lvert \hsigma_{in} -\tsigma_{in} \rvert\le \lvert \hsigma_{in}^2 -\tsigma_{in}^2 \rvert/\tsigma_{in} =O_p(\alpha_n)$, which is due to  \eqref{eq:wassPred3} and since $\tsigma_{in}\inv \le \sqrt{2} (\bbeta_K^T\bbeta_K \lambda_{\min}(\bSigmaiK))^{-1/2}$. Also, from \eqref{eq:wassPred3} and using $\lambda_{\min}(\bSigmaiK)\ge \kappa_0$ a.s. we have $\vert \hsigma_{in} - \tsigma_{in}\rvert  \le \vert \hsigma_{in}^2 - \tsigma_{in}^2\rvert / \tsigma_{in}  \le \vert \hsigma_{in}^2 - \tsigma_{in}^2\rvert  \sqrt{2} (\bbeta_K^T\bbeta_K \lambda_{\min}(\bSigmaiK))^{-1/2}=o_p(1)$ and then $\vert \hsigma_{in} - \sigma_{i}\rvert  \le \vert \hsigma_{in} - \tsigma_{in}\rvert + \vert \tsigma_{in} - \sigma_{i}\rvert=o_p(1)$. This along with the fact that $\hsigma_{in}\ge \lVert \hbbeta_K\rVert_2\kappa_0/2\ge \lVert \bbeta_K\rVert_2\kappa_0/4$ holds with probability tending to $1$ implies $\hsigma_{in}\inv \le 2 \sigma_{i}\inv$ with probability tending to $1$ as $n\to\infty$. Combining this with $\lambda_{\min}(\bSigmaiK)\ge \kappa_0$ a.s. then leads to
	\begin{align*}
		\Big \lvert \frac{1}{\hsigma_{in}} - \frac{1}{\tsigma_{in}} \Big \rvert =O_p(\alpha_n),
	\end{align*}
	where the bound is uniform in $i$, and similarly as in \eqref{eq:unifBoundDeltan} we obtain
	\begin{align}
		\lvert \hDelta_{in}(t)\rvert&\le \lvert t-\tu_{in} \rvert \ \Big\lvert \frac{1}{\hsigma_{in}} - \frac{1}{\tsigma_{in}} \Big\rvert + 
		\lvert \hu_{in}-\tu_{in} \rvert  \ \Big\lvert \frac{1}{\hsigma_{in}} - \frac{1}{\tsigma_{in}} \Big\rvert +  \lvert \hu_{in}-\tu_{in} \rvert  \frac{1}{\tsigma_{in}}\nonumber \\
		&\le O_p(\alpha_n) + O_p(\alpha_n)  \lvert t-\tu_{in} \rvert. \label{eq:unifBoundDeltaht}
	\end{align}
	Next
	\begin{align*}
		&\varphi(r_{in}(t)) \lvert t-\tu_{in} \rvert \\
		&\le \varphi(1) \sqrt{\bbeta_K^T \hbSigmaiK \bbeta_K} \le \varphi(1)  (\bbeta_K^T \bbeta_K)^{1/2}  \left(\norm{\hbSigmaiK  -\bSigmaiK}_{\text{op},2} + \norm{\bSigmaiK}_{\text{op},2} \right)^{1/2}\\
		=&O_p(1),
	\end{align*}
	where the $O_p(1)$ term is uniform in both $t$ and $i$. This combined with \eqref{eq:unifBoundDeltaht} shows that
	\begin{align}
		\sup_{t\in\mathbb{R}}  \varphi(r_{in}(t)) \lvert \hDelta_{in}(t)  \rvert&= O_p(\alpha_n).\label{eq:phi_Deltaht}
	\end{align}
	Setting  $\hr_{in}(t)=(t-\hu_{in})/\hsigma_{in}$, similar arguments as before lead to
	\begin{align*}
		\sup_{t\in\mathbb{R}}  \varphi(\hr_{in}(t)) \lvert  t-\tu_{in} \rvert&\le  \varphi(0) \sqrt{(\hu_{in}-\tu_{in})^2+4\hsigma_{in}^2}=O_p(1),
	\end{align*}
	where the last equality is due to $\lvert \hu_{in}-\tu_{in} \rvert = O_p(\alpha_n)$ and $\hsigma_{in}^2 \le \bbeta_K^T \bbeta_K  \Big(\norm{\hbSigmaiK  -\bSigmaiK}_{\text{op},2} + \norm{\bSigmaiK}_{\text{op},2}\Big )=O_p(1)$.  With \eqref{eq:unifBoundDeltaht} this implies 
	\begin{align}
		\sup_{t\in\mathbb{R}}  \varphi(\hr_{in}(t)) \lvert \hDelta_{in}(t)  \rvert&= O_p(\alpha_n).\label{eq:phirht_Deltaht}
	\end{align}
	Setting  $\hI_{in}=[\min\{\hu_{in},\tu_{in}\},\max\{\hu_{in},\tu_{in}\}]$, then similar arguments as the ones outlined in \eqref{eq:KSineq} and \eqref{eq:phiEpsIneq} shows that
	\begin{align}
		\sup_{t\in\mathbb{R}} \lvert  \hF_{iK}(t)-\tF_{iK}(t)\rvert & \le 
		\varphi(0) \big\lvert  \hDelta_{in}(t) \big \rvert 1_{\{ t \in\hI_{in} \}} + [\varphi(\hr_{in}(t))+\varphi(r_{in}(t))] \big\lvert  \hDelta_{in}(t) \big \rvert.
	\end{align}
	This together with $\sup_{t\in\mathbb{R}}  \big\lvert  \hDelta_{in}(t) \big \rvert 1_{\{ t \in\hI_{in} \}}\le O_p(\alpha_n) + O_p(\alpha_n)  \lvert \hu_{in}-\tu_{in} \rvert =O_p(\alpha_n)$, where the latter  follows from \eqref{eq:unifBoundDeltaht}, as well as \eqref{eq:phi_Deltaht} and \eqref{eq:phirht_Deltaht} then leads to
	\begin{align}
		\sup_{t\in\mathbb{R}} \lvert  \hF_{iK}(t)-\tF_{iK}(t)\rvert &= O_p(\alpha_n).\label{eq:KS2}
	\end{align}
	The result in \eqref{eq:wassKS} then follows from \eqref{eq:KS1}, \eqref{eq:KS2} and the triangle inequality.

	For the next result in \eqref{eq:wassDensL2}, similarly as before we first start by showing that $\lVert  \tf_{iK}-f_{iK} \rVert_{L^2(\mathbb{R})}  = O_p(a_n+b_n)$, where $ \tf_i(t):=\tF_i'(t)=\varphi((t-\tu_{in})/\tsigma_{in})/\tsigma_{in}$. Since $ f_i(t)=F_i'(t)=\varphi((t-u_i)/\sigma_i)/\sigma_i$, we have
	\begin{align}
		\Big \lVert \frac{1}{\tsigma_{in}} \varphi \left( \frac{\cdot - \tu_{in}}{\tsigma_{in}}\right)  - \frac{1}{\sigma_i} \varphi \left( \frac{\cdot - u_i}{\sigma_i}\right) \Big \rVert_{L^2(\mathbb{R} )} &\le 
		\frac{1}{\tsigma_{in}} \Big \lVert  \varphi \left( \frac{\cdot - \tu_{in}}{\tsigma_{in}}\right)  - \varphi \left( \frac{\cdot - u_i}{\sigma_i}\right) \Big \rVert_{L^2(\mathbb{R} )} \nonumber \\ 
		& \quad +  \Big \lvert \frac{1}{\tsigma_{in}}- \frac{1}{\sigma_i}\Big \rvert  \	\Big \lVert  \varphi \left( \frac{\cdot - u_i}{\sigma_i}\right) \Big \rVert_{L^2(\mathbb{R} )}.\label{eq:pdfIneq}
	\end{align}
	Thus, since $\lVert  \varphi \left( \frac{\cdot - u_i}{\sigma_i}\right) \rVert_{L^2(\mathbb{R} )}=O(\sigma_i^{1/2})$ and $\lvert \tsigma_{in}\inv -\sigma_i\inv  \rvert =O_p(a_n+b_n)$, we obtain 
	\begin{align}
		\Big \lvert \frac{1}{\tsigma_{in}}- \frac{1}{\sigma_i}\Big \rvert  \ \Big \lVert  \varphi \left( \frac{\cdot - u_i}{\sigma_i}\right) \Big \rVert_{L^2(\mathbb{R} )} &=O_p(a_n+b_n).\label{eq:pdfIneq_1}
	\end{align}
	Using  the relation $\varphi'(t)=-t \varphi(t)$ and a Taylor expansion, it follows that
	\begin{align*}
		\Big  \lVert  \varphi \left( \frac{\cdot - \tu_{in}}{\tsigma_{in}}\right)  - \varphi \left( \frac{\cdot - u_i}{\sigma_i}\right) \Big  \rVert_{L^2(\mathbb{R} )}^2  &= \int_\mathbb{R} (\varphi'(\varepsilon_t))^2 \Delta_{in}^2(t) dt= \int_\mathbb{R} \varepsilon_t^2 \varphi^2 (\varepsilon_t)\Delta_{in}^2(t) dt,
	\end{align*}
	where $\varepsilon_t$ is between $r_{in}(t)$ and $r_i(t)$. Hence, from \eqref{eq:pdfIneq} and \eqref{eq:pdfIneq_1} it  suffices to show that $\int_\mathbb{R} \varepsilon_t^2 \varphi^2 (\varepsilon_t)\Delta_{in}^2(t) dt=O_p((a_n+b_n)^2)$. Indeed, from the fact that $\lvert \varepsilon_t \rvert\le \lvert r_{in}(t)\rvert+\lvert r_i(t)\rvert$, $ \sup_{t\in I_{in}} \lvert r_{in}(t)\rvert =O_p(a_n+b_n)$, $ \sup_{t\in I_{in}} \lvert r_i(t)\rvert =O_p(a_n+b_n)$ and $ \varphi(\varepsilon_t) 1_{\{ t\in I_{in}^c\}} \le \varphi(r_{in}(t)) +\varphi(r_i(t))$, we have
	\begin{align}
		&\int_\mathbb{R} \varepsilon_t^2 \varphi^2 (\varepsilon_t)\Delta_{in}^2(t) dt\nonumber \\
		&= \int_{I_{in}} \varepsilon_t^2 \varphi^2 (\varepsilon_t)\Delta_{in}^2(t) dt+\int_{I_{in}^c} \varepsilon_t^2 \varphi^2 (\varepsilon_t)\Delta_{in}^2(t) dt \nonumber \\
		&\le  \varphi^2(0) O_p((a_n+b_n)^5) + \int_{I_{in}^c} [\varphi(r_{in}(t)) +\varphi(r_i(t))]^2 (r_{in}(t)+r_i(t))^2   \Delta_{in}^2(t) dt \nonumber \\
		&\le  O_p((a_n+b_n)^5)  + \int_{\mathbb{R}} [\varphi(r_{in}(t)) +\varphi(r_i(t))]^2 (r_{in}(t)+r_i(t))^2   \Delta_{in}^2(t) dt, \label{eq:pdfIneq_2}
	\end{align}
	where the first inequality follows from \eqref{eq:DeltatFinite} and the relation $\lvert r_{in}(t)\rvert \lvert r_i(t)\rvert 1_{\{ t\in I_{in}^c\}}  = r_{in}(t) r_i(t)1_{\{ t\in I_{in}^c\}}$. From $\int_{\mathbb{R}} \varphi^2(s) \lvert s\rvert^p ds<\infty $, $p\in\mathbb{N}$,  we obtain the following facts:
	\begin{align*}
		\int_{\mathbb{R}} \varphi^2(r_{in}(t)) r_{in}^2(t) \Delta_{in}^2(t)   dt&\le \sigma_{i}^{-2} O_p((a_n+b_n)^2),\\
		\int_{\mathbb{R}} \varphi^2(r_{in}(t)) r_{i}^2(t) \Delta_{in}^2(t)   dt&\le \sigma_{i}^{-4} O_p((a_n+b_n)^2),\\
		\Big\lvert \int_{\mathbb{R}} \varphi^2(r_{in}(t)) r_{in}(t)  r_{i}(t) \Delta_{in}^2(t)   dt \Big\rvert&\le \sigma_{i}^{-3} O_p((a_n+b_n)^2),\\
		\Big\lvert \int_{\mathbb{R}} \varphi^2(r_{i}(t)) r_{in}(t)  r_{i}(t) \Delta_{in}^2(t)   dt\Big\rvert &\le  (\bbeta_K^T\bbeta_K \lambda_{\min}(\bSigmaiK))^{-3/2} O_p((a_n+b_n)^2),\\
		\int_{\mathbb{R}} \varphi^2(r_{i}(t))  r_{i}^2(t) \Delta_{in}^2(t)   dt&\le  (\bbeta_K^T\bbeta_K \lambda_{\min}(\bSigmaiK))\inv O_p((a_n+b_n)^2),\\
		\int_{\mathbb{R}} \varphi^2(r_{i}(t))  r_{in}^2(t) \Delta_{in}^2(t)   dt&\le  (\bbeta_K^T\bbeta_K \lambda_{\min}(\bSigmaiK))^{-2} O_p((a_n+b_n)^2),\\
		\Big\lvert \int_{\mathbb{R}} \varphi(r_{in}(t)) \varphi(r_{i}(t))  r_{i}(t) r_{in}(t) \Delta_{in}^2(t)   dt\Big\rvert &\le \sigma_{i}^{-3}  O_p((a_n+b_n)^2).
	\end{align*}
	These facts along with \eqref{eq:pdfIneq_2} imply $\int_\mathbb{R} \varepsilon_t^2 \varphi^2 (\varepsilon_t)\Delta_{in}^2(t) dt\le O_p((a_n+b_n)^5)  + O_p((a_n+b_n)^2)=O_p((a_n+b_n)^2)$ and 
	\begin{align*}
		\Big \lVert  \tf_{iK}-f_{iK} \Big \rVert_{L^2(\mathbb{R})}&= O_p(a_n+b_n).
	\end{align*}
	Similar arguments imply $\lVert  \hf_{iK}-\tf_{iK} \rVert_{L^2(\mathbb{R})} =O_p(\alpha_n)$ and the result in \eqref{eq:wassDensL2}.
	
	Finally, from condition \ref{a:LM_uncond} we have  $\lambda_{\min}(\bSigmaiK)\ge\kappa_0$ and also $\sigma_i^2 = (\bbeta_K^T \bSigmaiK \bbeta_K) \ge \bbeta_{K}^T \bbeta_{K} \lambda_{\min}(\bSigmaiK)\ge \bbeta_{K}^T \bbeta_{K} \kappa_0$ a.s., which implies $\sigma_i\inv =O(1)$ and $ \lambda_{\min}(\bSigmaiK)\inv =O(1)$ a.s., where the $O(1)$ terms are uniform in $i$. Since $\norm{\bSigma_{iK}-\hat{\bSigma}_{iK}}_F=O(a_n+b_n)$ a.s.  as $n\to\infty$, where the $O(a_n+b_n)$ term is uniform over $i$, and $\normtwo{\hbxiiK-\tbxiiK}=O_p(a_n+b_n)$, where the $O_p(a_n+b_n)$ term is also uniform over $i$,  it can be easily checked from the previous arguments that the rates of convergence in \eqref{eq:wassCons}, \eqref{eq:wassKS} and \eqref{eq:wassDensL2} are uniform in $i$.
\end{proof}

\vspace{1cm}

%\noindent{\bf Theorem \ref{thm:WassDiscrep}}\vspace{-.2cm} 

\begin{proof}[Proof  of Theorem \ref{thm:WassDiscrep}]
	Recall that $\eta_{iK}:=\beta_0+\bbeta_K^T \bxi_{iK}$ is the $K$-truncated linear predictor for the $i$th subject and $\teta_{iK}:=\beta_0+\bbeta_K^T \tbxi_{iK}$ its best prediction. Also, recall that $\cP_{iK}$ corresponds to the predictive distribution of $\eta_{iK}$ given $\bXi$ and $\bTi$, and $\hcP_{iK}$ is the corresponding  estimate. Writing $Y_i=\beta_0+\bbeta_K^T \bxi_{iK}+\sum_{k\ge K+1} \beta_k \xi_{ik}+\epsilon_{iY}=\eta_{iK}+ R_{iK}+\epsilon_{iY}$, where $R_{iK}=\sum_{k\ge K+1} \beta_k \xi_{ik}$, the estimated Wasserstein discrepancy is given by $\hcD_{nK}=n\inv \sumin W_2^2(\delta_{Y_i}, \hcP_{iK})$, where
	\begin{align}
		&n\inv \sumin \cW_2^2(\cA_{Y_i}, \hcP_{iK})\nonumber \\
		&=n\inv \sumin (Y_i - \heta_{iK})^2 + n\inv \sumin \hbbeta_K^T \hbSigmaiK \hbbeta_K\nonumber  \\
		&=n\inv \sumin (\eta_{iK} - \heta_{iK})^2+ n\inv \sumin \epsilon_{iY}^2+n\inv \sumin  R_{iK}^2+ 2 n\inv \sumin (\eta_{iK} - \heta_{iK}) \epsilon_{iY}  \nonumber \\
		&\quad + 2 n\inv \sumin (\eta_{iK} - \heta_{iK}) R_{iK} + 2 n\inv \sumin  R_{iK} \epsilon_{iY}
		+ n\inv \sumin \hbbeta_K^T \hbSigmaiK \hbbeta_K.\label{eq:WassDiscrep_0}
	\end{align}
	Since $n_i=m_0<N_0$, by the central limit theorem,
	\begin{align*}
		n\inv \sumin (\eta_{iK} - \teta_{iK}) R_{iK}&= -\bbeta_K^T E\left( \bLambda_K\bPhi_{1K}^T\bSigma_1\inv  \sum_{k\ge K+1} \phi_k(\bT_1)\lambda_k \beta_k \right)+ O_p(n^{-1/2}),
	\end{align*}
	and
	\begin{align}
		n\inv \sumin  R_{iK}^2&=\sum_{k\ge K+1} \beta_k^2 \lambda_k+O_p(n^{-1/2})\label{eq:WassDiscrep_0_1}.
	\end{align}
	Combining this with $n\inv \sumin (\teta_{iK} - \heta_{iK})^2 =O_p(\alpha_n^2)$, as shown in the proof of Lemma \ref{lem:Wass_Discrep}, 
	\begin{align}
		n\inv \sumin (\eta_{iK} - \heta_{iK}) R_{iK}&= n\inv \sumin (\eta_{iK} - \teta_{iK}) R_{iK} +n\inv \sumin (\teta_{iK} - \heta_{iK}) R_{iK} \nonumber \\
		&= -\bbeta_K^T E\left( \bLambda_K\bPhi_{1K}^T\bSigma_1\inv  \sum_{k\ge K+1} \phi_k(\bT_1)\lambda_k\beta_k \right)+ O_p(\alpha_n)\label{eq:WassDiscrep_1}.
	\end{align}
	Next
	\begin{align}
		n\inv \sumin (\eta_{iK} - \heta_{iK}) \epsilon_{iY}&=n\inv \sumin (\eta_{iK} - \teta_{iK}) \epsilon_{iY}+n\inv \sumin (\teta_{iK} - \heta_{iK}) \epsilon_{iY} \nonumber\\
		&= O_p(n^{-1/2})+O_p(\alpha_n)=O_p(\alpha_n)\label{eq:WassDiscrep_2},
	\end{align}
	where the last equality follows from Lemma \ref{lem:Wass_Discrep} and since $n\inv \sumin (\eta_{iK} - \teta_{iK}) \epsilon_{iY}= O_p(n^{-1/2})$, which is due to the Central Limit Theorem. Similarly, from Lemma \ref{lem:Wass_Discrep2} we have
	\begin{align}
		n\inv \sumin (\eta_{iK} - \teta_{iK})^2 &=\bbeta_K^T E(\bSigma_{1K})\bbeta_K+O_p(n^{-1/2})\label{eq:WassDiscrep_3},
	\end{align}
	and
	\begin{align}
		n\inv \sumin (\eta_{iK}& - \heta_{iK})^2 -n\inv \sumin (\eta_{iK} - \teta_{iK})^2= n\inv \sumin (\teta_{iK} - \heta_{iK})^2\nonumber \\
		& +2 n\inv \sumin (\eta_{iK} - \teta_{iK})(\teta_{iK} - \heta_{iK}
		=O_p(\alpha_n)\label{eq:WassDiscrep_4},
	\end{align}
	where the last equality follows from the fact that $n\inv \sumin (\teta_{iK} - \heta_{iK})^2 =O_p(\alpha_n^2)$, \eqref{eq:WassDiscrep_3} and the Cauchy--Schwarz inequality. Combining \eqref{eq:WassDiscrep_3} and \eqref{eq:WassDiscrep_4} leads to
	\begin{align}
		n\inv \sumin (\eta_{iK} - \heta_{iK})^2&= \bbeta_K^T E(\bSigma_{1K})\bbeta_K+O_p(\alpha_n) \label{eq:WassDiscrep_5}.
	\end{align}
	We further note that
	\begin{align*}
		&\lvert \hbbeta_K^T \hbSigmaiK \hbbeta_K -\bbeta_K^T \bSigmaiK \bbeta_K\rvert \\
		&=\lvert
		\hbbeta_K^T (\hbSigmaiK-\bSigmaiK ) \hbbeta_K + (\hbbeta_K-\bbeta_K)^T \bSigmaiK \hbbeta_K+\bbeta_K^T \bSigmaiK (\hbbeta_K-\bbeta_K)\rvert \\
		&\le \lVert\hbbeta_K\rVert_2^2 \lVert \hbSigmaiK-\bSigmaiK\rVert_{\text{op},2} + \lVert \hbbeta_K-\bbeta_K\rVert_2 \lVert \bSigmaiK \rVert_{\text{op},2} (\lVert \hbbeta_K \rVert_2 + \lVert \bbeta_K \rVert_2 ).
	\end{align*}
	From the proof of Theorem \ref{thm:predictiveConsistency}, we have $\norm{\bSigma_{iK}-\hat{\bSigma}_{iK}}_F=O(a_n+b_n)$ a.s.  as $n\to\infty$, where the $O(a_n+b_n)$ term is uniform in $i$. Since $\lVert \bSigma_{iK}\rVert_{F}=O(1)$ uniformly over $i$, 
	\begin{align*}
		\Big\lvert n\inv \sumin (\hbbeta_K^T \hbSigmaiK \hbbeta_K-\bbeta_K^T \bSigmaiK \bbeta_K)\Big\rvert 
		&\le n\inv \sumin \lvert \hbbeta_K^T \hbSigmaiK \hbbeta_K -\bbeta_K^T \bSigmaiK \bbeta_K\rvert \\
		&\le 	\lVert \hbbeta_K-\bbeta_K\rVert_2  (\lVert \hbbeta_K \rVert_2 + \lVert \bbeta_K \rVert_2 ) \ n\inv \sumin  \lVert \bSigmaiK \rVert_{F} \\
		&\quad + \lVert\hbbeta_K\rVert_2^2  \ n\inv \sumin \lVert \hbSigmaiK-\bSigmaiK\rVert_{F} \\ 
		&\le \lVert \hbbeta_K-\bbeta_K\rVert_2  (\lVert \hbbeta_K \rVert_2 + \lVert \bbeta_K \rVert_2 ) O(1)  \\
		&+ \lVert\hbbeta_K\rVert_2^2 O(a_n+b_n) \quadas,
	\end{align*}
	as $\ntoinf$. From Lemma \ref{lem:Lemma_beta}, we have $\lVert \hbbeta_K-\bbeta_K\rVert_2=O_p(\alpha_n)$, which combined with $\lVert \hbbeta_K \rVert_2\le \lVert \hbbeta_K -\bbeta_K \rVert_2+\lVert \bbeta_K \rVert_2=O_p(1)$ leads to
	\begin{align*}
		n\inv \sumin \hbbeta_K^T \hbSigmaiK \hbbeta_K - n\inv \sumin \bbeta_K^T \bSigmaiK \bbeta_K &=O_p(\alpha_n).
	\end{align*}
	This along with an application of the Central Limit Theorem shows that
	\begin{align}
		n\inv \sumin \hbbeta_K^T \hbSigmaiK \hbbeta_K &=   \bbeta_K^T E(\bSigma_{1K})  \bbeta_K +O_p(\alpha_n)\label{eq:WassDiscrep_6}.
	\end{align}
	Finally, it is easy to show that $n\inv \sumin  R_{iK} \epsilon_{iY}= O_p(n^{-1/2})$ and $n\inv \sumin \epsilon_{iY}^2= \sigma_Y^2 + O_p(n^{-1/2})$, applying  the CLT. Combining with \eqref{eq:WassDiscrep_0_1}, \eqref{eq:WassDiscrep_1}, \eqref{eq:WassDiscrep_2}, \eqref{eq:WassDiscrep_5}, and \eqref{eq:WassDiscrep_6},
	\begin{align*}
		\hcD_{nK}
		&=2\bbeta_K^T E(\bSigma_{1K})\bbeta_K+\sigma_Y^2+\sum_{k\ge K+1} \beta_k^2 \lambda_k -2\bbeta_K^T E\left( \bLambda_K\bPhi_{1K}^T\bSigma_1\inv  \sum_{k\ge K+1} \phi_k(\bT_1)\lambda_k \beta_k \right)\\
		&+O_p(\alpha_n),
	\end{align*}
	implying the first result in \eqref{eq:WassDiscrepancy1}. Similar arguments show that the Wasserstein distance using true population quantities $\cD_{nK}$ is such that
	\begin{align*}
		\cD_{nK}=n\inv \sumin \cW_2^2(\cA_{Y_i}, \mathcal{P}_{iK})&=n\inv \sumin (Y_i - \teta_{iK})^2 + n\inv \sumin \bbeta_K^T \bSigmaiK \bbeta_K\\
		&=\cD_K +O_p(n^{-1/2}),
	\end{align*}
	where 
	\begin{align*}
		\cD_K
		&=2\bbeta_K^T E(\bSigma_{1K})\bbeta_K+\sigma_Y^2+\sum_{k\ge K+1} \beta_k^2 \lambda_k \\
		&-2\bbeta_K^T E\left( \bLambda_K\bPhi_{1K}^T\bSigma_1\inv  \sum_{k\ge K+1} \phi_k(\bT_1)\lambda_k \beta_k \right).
	\end{align*}

	Since $Y=\mu_Y+\int_\cT \beta(t) U(t) +\epsilon_Y$, where $\mu_Y=E(Y)$ and $U(t)=X(t)-\mu(t)$, we have $E(Y^2)=\mu_Y^2+\sigma_Y^2+E( \langle \beta, U \rangle_{L^2}^2)$, where $\langle \cdot,\cdot \rangle_L^2$ is the $L^2(\cT)$ inner product. From \ref{a:GaussProcess} it follows that $E( \langle \beta, U \rangle_{L^2}^2)=\sum_{j=1}^\infty \beta_j^2 \lambda_j$ as the FPCs are independent in the Gaussian case. Then
	\begin{align}
		n\inv \sumin (Y_i-\bar{Y}_n)^2&=\Var(Y)+O_p(n^{-1/2})=\sigma_Y^2+ \sum_{j=1}^\infty \lambda_j\beta_j^2 +O_p(n^{-1/2})\label{eq:WassDiscrep_7}.
	\end{align}
	Also, $\lvert \hbeta_j \rvert\le \lVert \hbeta_M\rVert_{L^2}$ and $\lvert \beta_j \rvert\le \lVert \beta\rVert_{L^2}$. With perturbation results as used in the proof of Lemma \ref{lem:auxLemma_beta_5} this leads to 
	\begin{align}
		&\Big \lvert  \sumM \hlambda_m \hbeta_m^2  - \lambda_m\beta_m^2     \Big \rvert\nonumber \\
		&\le 
		\sumM \lvert \hlambda_m -\lambda_m\rvert \lvert  \hbeta_m^2- \beta_m^2\rvert +
		\sumM \lvert \hlambda_m -\lambda_m\rvert \beta_m^2 +
		\sumM \lambda_m\lvert  \hbeta_m^2- \beta_m^2\rvert \nonumber \\
		&\le \lVert \hXi-\Xi \rVert_{\text{op}} (\lVert \hbeta_M\rVert_{L^2}+\lVert \beta\rVert_{L^2}) \sumM \lvert \hbeta_m- \beta_m \rvert +
		\lVert \hXi-\Xi \rVert_{\text{op}} \sumM \beta_m^2 \nonumber\\
		&\quad + 
		(\lVert \hbeta_M\rVert_{L^2}+\lVert \beta\rVert_{L^2}) \sumM \lambda_m \lvert \hbeta_m- \beta_m \rvert \label{eq:WassDiscrep_8}.
	\end{align}
	From the proof of Lemma \ref{lem:Lemma_beta} and since $\sum_{j=1}^\infty \lambda_j<\infty$, we have
	\begin{align}
		&\sumM \lambda_m \lvert \hbeta_m- \beta_m \rvert \nonumber \\
		&\le \lVert \hbeta_M-\beta\rVert_{L^2}\sumM \lambda_m \lVert \hphi_m-\phi_m \rVert_{L^2}+ \lVert\hbeta_M-\beta\rVert_{L^2}\left(\sumM \lambda_m\right)\nonumber\\
		&\quad + \lVert \beta \rVert_{L^2} \sumM \lambda_m \lVert \hphi_m-\phi_m \rVert_{L^2} \nonumber\\
		&\le \left(\sum_{j=1}^\infty \lambda_j\right) \lVert\hbeta_M-\beta\rVert_{L^2}+ 2\sqrt{2} \lVert \hXi-\Xi \rVert_{\text{op}}   \left(\lVert \hbeta_M-\beta\rVert_{L^2}+\lVert \beta\rVert_{L^2}\right) \left(\sumM\frac{\lambda_m}{\delta_m}\right)\quadas\nonumber\\
		&\le O_p(\alpha_n) + O_p(1) O(c_n^{\rho})=O_p(\alpha_n)\label{eq:WassDiscrep_9},
	\end{align}
	where the last inequality follows from Lemma \ref{lem:series_eigengap} and Lemma \ref{lem:Lemma_beta}. Similarly
	\begin{align}
		&\sumM \lvert \hbeta_m- \beta_m \rvert\nonumber \\
		&\le 2\sqrt{2}\lVert \hXi-\Xi \rVert_{\text{op}} \left(\lVert \hbeta_M-\beta\rVert_{L^2}+\lVert \beta\rVert_{L^2}\right) \left(\sumM \frac{1}{\delta_m}\right)+ \lVert \hbeta_M-\beta\rVert_{L^2} M \quadas \nonumber\\
		&\le O(c_n) O(c_n^{\rho-1}) \left(\lVert \hbeta_M-\beta\rVert_{L^2}+\lVert \beta\rVert_{L^2}\right) +\lVert \hbeta_M-\beta\rVert_{L^2} O(c_n^{\rho-1})\quadas \nonumber \\
		&\le O_p(c_n^\rho)+O_p(c_n^{\rho-1} \alpha_n)\label{eq:WassDiscrep_10},
	\end{align}
	where the second and third inequalities follow from Lemma \ref{lem:Lemma_beta} and using that $\sumM \delta_m\inv=O(c_n^{\rho-1})$, which was shown in the proof of Lemma \ref{lem:auxLemma_beta_5}, along with the fact that $M=O(c_n^{\rho-1})$, which is due to the condition $\sum_{m=1}^M \frac{1}{\sqrt{\lambda_m}\delta_m}=O(c_n^{\rho-1})$ and $0<\delta_m<\lambda_m\le \lambda_1$. Combining \eqref{eq:WassDiscrep_8}, \eqref{eq:WassDiscrep_9} and \eqref{eq:WassDiscrep_10} leads to
	\begin{align*}
		\Big \lvert  \sumM \hlambda_m \hbeta_m^2  - \sumM \lambda_m\beta_m^2     \Big \rvert&=O_p(\alpha_n) .
	\end{align*}
	This implies
	\begin{align*}
		\Big \lvert \sumM \hlambda_m  \hbeta_m^2 -\sum_{m=1}^\infty \lambda_m \beta_m^2    \Big \rvert&\le O_p(\alpha_n) + \sum_{m\ge M+1} \lambda_m \beta_m^2,
	\end{align*}
	and the result in \eqref{eq:WassDiscrepancy2} follows from \eqref{eq:WassDiscrep_7}.
\end{proof}

\vspace{1cm}

%\noindent{\bf Theorem \ref{thm:predictiveShrinkage_FLM}}\vspace{-.2cm} 

%\begin{proof}[Proof]
\begin{proof}[Proof of Theorem \ref{thm:predictiveShrinkage_FLM}]
	Note that
	\begin{align*}
		\lVert  \tbxiKs - \bxi_{K}^* \rVert_2^2
		&=
		\sumkK (  \lambda_k \bphi_{k}^{*T} \bSigma^{*-1} (\bX^*-\bmu^*) - \xi_{k}^*)^2\\
		&\lesssim
		\sumkK (  \be_k^{*T}\bSigma^{*-1} (\bX^*-\bmu^*))^2
		+
		\sumkK (  \bphi_{k}^{*T}\bW^* (\bY^*-\bmu^*)-\xi_{k}^*)^2\\
		&+
		\sumkK (  \bphi_{k}^{*T}\bW^*\beps^*)^2.
	\end{align*}
	Similar to the proof of Theorem \ref{thm:functionalShrinkage}, we have
	\begin{align*}
		\lvert  \bphi_{k}^{*T}\bW^* (\bY^*-\bmu^*)-\xi_{k}^* \rvert
		&\le
		\lambda_k\inv \left( \sumlm w_l^{*2} + (1-T^{*(m)})^2  + (1-T^{*(m)}) \right),
	\end{align*}
	where $T^{*(m)}=\max_{j=1,\dots,m^*} \bT_j^*$. This implies
	\begin{align*}
		E\left(\sumkK (  \bphi_{k}^{*T}\bW^* (\bY^*-\bmu^*)-\xi_{k}^*)^2\right)
		&=
		O(m^{*-2}).
	\end{align*}
	Also
	\begin{align*}
		E\left(\sumkK (  \bphi_{k}^{*T}\bW^*\beps^*)^2\right)
		&=
		O(m^{*-1}),
	\end{align*}
	and
	\begin{align*}
		E\left( (\be_k^{*T}\bSigma^{*-1} (\bX^*-\bmu^*))^2 \right)
		&=O(m^{*-1}).
	\end{align*}
	Therefore
	\begin{align}
		E\left(	\lVert  \tbxiKs - \bxi_{K}^* \rVert_2^2\right)
		&=
		O(m^{*-1}).\label{eq:predictiveShrinkage_FLM_key1}
	\end{align}
	Recall that $\cP_{K}^* \overset{d}{=} N(\beta_0 +\bbeta_K^T \tbxiKs, \bbeta_K^T \bSigma_{K}^* \bbeta_K)$. By construction of the $2$-Wasserstein distance,
	\begin{align*}
		\cW_2^2(\cP_{K}^*,\cA_{\beta_0 +\bbeta_K^T \bxi_K^*})
		&=
		(\bbeta_K^T (\tbxiKs - \bxi_{K}^*))^2+\bbeta_K^T \bSigma_{K}^* \bbeta_K\\
		&\le
		\lVert \bbeta_{K}\rVert_2^2 \lVert  \tbxiKs - \bxi_{K}^* \rVert_2^2 
		+ \lVert \bbeta_{K}\rVert_2^2 	\lVert \bSigma_{K}^*\rVert_{\text{op},2}\\
		&=O_p(m^{*-1}),
	\end{align*}
	where the last equality is due to \eqref{eq:predictiveShrinkage_FLM_key1} and using that $\lVert \bSigma_{K}^*\rVert_{\text{op},2}=O_p(m^{*-1})$, which follows analogously as in the proof of Theorem \ref{thm:xiEst2}. This shows the first result. Next,
	\begin{align}\label{eq:predictiveShrinkage_FLM_wassIneq1}
		\cW_2^2(\hcP_{K}^*,\cA_{\beta_0 +\bbeta_K^T \bxi_K^*})
		&= \hbbeta_K^T  \hbSigma_K^* \hbbeta_K+ ( \hbeta_0+\hbbeta_K^T \hbxiKs -\beta_0- \bbeta_K^T  \bxi_{K}^*)^2 \nonumber\\
		&\lesssim
		\lVert \hbbeta_K \rVert_2^2 \lVert \hbSigma_{K}^* -   \bSigma_K^* \rVert_{\text{op},2}
		+ 
		\lVert \hbbeta_K \rVert_2^2 \lVert \bSigma_K^* \rVert_{\text{op},2}
		+
		(\hbeta_0-\beta_0)^2 \nonumber \\
		&+
		\lVert \hbbeta_K \rVert_2^2 \lVert \hbxiKs-\bxi_{K}^*  \rVert_2^2
		+
		\lVert \hbbeta_K-\bbeta_K \rVert_2^2 \lVert \bxi_{K}^* \rVert_2^2 \nonumber\\
		&= O_p\left( m^{*2} (a_n+b_n)^2 + m^{*-1} + a_n+b_n +r_n^{*2}\right),
	\end{align}
	where the last equality is due to Theorem \ref{thm:xiEstHat}, Theorem \ref{thm:xiEstHat2}, the fact that $\lVert \bxi_{K}^* \rVert_2=O_p(1)$, $\lVert \bSigma_{K}^*\rVert_{\text{op},2}=O_p(m^{*-1})$, and using Lemma \ref{lem:Lemma_beta} with $h=n^{-1/3}$. The second result follows.
\end{proof}

\vspace{.5cm}

In the following, we say that a process $X$ is explained by its first $K$ principal components if $X(t)=\mu(t)+\sum_{k=1}^K \xi_k \phi_k(t)$ and thus  is of finite dimension $K$.

\begin{lem}\label{lem:SigmaKeigengap}
	Suppose that the process $X$ is finite dimensional and explained by its first $K=2$ principal components. If $\phi_1$ and $\phi_2$ are bijective and differentiable in a finite partition of $\cT$, then $\bSigmaiK$ has a  positive eigengap almost surely.
\end{lem}
\begin{proof}[Proof of Lemma \ref{lem:SigmaKeigengap}]
	Recalling that $\bSigmaiK = \bLambda_K - \bLambda_K \bPhiiK^T \bSigma_i^{-1} \bPhiiK\bLambda_K$ and since $K=2$, it follows that the characteristic polynomial of $\bSigmaiK$ is given by $p(\lambda)=\lambda^2-\tr(\bSigmaiK) \lambda +\det(\bSigmaiK)$, and thus the eigengap is equal to $\sqrt{\Delta_p}$, where $\Delta_p$ is the discriminant of the quadratic polynomial $p$. It is easy to show that
	\begin{align*}
		\Delta_p=(\lambda_1-\lambda_2+\lambda_2^2 \phi_{i2}^T \bSigma_i\inv \phi_{i2}  -\lambda_1^2 \phi_{i1}^T \bSigma_i\inv \phi_{i1} )^2 +4\lambda_1^2 \lambda_2^2 ( \phi_{i1}^T \bSigma_i\inv \phi_{i2})^2,
	\end{align*}
	so that it suffices to check that $\phi_{i1}^T \bSigma_i\inv \phi_{i2}$ is not identically zero almost surely. Let $B=\sigma^2 I_{n_i}+\lambda_1 \phi_{i1} \phi_{i1}^T$, where $I_{n_i}$ denotes the $n_i\times n_i$ identity matrix, and denote by  $\lVert \cdot\rVert_2$ the Euclidean norm in $\mathbb{R}^{n_i}$. By the Sherman-Morrison formula, it follows that $B\inv = \sigma^{-2} \left(I_{n_i}-  \frac{\lambda_1 \phi_{i1} \phi_{i1}^T}{\sigma^2+\lambda_1 \norm{\phi_{i1}}_2^2}\right)$, and a second application of the formula leads to
	\begin{align*}
		\bSigma_i\inv= B\inv - \frac{B\inv \lambda_2\phi_{i2}\phi_{i2}^T B\inv }{1+\lambda_2\phi_{i2}^T B\inv \phi_{i2}}.
	\end{align*}
	Thus
	\begin{align*}
		\phi_{i1}^T \bSigma_i\inv \phi_{i2} = \frac{ \phi_{i1}^T B\inv \phi_{i2} }{1+\lambda_2 \phi_{i2}^T B\inv \phi_{i2}},
	\end{align*}
	where $\phi_{i1}^T B\inv \phi_{i2} =\frac{\phi_{i1}^T \phi_{i2}}{ \sigma^2 + \lambda_1\norm{\phi_{i1}}_2^2 }$ and $\phi_{i2}^T B\inv \phi_{i2}>0$ a.s.  since the eigenvalues of $B$ are bounded below by $\sigma^2$. The conclusion then follows if we can show that $\phi_{i1}^T \phi_{i2}\neq0$ almost surely. Note that $\phi_{i1}^T \phi_{i2}=\sum_{j=1}^{n_i} \phi_1(T_{ij}) \phi_2(T_{ij})$ and the $T_{ij}$ are i.i.d.\ with a continuous distribution supported on $\cT$. Thus, the distribution of $\phi_{i1}^T \phi_{i2}$ corresponds to the $n$-fold convolution of the continuous distribution associated with $\phi_1(T_{i1})\phi_2(T_{i1})$, which is a continuous probability measure, and hence $\phi_{i1}^T \phi_{i2}\neq0$ holds almost surely.
\end{proof}

\begin{lem}\label{lem:unifGaps}
	Let $T_1,\dots,T_m$ be i.i.d.\ with density function $f(t)$, $t\in \mathcal{T}=[0,1]$ and let $T_{(1)},\dots,T_{(m)}$ be the order statistics. Let $w_l:=T_{(l)}-T_{(l-1)}$, $l=1,\dots,m$, where $T_{(0)}:=0$, be the spacing between the order statistics. Suppose that there exists $c_0>0$ such that $f(t)\ge c_0$ for all $t\in \mathcal{T}$. Then, for any integer $p\ge 1$ it holds that,
	\begin{equation*}
		E(w_l^p)=O(m^{-p}), \quad l=1,\dots,m,
	\end{equation*}
	and
	\begin{equation*}
		E[\left(1-T_{(m)}\right)^p]=O(m^{-p}).
	\end{equation*}
\end{lem}
\begin{proof}[Proof of Lemma \ref{lem:unifGaps}]
	One can replace $T_l$ with i.i.d.\ copies $Q(U_l)$, $l=1,\dots,m$, where the $U_l\overset{iid}{\sim} U(0,1)$ and $Q$ is the quantile function corresponding to $f$. Since $f$ is strictly positive, then $T_{(l)}=Q(U_{(l)})$, $l=1,\dots,m$. From a Taylor expansion of $Q(\cdot)$, we have
	\begin{gather*}
		E\left( w_l^p\right)= E[Q'(\eta_l) (U_{(l)}-U_{(l-1)}  )]^p\le c_0^{-p} E[U_{(l)}-U_{(l-1)}]^p,
	\end{gather*}
	where $\eta_l$ is between $U_{(l-1)}$ and $U_{(l)}$, and the last inequality follows from the fact that $Q'(t)=1/f(Q(t))\le c_0\inv$. The first result follows since $U_{(l)}-U_{(l-1)}\sim \text{Beta}(1,m)$ which implies $E[U_{(l)}-U_{(l-1)}]^p=O(m^{-p})$.
	Similarly, by expanding $Q(U_{(m)})$ around $Q(1)=1$ and since it can be verified that $E[(1-U_{(m)})^p]= m! p! /(m+p)!=O(m^{-p})$, the second result follows. 
\end{proof}

The next two lemmas are for establishing  Theorem \ref{thm:xiEstHat} and Theorem \ref{thm:functionalShrinkage_estimated}.

\begin{lem}\label{lem:secondMomentFPCscoresNorm}
	Suppose that assumptions \ref{a:Diff}, \ref{a:GaussProcess}, \ref{a:eigendecay} and \ref{a:K}--\ref{a:Ubeta} are satisfied. Consider either a sparse design setting when $n_i\le N_0<\infty$ or a dense design when $n_i=m\to\infty$, $i=1,\dots,n$. Set $a_n=a_{n1}$ and $b_n=b_{n1}$ for the sparse case, and $a_n=a_{n2}$ and $b_n=b_{n2}$ for the dense case. For a new independent subject $i^*$, suppose that $m^*=m^*(n)\to\infty$ is such that $m^* (a_n+b_n)=o(1)$ as $\ntoinf$. If $K=K(n)$ satisfies $(a_n+b_n) \sumkK \lambda_k\inv=o(1)$ as $\ntoinf$, then
	\begin{align*}
		&\lVert  \hbxiKs-\tbxiKs \rVert_2^2 =O_p(R_n^*),
	\end{align*}
	where
	\begin{align*}
		R_n^*&=
		m^* (a_n+b_n)^2 \sumkK \delta_k^{-2}\lambda_k^{-2}   + m^{*-1} \sumkK \lambda_k^{-2}\\
		&+
		m^{*2}(a_n+b_n)^2 \sumkK \lambda_k^{-2}
		+ m^{*4}(a_n+b_n)^4 \sumkK \delta_k^{-2} \lambda_k^{-2}.
	\end{align*}
\end{lem}
\begin{proof}[Proof of Lemma \ref{lem:secondMomentFPCscoresNorm}]
	Similarly as in the proof of Theorem \ref{thm:xiEst}, write 
	\begin{align}
		\hbe_k^* = \int_\cT \hGamma(\bT^*, t) \hphi_k(t) dt - \hbSigma^{*} \bW^* \hbphi_{k}^* \label{eq:hbe_k*definition}.
	\end{align}
	From Theorem $5.2$ in \cite{zhan:16}, we have
	\begin{align}
		\lVert \hGamma-\Gamma\rVert_{\infty}&=O(a_n+b_n) \quadas \label{eq:shrinkEstFunc_DeltaGamma},
	\end{align}
	as $\ntoinf$, which implies
	\begin{align}
		\lVert \hXi -\Xi \rVert_{\text{op}}&=O(a_n+b_n) \quadas \label{eq:shrinkEstFunc_opNorm},
	\end{align}
	as $\ntoinf$. This combined with perturbation results \citep{bosq:00} show that for any $k\ge 1$,
	\begin{align}
		\lVert  \hphi_k-\phi_k \rVert_{L^2}&\le 2\sqrt{2}\delta_k\inv \lVert \hXi -\Xi \rVert_{\text{op}}= O((a_n+b_n) \delta_k\inv) \quadas \label{eq:shrinkEstFunc_eigenvec}, 
	\end{align}
	and
	\begin{align}
		\lvert \hlambda_k - \lambda_k \rvert&\le  \lVert \hXi -\Xi \rVert_{\text{op}}=O(a_n+b_n) \quadas \label{eq:shrinkEstFunc_eigenvals},
	\end{align}
	as $\ntoinf$. Similar to the proof of Theorem $2$ in \cite{dai:16:1} and employing Theorem $5.1$ and $5.2$ in \cite{zhan:16}, it holds that
	\begin{align}
		\lVert\hbSigma^{*-1} -\bSigma^{*-1} \rVert_{\text{op},2}
		&\lesssim m^*
		(\lvert \hsigma^2 - \sigma^2\rvert + \lVert \hGamma-\Gamma\rVert_\infty)
		=O(m^*(a_n+b_n)) \quadas \label{eq:shrinkEstFunc_DeltaEstXi_key2},
	\end{align}
	as $\ntoinf$. Also note that for $1\le k\le K$,
	\begin{align}
		\lVert \bW^* \bphi_{k}^*\rVert_2
		&=
		O\left(\lambda_k\inv \left(\sum_{r=1}^{m^*} w_r^{*2}\right)^{1/2}\right)\label{eq:secondMomentFPCscoresNorm_boundWPhi}.
	\end{align}
	Similar arguments as in the proof of Theorem $2$ in \cite{mull:05:4} along with perturbation results \citep{bosq:00}, \eqref{eq:shrinkEstFunc_DeltaGamma}, and \eqref{eq:shrinkEstFunc_eigenvec} show that
	\begin{align}
		\sup_{t\in\cT}\  \lvert \hlambda_k \hphi_k(t) -\lambda_k\phi_k(t)\rvert
		&\le
		\lVert \hGamma-\Gamma\rVert_\infty + \lVert \Gamma\rVert_\infty  \lVert \hphi_k-\phi_k\rVert_{L^2}\nonumber\\
		&=O\left(
		(a_n+b_n) (1+ \delta_k\inv)\right)\quadas \label{eq:secondMomentFPCscoresNorm_deltaLambdaPhi},
	\end{align}
	as $\ntoinf$. By the Cauchy--Schwarz inequality and employing the orthonormality of the $\hphi_k$,
	\begin{align}
		\lvert \hlambda_k \hphi_k(t) \rvert
		&=  \Big\lvert \int_\cT \hGamma (t,s)\hphi_k(s) ds \Big\rvert
		\le
		\left( \int_\cT \hGamma^2 (t,s) ds \right)^{1/2}
		\le
		\lVert \hGamma \rVert_\infty\label{eq:secondMomentFPCscoresNorm_boundLambdaPhi}.
	\end{align}
	Since for large enough $n$ we have
	\begin{align*}
		\lambda_K\inv \lVert \hXi -\Xi \rVert_{\text{op}} 
		&\le
		\sumkK \lambda_k\inv \lVert \hXi -\Xi \rVert_{\text{op}}=O\left( (a_n+b_n) \sumkK \lambda_k\inv \right)=o(1)\quadas,
	\end{align*}
	where the first equality is due to \eqref{eq:shrinkEstFunc_opNorm} and the last is due to the condition $(a_n+b_n)\nu_K=o(1)$ as $\ntoinf$, we have $\lVert \hXi -\Xi \rVert_{\text{op}} \le \lambda_K/2 \le \lambda_k/2$ a.s. for large enough $n$. In view of \eqref{eq:shrinkEstFunc_eigenvals}, it follows that for any $1\le k\le K$,
	\begin{align}
		\lvert \hlambda_k -\lambda_k\rvert
		&\le \lambda_k/2 \quadas \label{eq:secondMomentFPCscoresNorm_deltaLambdaBoundLambda},
	\end{align}
	as $\ntoinf$. Combining with \eqref{eq:secondMomentFPCscoresNorm_boundLambdaPhi} and \eqref{eq:shrinkEstFunc_DeltaGamma} leads to
	\begin{align}
		\lVert \hphi_k\rVert_\infty
		&\le
		\hlambda_k\inv \lVert \hGamma \rVert_\infty
		\le 2 \lambda_k\inv (\lVert \hGamma -\Gamma \rVert_\infty +\lVert \Gamma \rVert_\infty)
		=O(\lambda_k\inv)\quadas \label{eq:secondMomentFPCscoresNorm_estimatedPhiBound},
	\end{align}
	for large enough $n$. This along with \eqref{eq:shrinkEstFunc_eigenvals}, \eqref{eq:secondMomentFPCscoresNorm_deltaLambdaPhi}, and \eqref{eq:secondMomentFPCscoresNorm_deltaLambdaBoundLambda} implies
	\begin{align}
		\sup_{t\in\cT} \ \lvert \hphi_k(t) -\phi_k(t) \rvert
		&\le
		\lVert \hphi_k\rVert_\infty \lambda_k\inv \lvert\hlambda_k-\lambda_k\rvert
		+\lambda_k\inv \lVert  \hlambda_k \hphi_k -\lambda_k\phi_k \rVert_\infty\nonumber\\
		&=O\left( (a_n+b_n) (\lambda_k^{-2} + \lambda_k\inv + \lambda_k\inv \delta_k\inv)
		\right) \quadas \label{eq:secondMomentFPCscoresNorm_supDeltaPhi},
	\end{align}
	as $\ntoinf$. Thus, using that $\delta_k\le \lambda_k$ we obtain
	\begin{align}
		\lVert \bW^* (\hbphi_{k}^* -\bphi_{k}^*)\rVert_2
		&=O\left(
		\left(\sum_{r=1}^{m^*} w_r^{*2}\right)^{1/2}  \lambda_k\inv (a_n+b_n) (1+ \delta_k\inv) \right)\quadas\label{eq:secondMomentFPCscoresNorm_WDeltaPhiBound},
	\end{align}
	as $\ntoinf$. Let $\bphi_{k}^*=\phi_{k}(\bT^*)$ and $\hbphi_{k}^*=\hphi_{k}(\bT^*)$. From \eqref{eq:hbe_k*definition}, note that
	\begin{align*}
		\lVert \hbe_k^*\rVert_2 &\le 
		\Big\lVert \int_\cT \hGamma(\bT^*,s)\hphi_k(s)ds - \hGamma(\bT^*,\bT^{*T}) \bW^* \hbphi_{k}^* \Big\rVert_2 + \lVert \hsigma^2 \bW^* \hbphi_{k}^*\rVert_2,
	\end{align*}
	where
	\begin{align}
		\lVert \hsigma^2 \bW^* \hbphi_{k}^*\rVert_2^2
		&\le
		(\lvert \hsigma^2 -\sigma^2\rvert + \sigma^2)^2
		\lVert \hphi_k \rVert_\infty^2 \sum_{l=1}^{m^*} w_l^{*2} 
		\lesssim
		\lambda_k^{-2} \sum_{l=1}^{m^*} w_l^{*2} \quadas \label{eq:secondMomentFPCscoresNorm_secondComponent},
	\end{align}
	for large enough $n$ and the last upper bound depends on $k$ only through $\lambda_k^{-2}$. Here the last inequality uses that $\lVert \hphi_k \rVert_\infty=O(\lambda_k\inv)$ a.s.\ and $\lvert \hsigma^2-\sigma^2\rvert=O(a_n+b_n)$ a.s.\ as $\ntoinf$. Observe
	\begin{align}
		\int_\cT \hGamma(\bT^*, s) \hphi_k(s) ds - \hGamma(\bT^*,\bT^{*T})\bW^*\hbphi_{k}^*
		&=
		\int_\cT \hGamma(\bT^*, s) \hphi_k(s) ds - \int_\cT \Gamma(\bT^*, s) \phi_k(s) ds \nonumber\\
		&+ \int_\cT \Gamma(\bT^*, s) \phi_k(s) ds - \Gamma(\bT^*,\bT^{*T})\bW^* \bphi_{k}^* \nonumber\\
		&+ \Gamma(\bT^*,\bT^{*T})\bW^* \bphi_{k}^*- \hGamma(\bT^*,\bT^{*T})\bW^*\hbphi_{k}^*\label{eq:secondMomentFPCscoresNorm_keyDecomposition},
	\end{align}
	Hence, it suffices to control each of the differences in \eqref{eq:secondMomentFPCscoresNorm_keyDecomposition}. First,
	\begin{align*}
		&\int_\cT \hGamma(\bT^*, s) \hphi_k(s)  - \Gamma(\bT^*, s) \phi_k(s) ds  \\
		&= 
		\int_\cT (\hGamma(\bT^*, s)-\Gamma(\bT^*, s)  )\hphi_k(s) ds + \int_\cT \Gamma(\bT^*, s) (\hphi_k(s) -\phi_k(s)) ds,
	\end{align*}
	where, for $j=1,\dots,m^*$, and by using the orthonormality of the $\hphi_k$,
	\begin{align*}
		\left\lvert \int_\cT (\hGamma(T_{j}^*, s)-\Gamma(T_{j}^*, s)  )\hphi_k(s) ds \right\rvert 
		&\le \left( \int_\cT (\hGamma(T_{j}^*, s)-\Gamma(T_{j}^*, s))^2 ds \right)^{1/2} \\
		&\le \lVert \hGamma - \Gamma\rVert_\infty\\
		&=O(a_n+b_n) \quadas,
	\end{align*}
	and
	\begin{align*}
		\left\lvert \int_\cT \Gamma(T_{j}^*, s) (\hphi_k(s) -\phi_k(s)) ds \right\rvert&\le  \lVert \Gamma\rVert_\infty \lVert \hphi_k -\phi_k\rVert_{L^2}
		=O\left( (a_n+b_n) \delta_k\inv \right) \quadas,
	\end{align*}
	where we use that $|\cT|=1$ and $\Gamma(s,t)$ is continuous over the compact set $\cT^2$. Thus
	\begin{align}
		\normtwo{\int_\cT \hGamma(\bT^*, s) \hphi_k(s)  - \Gamma(\bT^*, s) \phi_k(s) ds}  &= O\left(\sqrt{m^*}(a_n+b_n)(1+\delta_k\inv) \right) \quadas \label{eq:secondMomentFPCscoresNorm_keyDecomposition_firstDiff},
	\end{align}
	as $n\to\infty$, and the bound depends on $k$ only through $\delta_k\inv$. Second, from the Riemann sum approximation in \eqref{eq:intErrBound} and noting that the application $g_j(t)=\Gamma(T_j^*,t)\phi_k(t)$ satisfies $\lVert g_j\rVert_\infty = O(\lambda_k\inv)$ and $\lVert g_j'\rVert_\infty =O(\lambda_k\inv)$ by \ref{a:GammaDiff},  where the $O(\lambda_k\inv)$ terms are uniform in $j$ and depend on $k$ only through $\lambda_k\inv$, we have
	\begin{align*}
		&\left\lvert \int_\cT \Gamma(T_j^*, s) \phi_k(s)ds  - \Gamma(T_j^*,\bT^{*T}) \bW^* \bphi_{k}^* \right\rvert\\ &\lesssim
		\lambda_k\inv \left( \sum_{l=1}^{m^*} w_l^{*2} + (1-\bT^{(m^*)})^2 + (1-\bT^{(m^*)})   \right),
	\end{align*}
	where $\bT^{(m^*)}:=\max_{j=1,\dots,m^*} T_j^*$ and the upper bound is uniform in $j$ and depends on $k$ only through $\lambda_k\inv$. Thus
	\begin{align}
		&\normtwo{\int_\cT \Gamma(\bT^*, s) \phi_k(s) ds - \Gamma(\bT^*,\bT^{*T})\bW^* \bphi_{k}^*}  \nonumber\\
		&= 
		O\left( \sqrt{m^*}  	\lambda_k\inv \left( \sum_{l=1}^{m^*} w_l^{*2} + (1-\bT^{(m^*)})^2 + (1-\bT^{(m^*)})   \right)  \right)\label{eq:secondMomentFPCscoresNorm_keyDecomposition_secondDiff}.
	\end{align}
	Third, observe
	\begin{align}
		\Gamma(\bT^*,\bT^{*T})\bW^* \bphi_{k}^*- \hGamma(\bT^*,\bT^{*T})\bW^* \hbphi_{k}^* &=(\Gamma(\bT^*,\bT^{*T})-\hGamma(\bT^*,\bT^{*T}))\bW^* \bphi_{k}^* \nonumber \\
		&+  \hGamma(\bT^*,\bT^{*T}) (\bW^* \bphi_{k}^*- \bW^* \hbphi_{k}^*)\label{eq:secondMomentFPCscoresNorm_keyDecomposition_thirdDiff_key0}.
	\end{align}
	Note that
	\begin{align*}
		&\normtwo{(\Gamma(\bT^*,\bT^{*T})-\hGamma(\bT^*,\bT^{*T}))\bW^* \bphi_{k}^*} \\
		&\le \norm{\Gamma(\bT^*,\bT^{*T})-\hGamma(\bT^*,\bT^{*T})}_{\text{op},2} \normtwo{\bW^* \bphi_{k}^*} \\
		&\lesssim
		\lambda_k\inv  \left(\sum_{l=1}^{m^*} w_l^{*2} \right)^{1/2}
		\norm{\Gamma(\bT^*,\bT^{*T})-\hGamma(\bT^*,\bT^{*T})}_{\text{op},2},
	\end{align*}
	where the last equality follows similarly as in \eqref{eq:Wphi} and using that $\lVert \phi_k \rVert_\infty=O(\lambda_k\inv)$. Since $\norm{A}_{\text{op},2} \le \norm{A}_{F}$, where $\norm{A}_{F}$ denotes the Frobenius norm of a squared matrix $A$, and
	\begin{align*}
		\norm{\Gamma(\bT^*,\bT^{*T})-\hGamma(\bT^*,\bT^{*T})}_{F}^2 
		&\le 
		m^{*2}  \sup_{s,t\in\cT} |\hGamma(s,t) - \Gamma(s,t) |^2
		= O (m^{*2}(a_n+b_n)^2)\quadas,
	\end{align*}
	as $n\to\infty$, it follows that 
	\begin{align}
		\normtwo{(\Gamma(\bT^*,\bT^{*T})-\hGamma(\bT^*,\bT^{*T}))\bW^* \bphi_{k}^* } &\lesssim
		\lambda_k\inv  m^* (a_n+b_n) \left(\sum_{l=1}^{m^*} w_l^{*2} \right)^{1/2} \quadas \label{eq:secondMomentFPCscoresNorm_keyDecomposition_thirdDiff_key1},
	\end{align}
	as $\ntoinf$. Also,
	\begin{align*}
		&\normtwo{\hGamma(\bT^*,\bT^{*T}) ( \bW^* \bphi_{k}^*- \bW^* \hbphi_{k}^* )}\nonumber\\
		&=\left(\norm{
			\hGamma(\bT^*,\bT^{*T})-\Gamma(\bT^*,\bT^{*T})}_{\text{op},2}+\norm{\Gamma(\bT^*,\bT^{*T})}_{\text{op},2} \right) \normtwo{\bW^* (\bphi_{k}^*-\hbphi_{k}^*)} \nonumber\\
		&\lesssim  (m^*(a_n+b_n) + m^*)
		\lambda_k\inv (a_n+b_n) (1+ \delta_k\inv) \left(\sum_{l=1}^{m^*} w_l^{*2} \right)^{1/2} \quadas
		\nonumber\\
		&\lesssim
		m^* \lambda_k\inv (a_n+b_n) (1+ \delta_k\inv) \left(\sum_{l=1}^{m^*} w_l^{*2} \right)^{1/2}\quadas,
	\end{align*}
	as $\ntoinf$, where the first inequality follows from \eqref{eq:secondMomentFPCscoresNorm_WDeltaPhiBound} and the last inequality uses the condition $m^*(a_n+b_n)=o(1)$ as $\ntoinf$. This along with \eqref{eq:secondMomentFPCscoresNorm_keyDecomposition_thirdDiff_key0} and \eqref{eq:secondMomentFPCscoresNorm_keyDecomposition_thirdDiff_key1} implies
	\begin{align}
		\normtwo{\Gamma(\bT^*,\bT^{*T}) \bW^* \bphi_{k}^* - \hGamma(\bT^*,\bT^{*T})\bW^* \hbphi_{k}^*}
		&\lesssim
		m^* \lambda_k\inv (a_n+b_n) (1+ \delta_k\inv) \left(\sum_{l=1}^{m^*} w_l^{*2} \right)^{1/2}\label{eq:secondMomentFPCscoresNorm_keyDecomposition_thirdDiff_key0_rate},
	\end{align}
	almost surely as $\ntoinf$, where the bound depends on $k$ only through $\lambda_k\inv$ and $\delta_k\inv$. Combining \eqref{eq:secondMomentFPCscoresNorm_keyDecomposition}, \eqref{eq:secondMomentFPCscoresNorm_keyDecomposition_firstDiff}, \eqref{eq:secondMomentFPCscoresNorm_keyDecomposition_secondDiff}, and \eqref{eq:secondMomentFPCscoresNorm_keyDecomposition_thirdDiff_key0_rate} leads to
	\begin{align}
		&\normtwo{\int_\cT \hGamma(\bT^*, t) \hphi_k(t) dt - \hGamma(\bT^*,\bT^{*T})\bW^*\hbphi_{k}^*}\nonumber\\
		&\lesssim
		\sqrt{m^*}(a_n+b_n)(1+\delta_k\inv)\nonumber\\
		&+
		\sqrt{m^*}  	\lambda_k\inv \left( \sum_{l=1}^{m^*} w_l^{*2} + (1-\bT^{(m^*)})^2 + (1-\bT^{(m^*)})\right)\nonumber\\
		&+
		m^* \lambda_k\inv (a_n+b_n) (1+ \delta_k\inv) \left(\sum_{l=1}^{m^*} w_l^{*2} \right)^{1/2}\quadas\label{eq:secondMomentFPCscoresNorm_keyDecomposition_rate},
	\end{align}
	as $\ntoinf$. This along with \eqref{eq:secondMomentFPCscoresNorm_secondComponent} implies
	\begin{align}
		\lVert \hbe_k^*\rVert_2 &\lesssim
		\sqrt{m^*}(a_n+b_n)(1+\delta_k\inv)
		+
		\sqrt{m^*}  	\lambda_k\inv \left( \sum_{l=1}^{m^*} w_l^{*2} + (1-\bT^{(m^*)})^2 + (1-\bT^{(m^*)})\right) \nonumber\\
		&+
		m^* \lambda_k\inv (a_n+b_n) (1+ \delta_k\inv) \left(\sum_{l=1}^{m^*} w_l^{*2} \right)^{1/2}
		+
		\lambda_k^{-1} \left( \sum_{l=1}^{m^*} w_l^{*2} \right)^{1/2} \quadas
		\label{eq:secondMomentFPCscoresNorm_keyDecomposition_bound},
	\end{align}
	as $\ntoinf$, where the bound depends on $k$ only through $\lambda_k\inv$ and $\delta_k\inv$. Define auxiliary quantities $Z_{m^*,n,K}:=\sumkK [\hbe_k^{*T} \hbSigma^{*-1} (\bX^*-\hat{\bmu}^*)]^2$, $\tilde{Z}_{m^*,n,K}:=\sumkK [\hbe_k^{*T} \hbSigma^{*-1} (\bX^*-\bmu^*)]^2$, $\mu_{m^*,n,K}:=\sumkK [\hbe_k^{*T}\hbSigma^{*-1}(\bmu^*-\hat{\bmu}^*)]^2$, and observe
	\begin{align}
		Z_{m^*,n,K} 
		&\lesssim 
		\tilde{Z}_{m^*,n,K} +\mu_{m^*,n,K}\label{eq:secondMomentFPCscoresNorm_keyDecomposition_decomp}.
	\end{align}
	By independence of the new subject's observations from the estimated population quantities, we have
	\begin{align}
		E[Z_{m^*,n,K} \vert \bT^*,\hGamma, \hphi_k, \hsigma,\hmu ]
		&\lesssim
		E\Big[ \tilde{Z}_{m^*,n,K} \vert \bT^*,\hGamma, \hphi_k, \hsigma,\hmu \Big]
		+\mu_{m^*,n,K}\nonumber\\
		&=
		\sumkK \hbe_k^{*T} \hbSigma^{*-1} \bSigma^* \hbSigma^{*-1} \hbe_k^*
		+\mu_{m^*,n,K}  \quadas
		\label{eq:secondMomentFPCscoresNorm_keyDecomposition_decomp2},
	\end{align}
	and for large enough $n$
	\begin{align}
		& \Big\lvert \sumkK \hbe_k^{*T} \hbSigma^{*-1} \bSigma^* \hbSigma^{*-1} \hbe_k^*  \Big\rvert \nonumber \\
		&\le  
		\Big\lvert  \sumkK [
		\hbe_k^{*T} (\hbSigma^{*-1} -\bSigma^{*-1}) \bSigma^* (\hbSigma^{*-1} -\bSigma^{*-1}) \hbe_k^* +2 \hbe_k^{*T} (\hbSigma^{*-1} -\bSigma^{*-1})\hbe_k^*+ \hbe_k^{*T} \bSigma^{*-1} \hbe_k^*] \Big\rvert \nonumber\\
		&\lesssim
		\sumkK [
		m^{*3} (a_n+b_n)^2  \normtwo{\hbe_k^*}^2 + m^*(a_n+b_n)  \normtwo{\hbe_k^*}^2 + \normtwo{\hbe_k^*}^2 ] \quadas\nonumber\\
		&\lesssim (1+m^{*3} (a_n+b_n)^2)
		\sumkK 
		\normtwo{\hbe_k^*}^2 \quadas\nonumber\\
		&\lesssim
		[m^* (a_n+b_n)^2 \left(\sumkK \delta_k^{-2} \right)
		+
		m^* \left(\sumkK \lambda_k^{-2} \right) \left( \sum_{l=1}^{m^*} w_l^{*2} + (1-\bT^{(m^*)})^2 + (1-\bT^{(m^*)})\right)^2\nonumber\\
		&+
		m^{*2}  (a_n+b_n)^2  \left(\sumkK \lambda_k^{-2} \delta_k^{-2} \right) \sum_{l=1}^{m^*} w_l^{*2}
		+  \left(\sumkK \lambda_k^{-2} \right) \sum_{l=1}^{m^*} w_l^{*2}
		] (1+m^{*3} (a_n+b_n)^2) \nonumber\\
		&= (1+m^{*3} (a_n+b_n)^2)\nonumber\\
		& O_p\left(  
		m^* (a_n+b_n)^2 \left(\sumkK \delta_k^{-2} \right)
		+
		m^{*-1} \left(\sumkK \lambda_k^{-2} \right)
		+
		m^{*}  (a_n+b_n)^2  \left(\sumkK \lambda_k^{-2} \delta_k^{-2} \right)
		\right)\nonumber\\
		&=
		(1+m^{*3} (a_n+b_n)^2)
		O_p\left(
		m^{*-1}\left(\sumkK \lambda_k^{-2} \right)
		+
		m^{*}  (a_n+b_n)^2  \left(\sumkK \lambda_k^{-2} \delta_k^{-2} \right)
		\right)\nonumber\\
		&=O_p(R_n^*) \label{eq:secondMomentFPCscoresNorm_keyDecomposition_bound_2_1},
	\end{align}
	where the second inequality is due to $\lVert \hbSigma^{*-1} -\bSigma^{*-1}\rVert_{\text{op},2}=O(m^* (a_n+b_n))$ a.s. as $\ntoinf$, $\lVert \bSigma^{*-1} \rVert_{\text{op},2}\le \sigma^{-2}$, $\lVert \bSigma^* \rVert_{\text{op},2}=O(m^*)$, and the fourth inequality follows from \eqref{eq:secondMomentFPCscoresNorm_keyDecomposition_bound}. This shows that
	\begin{align*}
		\sumkK  \hbe_k^{*T} \hbSigma^{*-1} \bSigma^* \hbSigma^{*-1} \hbe_k^*
		&= O_p(R_n^*).
	\end{align*}
	Thus, for any $\epsilon>0$ there exists $N_0=N_0(\epsilon)\ge 1$ and $M_0=M_0(\epsilon)> 0$ such that for all $n\ge N_0$
	\begin{align}
		P\left( R_n^{*-1}  \Big\lvert \sumkK \hbe_k^{*T} \hbSigma^{*-1} \bSigma^* \hbSigma^{*-1} \hbe_k^*\Big\rvert> M_0  \right)\le \epsilon \label{eq:secondMomentFPCscoresNorm_keyDecomposition_bound_3}.
	\end{align}
	Let $M>0$ and define
	\begin{align*}
		u_{m^*,n,K}&=P\left( R_n^{*-1} \tilde{Z}_{m^*,n,K} >M  \vert \bT^*,\hGamma, \hphi_k, \hsigma,\hmu \right).
	\end{align*}
	Choosing $M=M(\epsilon)=M_0/\epsilon$ and using that $u_{m^*,n,K}\le 1$ along with the relation
	\begin{align*}
		u_{m^*,n,K}
		&\lesssim
		\frac{1}{ R_n^{*} M} \sumkK  \hbe_k^{*T} \hbSigma^{*-1} \bSigma^* \hbSigma^{*-1} \hbe_k^*,
	\end{align*}
	which follows analogously as in \eqref{eq:secondMomentFPCscoresNorm_keyDecomposition_decomp2}, leads to
	\begin{align*}
		P\left( R_n^{*-1} \tilde{Z}_{m^*,n,K} >M \right)
		&=
		E(u_{m^*,n,K} 1_{\{u_{m^*,n,K}\le \epsilon\}} + u_{m^*,n,K} 1_{\{u_{m^*,n,K}> \epsilon\}})\\
		&\le
		\epsilon + P(u_{m^*,n,K}> \epsilon)\\
		&\le 2\epsilon,
	\end{align*}
	where the last inequality follows from \eqref{eq:secondMomentFPCscoresNorm_keyDecomposition_bound_3}. Therefore
	\begin{align*}
		\tilde{Z}_{m^*,n,K}&= O_p(R_n^{*} ).
	\end{align*}
	Also, for large enough $n$ and using \eqref{eq:secondMomentFPCscoresNorm_keyDecomposition_bound} along with $\lVert \hbmu^*-\bmu^* \rVert_2^2=O(m^* (a_n+b_n)^2)$ a.s., we obtain
	\begin{align*}
		\mu_{m^*,n,K} 
		&\lesssim
		m^{*} (a_n+b_n)^2 \sumkK \normtwo{\hbe_k^*}^2 \quadas ,
	\end{align*}
	which in view of the third inequality in \eqref{eq:secondMomentFPCscoresNorm_keyDecomposition_bound_2_1} and the condition $m^*(a_n+b_n)=o(1)$ as $\ntoinf$ is of slower order compared to the rate $O_p(R_n^{*})$. These along with \eqref{eq:secondMomentFPCscoresNorm_keyDecomposition_decomp} leads to
	\begin{align}
		Z_{m^*,n,K}&=O_p(R_n^{*} ) \label{eq:secondMomentFPCscoresNorm_decompositionDiffScores_key1}.
	\end{align}
	Then a conditioning argument leads to 
	\begin{align*}
		E[(\be_k^{*T} \bSigma^{*-1} (\bX^*-\bmu^*))^2]
		&=
		E(E[(\be_k^{*T} \bSigma^{*-1} (\bX^*-\bmu^*))^2\vert \bT^*])\\
		&=
		E(\Var[\be_k^{*T} \bSigma^{*-1} (\bX^*-\bmu^*)\vert \bT^*])\\
		&=
		E(\be_k^{*T} \bSigma^{*-1}\be_k^{*}  )\\
		&\le \sigma^{-2} E(\lVert \be_k^*\rVert_2^2 )\\
		&\lesssim m^{*-1} \lambda_k^{-2},
	\end{align*}
	where the last inequality holds for large enough $n$ and follows analogously as in  \eqref{eq:shrinkageFuncSubject_expectation_ek}. This implies
	\begin{align*}
		E\Big[\sumkK (\be_k^{*T} \bSigma^{*-1} (\bX^*-\bmu^*))^2\Big]
		&\le m^{*-1} \sumkK \lambda_k^{-2}.
	\end{align*}
	Hence
	\begin{align}
		\sumkK(\be_k^{*T} \bSigma^{*-1} (\bX^*-\bmu^*))^2
		&=
		O_p\left(  m^{*-1} \sumkK \lambda_k^{-2} \right)
		\label{eq:secondMomentFPCscoresNorm_decompositionDiffScores_key2}.
	\end{align}
	For any $k=1,\dots,K$, observe
	\begin{align}
		&\hxi_{k}^*-\txi_{k}^* \nonumber\\
		&= 
		\hbe_k^{*T} \hbSigma^{*-1} (\bX^*-\hbmu^*)
		+ \hbphi_{k}^{*T} \bW^* (\bX^*-\hbmu^*)
		-
		\be_k^{*T} \bSigma^{*-1} (\bX^*-\bmu^*)
		-
		\bphi_{k}^{*T} \bW^* (\bX^*-\bmu^*)\label{eq:secondMomentFPCscoresNorm_decompositionDiffScores}.
	\end{align}
	From \eqref{eq:secondMomentFPCscoresNorm_boundWPhi}, \eqref{eq:secondMomentFPCscoresNorm_WDeltaPhiBound},  and using that $\lVert \bX^*-\bmu^* \rVert_2^2=O_p(m^*)$ and $\lVert \hbmu^*-\bmu^* \rVert_2^2=O_p(m^* (a_n+b_n)^2)$, we obtain
	\begin{align}
		&\sumkK [(\hbphi_{k}^{*T} \bW^* (\bX^*-\hbmu^*)-\bphi_{k}^{*T} \bW^* (\bX^*-\bmu^*))^2]\nonumber\\
		&\lesssim
		\sumkK [\lVert \bW^* (\hbphi_{k}^*-\bphi_{k}^*)\rVert_2^2 
		\lVert \bX^*-\hbmu^* \rVert_2^2
		+\lVert \bW^* \bphi_{k}^*\rVert_2^2 
		\lVert \hbmu^*-\bmu^* \rVert_2^2]\nonumber\\
		&=
		O_p\left(
		(a_n+b_n)^2 \sumkK \lambda_k^{-2}\delta_k^{-2}
		\right) \label{eq:secondMomentFPCscoresNorm_decompositionDiffScores_key3}.
	\end{align}
	Combining \eqref{eq:secondMomentFPCscoresNorm_decompositionDiffScores_key1}, \eqref{eq:secondMomentFPCscoresNorm_decompositionDiffScores_key2}, \eqref{eq:secondMomentFPCscoresNorm_decompositionDiffScores}, and \eqref{eq:secondMomentFPCscoresNorm_decompositionDiffScores_key3} leads to
	\begin{align*}
		\lVert  \hbxiKs-\tbxiKs \rVert_2^2
		&=
		\sumkK [(	\hxi_{k}^*-\txi_{k}^*)^2]
		=O_p(R_n^*),
	\end{align*}
	which shows the result.
\end{proof}

\begin{lem}\label{lem:traceSigmaK}
	Suppose that assumptions \ref{a:Diff}, \ref{a:GaussProcess}, \ref{a:eigendecay} and \ref{a:K}--\ref{a:Ubeta} are satisfied. Consider either a sparse design setting when $n_i\le N_0<\infty$ or a dense design when $n_i=m\to\infty$, $i=1,\dots,n$. Set $a_n=a_{n1}$ and $b_n=b_{n1}$ for the sparse case, and $a_n=a_{n2}$ and $b_n=b_{n2}$ for the dense case. Let $\upsilon_K=\sumkK \lambda_k^{-1/2}\delta_k\inv$. For a new independent subject $i^*$, suppose that $m^*=m^*(n)\to\infty$ is such that $m^{*}(a_n+b_n)=o(1)$ and $K=K(n)$ satisfies $(a_n+b_n) \upsilon_K=o(1)$ as $\ntoinf$. Then
	\begin{align*}
		\text{trace}( \hbSigma_{K}^*-\bSigma_{K}^*)
		&=
		O_p\left( m^{*} (a_n+b_n)^2 \sumkK \lambda_k^{-2} \delta_k^{-2}
		+
		(a_n+b_n) \sumkK \lambda_k^{-2} \delta_k\inv
		\right).
	\end{align*}
\end{lem}
\begin{proof}[Proof of Lemma \ref{lem:traceSigmaK}]
	In effect, for $j=1,\dots,K$, the $(j,j)$-element of $\hbSigma_{K}^*-\bSigma_{K}^*$ is given by
	\begin{align}
		[\hbSigma_{K}^*-\bSigma_{K}^*]_{j,j}
		&=
		\hbejsT \hbSigma^{*-1} \hbejs+\hbejsT \bW^* \hbphi_{j}^*+\hbphi_{j}^{*T} \bW^* \hbejs +\hbphi_{j}^{*T} \bW^* \hbSigma^* \bW^* \hbphi_{j}^*\nonumber\\
		&-
		(\bejsT \bSigma^{*-1} \bejs+\bejsT \bW^* \bphi_{j}^*+\bphi_{j}^{*T} \bW^* \bejs +\bphi_{j}^{*T} \bW^* \bSigma^* \bW^* \bphi_{j}^*)\label{eq:traceSigmaK_keyDecomposition},
	\end{align}
	where $\hbejs$ is defined as in \eqref{eq:hbe_k*definition}. Note that the conditions of Lemma \ref{lem:secondMomentFPCscoresNorm} hold since $(a_n+b_n) \nu_K=o(1)$ which is due to $\nu_K\le \upsilon_K$ and $\delta_k\le\lambda_k$, where $\nu_K=\sumkK \lambda_k\inv$. Observing for any $k=1,\dots,K$,
	\begin{align*}
		\delta_k\inv 
		&\le
		\sumkK\delta_k\inv = \sumkK \lambda_k^{-1/2} \delta_k\inv \lambda_k^{1/2}
		\le  \lambda_1^{1/2} \sumkK \lambda_k^{-1/2} \delta_k\inv,
	\end{align*}
	along with the condition $\upsilon_K (a_n+b_n)=o(1)$ as $\ntoinf$ leads to
	\begin{align*}
		\delta_k\inv (a_n+b_n) &\le \lambda_1^{1/2} (a_n+b_n) \sumkK \lambda_k^{-1/2} \delta_k\inv =o(1),
	\end{align*}
	as $\ntoinf$, where the bound is uniform in $k$. This along with \eqref{eq:secondMomentFPCscoresNorm_boundWPhi} and \eqref{eq:secondMomentFPCscoresNorm_WDeltaPhiBound} imply
	\begin{align}
		\lVert \bW^* \hbphi_{k}^* \rVert_2 &\le \lVert \bW^* (\hbphi_{k}^*-\bphi_{k}^*)\rVert_2 + 	\lVert \bW^* \bphi_{k}^*\rVert_2
		= O\left(\lambda_k\inv \left(\sum_{r=1}^{m^*} w_r^{*2}\right)^{1/2}\right)\quadas\label{eq:traceSigmaK_WPhihat},
	\end{align}
	as $\ntoinf$, where the bound depends on $k$ only through $\lambda_k\inv$. Also, using \eqref{eq:shrinkEstFunc_DeltaGamma} and since $m^*(a_n+b_n)=o(1)$ and $\lvert\hsigma^2-\sigma^2\rvert=O(a_n+b_n)$ as $\ntoinf$, which follows from Proposition $1$ in \cite{dai:16:1}, we obtain
	\begin{align}
		\lVert \hbSigma^* \rVert_{\text{op},2}&\le \hsigma^2 +	\lVert \hGamma(\bT^*,\bT^*)\rVert_{\text{op},2}
		\le
		\hsigma^2 + m^* \lVert \hGamma\rVert_\infty = O(m^*)\quadas\label{eq:traceSigmaK_normSigmaHat},
	\end{align}
	as $\ntoinf$. This along with \eqref{eq:secondMomentFPCscoresNorm_WDeltaPhiBound}, \eqref{eq:traceSigmaK_WPhihat}, and \eqref{eq:traceSigmaK_normSigmaHat} leads to
	\begin{align}
		\lvert (\hbphi_{j}^*-\bphi_{j}^*)^T \bW^* \hbSigma^* \bW^* \hbphi_{j}^* \rvert
		&\le
		\lVert \bW^* (\hbphi_{j}^*-\bphi_{j}^*)  \rVert_2 \lVert \hbSigma^* \bW^* \hbphi_{j}^* \rVert_2 \nonumber\\
		&\le
		\lVert \bW^* (\hbphi_{j}^*-\bphi_{j}^*)  \rVert_2 	\lVert \hbSigma^* \rVert_{\text{op},2} \lVert \bW^* \hbphi_{j}^* \rVert_2\nonumber \\
		&=O\left(
		\left(\sum_{r=1}^{m^*} w_r^{*2} \right) m^* (a_n+b_n)  \lambda_j^{-2}  (1+\delta_j\inv)
		\right)\quadas\label{eq:traceSigmaK_DeltaPhiWSWphihat},
	\end{align}
	as $\ntoinf$, where the bound depends on $j$ only through $\lambda_j\inv$ and $\delta_j\inv$. Using the fact that $\lVert \hbSigma^* -\bSigma^* \rVert_{\text{op},2}=O(m(a_n+b_n))$ a.s. as $\ntoinf$ along with \eqref{eq:secondMomentFPCscoresNorm_WDeltaPhiBound} and  \eqref{eq:traceSigmaK_WPhihat}, we obtain
	\begin{align}
		\lVert \hbSigma^* \bW^* \hbphi_{j}^* - \bSigma^* \bW^* \bphi_{j}^* \rVert_2
		&\le
		\lVert \hbSigma^* -\bSigma^* \rVert_{\text{op},2}  \lVert \bW^* \hbphi_{j}^* \rVert_2
		+
		\lVert \bSigma^* \rVert_{\text{op},2}  \lVert \bW^* (\hbphi_{j}^*-\bphi_{j}^*) \rVert_2\nonumber\\
		&=
		O\left( m^*(a_n+b_n)  \lambda_j\inv \left(\sum_{r=1}^{m^*} w_r^{*2}\right)^{1/2} (1+ \delta_j\inv)  \right)\quadas\label{eq:traceSigmaK_deltaSWPhihat},
	\end{align}
	as $\ntoinf$, where the bound depends on $j$ only through $\lambda_j\inv$ and $\delta_j\inv$. Thus
	\begin{align}
		\lvert \bphi_{j}^{*T} \bW^* (\hbSigma^* \bW^* \hbphi_{j}^* - \bSigma^* \bW^* \bphi_{j}^*)\rvert
		&\le
		\lVert \bW^*\bphi_{j}^*\rVert_2 \lVert \hbSigma^* \bW^* \hbphi_{j}^* - \bSigma^* \bW^* \bphi_{j}^*\rVert_2\nonumber\\
		&=O\left( m^*(a_n+b_n)  \lambda_j^{-2} \left(\sum_{r=1}^{m^*} w_r^{*2} \right) (1+ \delta_j\inv)  \right)\quadas\label{eq:traceSigmaK_PhiWdeltaSWPhihat},
	\end{align}
	as $\ntoinf$, where the bound depends on $j$ only through $\lambda_j^{-2}$ and $\delta_j\inv$. This combined with \eqref{eq:traceSigmaK_DeltaPhiWSWphihat} leads to
	\begin{align}
		&\lvert \hbphi_{j}^{*T} \bW^* \hbSigma^* \bW^* \hbphi_{j}^* -
		\bphi_{j}^{*T} \bW^* \bSigma^* \bW^* \bphi_{j}^*\rvert \nonumber\\
		&\le
		\lvert(\hbphi_{j}^*-\bphi_{j}^*)^T \bW^* \hbSigma^* \bW^* \hbphi_{j}^* \rvert 
		+
		\lvert \bphi_{j}^{*T} \bW^* (\hbSigma^* \bW^* \hbphi_{j}^* - \bSigma^* \bW^* \bphi_{j}^*) \rvert\nonumber\\
		&=O\left( m^*(a_n+b_n)  \lambda_j^{-2} \left(\sum_{r=1}^{m^*} w_r^2\right) (1+ \delta_j\inv)  \right)\quadas\label{eq:traceSigmaK_DeltaPhihatWSWPhihat},
	\end{align}
	as $\ntoinf$, where the bound depends on $j$ only through $\lambda_j^{-2}$ and $\delta_j\inv$. Write $\hDelta_k=\hbe_k^*-\be_k^*$, $k=1,\dots,K$, and observe
	\begin{align*}
		\lVert \hlambda_j\hbphi_{j}^* -\lambda_j\bphi_{j}^*\rVert_2
		&\le
		m^{*1/2} \lVert \hlambda_j\hphi_{j} -\lambda_j\phi_{j}\rVert_\infty
		=O\left(m^{*1/2} (a_n+b_n) (1+\delta_j\inv) \right)\quadas ,
	\end{align*}
	as $\ntoinf$, where the last equality is due to \eqref{eq:secondMomentFPCscoresNorm_deltaLambdaPhi} and the bound depends on $j$ only through $\delta_j\inv$. This along with \eqref{eq:traceSigmaK_deltaSWPhihat} leads to
	\begin{align}
		\lVert \hDelta_j \rVert_2
		&= 
		\Big \lVert \int_\cT \hGamma(\bT^*,s) \hphi_{j}(s)-\Gamma(\bT^*,s) \phi_{j}(s) ds + \bSigma^* \bW^* \bphi_{j}^*- \hbSigma^* \bW^* \hbphi_{j}^* \Big\rVert_2\nonumber\\
		&\le
		\lVert \hlambda_j\hbphi_{j}^* -\lambda_j\bphi_{j}^*\rVert_2+ \lVert \hbSigma^* \bW^* \hbphi_{j}^*-\bSigma^* \bW^* \bphi_{j}^* \rVert_2\nonumber\\
		&=
		O\left(m^{*1/2} (a_n+b_n) (1+ \delta_j\inv )  \left(1+m^{*1/2}\lambda_j\inv \left(\sum_{r=1}^{m^*} w_r^{*2} \right)^{1/2} \right) \right)
		\quadas\label{eq:traceSigmaK_Deltaj},
	\end{align}
	as $\ntoinf$, where the bound depends on $j$ only through $\lambda_j\inv$ and $\delta_j\inv$. Using that $m^*(a_n+b_n)=o(1)$ along with $\lVert \hbSigma^{*-1} -\bSigma^{*-1} \rVert_{\text{op},2}=O(m^*(a_n+b_n))$ a.s. as $\ntoinf$, observe
	\begin{align}
		\lvert \hbejsT \hbSigma^{*-1} \hbejs -\bejsT \bSigma^{*-1} \bejs \rvert
		&=O\left( \lVert \hDelta_j \rVert_2^2 + \lVert \hDelta_j \rVert_2 \lVert \bejs \rVert_2  + m^*(a_n+b_n) \lVert \bejs \rVert_2^2 \right)\quadas\label{eq:traceSigmaK_ejhatSejhat},
	\end{align}
	as $\ntoinf$, where the bound depends on $j$ only through $ \lVert \hDelta_j \rVert_2$ and $\lVert \bejs \rVert_2$. Also,
	\begin{align}
		\lvert \hbejsT \bW^* \hbphi_{j}^* - \bejsT \bW^* \bphi_{j}^*\rvert
		&\le
		\lVert \hDelta_j \rVert_2  \lVert \bW^* \hbphi_{j}^*\rVert_2+\lVert \bejs \rVert_2 \lVert \bW^* (\hbphi_{j}^*-\bphi_{j}^*)\rVert_2\label{eq:traceSigmaK_ejhatWPhihat}.
	\end{align}
	For large enough $n$ and in view of \eqref{eq:traceSigmaK_DeltaPhihatWSWPhihat}, \eqref{eq:traceSigmaK_Deltaj}, and using that the bound \eqref{eq:shrinkageFuncSubject_bound_ek} holds analogously for $\lVert \bejs \rVert_2$ and the time points $\bT^*$, we obtain
	\begin{align}
		\sumjK [\hbphi_{j}^{*T} \bW^* \hbSigma^* \bW^* \hbphi_{j}^* -
		\bphi_{j}^{*T} \bW^* \bSigma^* \bW^* \bphi_{j}^*]
		&=
		O_p\left(
		(a_n+b_n)   \sumkK \lambda_k^{-2}  \delta_k\inv
		\right) \label{eq:traceSigmaK_Sum_DeltaPhihatWSWPhihat},
	\end{align}
	and 
	\begin{align}
		\sumjK [ \lVert \hDelta_j  \rVert_2 \lVert \bejs  \rVert_2 ]
		&=
		O_p\left( (a_n+b_n) \sumjK \lambda_j^{-2} \delta_j\inv
		\right) \label{eq:traceSigmaK_Sum_Deltaej},
	\end{align}
	and
	\begin{align}
		\sumjK [ \lVert \hDelta_j  \rVert_2^2 ]
		&=
		O_p\left( m^{*} (a_n+b_n)^2 \sumjK \lambda_j^{-2} \delta_j^{-2}
		\right) \label{eq:traceSigmaK_Sum_Delta4}.
	\end{align}
	Since $E(\lVert \bejs \rVert_2^2)\lesssim m^{*-1}\lambda_j^{-2}$, which follows analogously as in  \eqref{eq:shrinkageFuncSubject_expectation_ek}, we also have
	\begin{align}
		\sumjK \lVert \bejs \rVert_2^2 
		&=
		O_p\left( m^{*-1} \sumjK \lambda_j^{-2}
		\right) \label{eq:traceSigmaK_Sum_ej4}.
	\end{align}
	From \eqref{eq:traceSigmaK_WPhihat} and \eqref{eq:traceSigmaK_Deltaj}, we obtain
	\begin{align}
		\sumjK [ \lVert \hDelta_j\rVert_2 \lVert  \bW^* \hbphi_{j}^* \rVert_2 ]
		&=
		O_p\left( (a_n+b_n) \sumjK \lambda_j^{-2} \delta_j\inv
		\right)\label{eq:traceSigmaK_Sum_DeltajWPhihat},
	\end{align}
	and using \eqref{eq:secondMomentFPCscoresNorm_WDeltaPhiBound} we also have
	\begin{align}
		\sumjK  [\lVert \bejs \rVert_2 \lVert \bW^* (\hbphi_{j}^*-\bphi_{j}^*)\rVert_2 ]
		&=
		O_p\left( (a_n+b_n) m^{*-1} \sumjK \lambda_j^{-2} \delta_j\inv
		\right)\label{eq:traceSigmaK_Sum_ejDeltaPhi}.
	\end{align}
	Combining \eqref{eq:traceSigmaK_ejhatSejhat}, \eqref{eq:traceSigmaK_Sum_Deltaej}, \eqref{eq:traceSigmaK_Sum_Delta4}, and \eqref{eq:traceSigmaK_Sum_ej4} implies
	\begin{align}
		&\sumjK  \lvert \hbejsT \hbSigma^{*-1} \hbejs -\bejsT \bSigma^{*-1} \bejs \rvert \nonumber\\
		&\lesssim
		\sumjK [\lVert \hDelta_j \rVert_2^2 + \lVert \hDelta_j \rVert_2 \lVert \bejs \rVert_2  + m^{*}(a_n+b_n) \lVert \bejs \rVert_2^2 ]\nonumber\\
		&=
		O_p\left(
		m^{*} (a_n+b_n)^2 \sumjK \lambda_j^{-2} \delta_j^{-2}
		+
		(a_n+b_n) \sumjK \lambda_j^{-2} \delta_j\inv
		\right) \label{eq:traceSigmaK_ejhatSejhat_rate},
	\end{align}
	while combining \eqref{eq:traceSigmaK_ejhatWPhihat}, \eqref{eq:traceSigmaK_Sum_DeltajWPhihat}, and \eqref{eq:traceSigmaK_Sum_ejDeltaPhi} leads to
	\begin{align}
		\sumjK \lvert \hbejsT \bW^* \hbphi_{j}^* - \bejsT \bW^* \bphi_{j}^* \rvert
		&=
		O_p\left( (a_n+b_n)  \sumjK \lambda_j^{-2} \delta_j\inv
		\right)\label{eq:traceSigmaK_ejhatWPhihat_rate}.
	\end{align}
	Combining \eqref{eq:traceSigmaK_keyDecomposition},
	\eqref{eq:traceSigmaK_Sum_DeltaPhihatWSWPhihat}, \eqref{eq:traceSigmaK_ejhatSejhat_rate}, and  \eqref{eq:traceSigmaK_ejhatWPhihat_rate} leads to
	\begin{align*}
		&\lvert \text{trace}( \hbSigma_{K}^*-\bSigma_{K}^*)\rvert\\
		&\le
		\sumjK \lvert [\hbSigma_{K}^*-\bSigma_{K}^*]_{j,j}\rvert
		&=
		O_p\left( m^{*} (a_n+b_n)^2 \sumjK \lambda_j^{-2} \delta_j^{-2}
		+
		(a_n+b_n) \sumjK \lambda_j^{-2} \delta_j\inv
		\right),
	\end{align*}
	and the result follows.
\end{proof}

\bco

Consider an independent densely measured subject $i^*$ as in Section \ref{S:setup}. The next result shows shrinkage of the conditional variance corresponding to the $K$-truncated distribution.
\begin{thm} \label{thm:xiEstHat2}
	Suppose that \ref{a:Diff}, \ref{a:GaussProcess}, \ref{a:eigendecay} and \ref{a:K}--\ref{a:Ubeta} hold. Let $K>0$ be fixed and consider either a sparse design setting when $n_i\le N_0<\infty$ or a dense design when $n_i=m\to\infty$, $i=1,\dots,n$. Set $a_n=a_{n1}$ and $b_n=b_{n1}$ for the sparse case, and $a_n=a_{n2}$ and $b_n=b_{n2}$ for the dense design. For a new independent subject $i^*$, if $m^*(a_n+b_n)=o(1)$ as $\ntoinf$, where $m^*=m^*(n)\to\infty$, 
	\begin{align*}
		\lVert  \hbSigma_{K}^{*} -\bSigma_{K}^{*}\rVert_{\text{op},2}&=O_p(a_n+b_n).
	\end{align*}
\end{thm}
\begin{proof}[Proof of Theorem \ref{thm:xiEstHat2}]
	Recall that $\hat{\bmu}^*=\hat{\mu}(\bT^*)$, $\bT^*=(T_1^*,\dots,T_{m^*}^*)^T$, the estimated FPCs $\hxi_{k}^*=\hlambda_k \hphi_{k}(\bT^*)^T \hbSigma^{*-1}(\bX^*-\hat{\bmu}^*)$, $\hbPhi_{K}^{*}$ is analogous to $\hbPhi_{iK}$ while replacing the $T_{ij}$ with $T_{j}^*$, and similarly for quantities such as $\bPhi_{K}^{*}$, $\hbSigma^{*-1}$, and $\bSigma^{*-1}$. Note that
	\begin{align}
		\bSigma_{K}^{*} -\hbSigma_{K}^{*}&= \bLambda_K-\hbLambda_K  +  \hbLambda_K \hbPhi_{K}^{*T} \hbSigma^{*-1} \hbPhi_{K}^{*}\hbLambda_K- \bLambda_K \bPhi_{K}^{*T} \bSigma^{*-1} \bPhi_{K}^{*}\bLambda_K\label{eq:thm2_1},
	\end{align}
	where $ \lVert \bLambda_K-\hbLambda_K  \rVert_{\text{op},2} =O_p(a_n+b_n)$ follows from Theorem $5.2$ in \cite{zhan:16} along with perturbation results \citep{bosq:00} and the fact that  $\lVert \bLambda_K-\hbLambda_K  \rVert_{\text{op},2}\le \sqrt{K} \max_{1\le k\le K}\lvert \lambda_k-\hlambda_k\rvert$. Since $\hlambda_k\hbphi_{k}^*= \int_\cT \hGamma(\bT^*, t) \hphi_k(t) dt $ and writing  $\hbe_k^{*}=\int_\cT \hGamma(\bT^*, t) \hphi_k(t) dt - \hbSigma^* \bW^* \hbphi_{k}^*$, we have that  the $(j,l)$ entry of $\hbLambda_K \hbPhi_{K}^{*T} \hbSigma^{*-1} \hbPhi_{K}^{*}\hbLambda_K$ is given by
	\begin{align}
		[\hbLambda_K \hbPhi_{K}^{*T} \hbSigma^{*-1} \hbPhi_{K}^{*}\hbLambda_K]_{j,l}&= (\hbe_j^{*T}\hbSigma^{*-1}+ \hbphi_{j}^{*T} \bW^{*} ) (\hbe_l^{*}+\hbSigma^{*} \bW^{*} \hbphi_{l}^{*})\nonumber \\
		&= \hbe_j^{*T}\hbSigma^{*-1} \hbe_l^{*}+ \hbe_j^{*T} \bW^{*} \hbphi_{l}^{*}+ \hbphi_{j}^{*T} \bW^{*} \hbe_l^{*}+ \hbphi_{j}^{*T} \bW^{*} \hbSigma^{*} \bW^{*} \hbphi_{l}^{*}\label{eq:thm2_2},
	\end{align}
	where $1\le j,l \le K$. Denote by $\hGamma(\bT^*,\bT^{*T})$ the matrix whose $(i,j)$ element is  $\hGamma(T_i^*,T_j^*)$, $1\le i,j\le m^*$, and similarly define $\Gamma(\bT^*,\bT^{*T})$. Also note that $\hbSigma^{*}=\hsigma^2 I_{m^*}+\hGamma(\bT^*,\bT^{*T})$, where $I_{m^*}\in\bbR^{m^*\times m^*}$ is the identity matrix. 
	From \eqref{eq:shrinkageFuncSubject_expectation_ek}, \eqref{eq:shrinkEstFunc_DeltaGamma}, \eqref{eq:secondMomentFPCscoresNorm_WDeltaPhiBound}, \eqref{eq:traceSigmaK_Deltaj}, Lemma \ref{lem:unifGaps}, and using that $ \lVert \hbSigma^{*-1} -\bSigma^{*-1} \rVert_{\text{op},2}=O_p(m^*(a_n+b_n))$  along with the condition $m^*(a_n+b_n)=o(1)$ as $\ntoinf$, it follows that
	$\lVert \hGamma(\bT^*,\bT^{*T})-\Gamma(\bT^*,\bT^{*T})  \rVert_2=O_p(m^*(a_n+b_n))$, $\lVert  \bW^{*}(\hbphi_{p}^{*}-\bphi_{p}^{*}) \rVert_2=O_p(m^{*-1/2}(a_n+b_n))$, $p=j,l$, $\lVert \Gamma(\bT^*,\bT^{*T})  \rVert_{\text{op},2}=O(m^*)$, $\lVert \bSigma^{*}\rVert_{\text{op},2}=O(m^*)$, $\lVert \hbSigma^{*}\rVert_{\text{op},2}=O_p(m^*)$, $\lVert \bW^{*}\bphi_{p}^{*} \rVert_2=O_p(m^{*-1/2})$, $p=j,l$, $ \lVert \hbSigma^{*} -\bSigma^{*} \rVert_{\text{op},2}=O_p(m^*(a_n+b_n))$, $\lVert \bW^{*} \rVert_2=O_p(m^{*-1/2})$, $\lVert \be_p^{*}\rVert_2=O_p(m^{*-1/2})$ and $\lVert \hbe_p^{*}-\be_p^{*} \rVert_2=O_p(m^{*1/2}(a_n+b_n))$, $p=j,l$. These bounds imply
	\begin{align*}
		\hbphi_{j}^{*T} \bW^{*} \hbSigma^{*} \bW^{*} \hbphi_{l}^{*}-
		\bphi_{j}^{*T} \bW^{*} \bSigma^{*} \bW^{*} \bphi_{l}^{*}=O_p(a_n+b_n), \\
		\hbe_j^{*T}\hbSigma^{*-1} \hbe_l^{*}-\be_j^{*T}\bSigma^{*-1} \be_l^{*}=O_p(a_n+b_n),\\
		\hbe_j^{*T} \bW^{*} \hbphi_{l}^{*}- \be_j^{*T} \bW^{*} \bphi_{l}^{*}=O_p(a_n+b_n),\\
		\hbphi_{j}^{*T} \bW^{*} \hbe_l^{*}- \bphi_{j}^{*T} \bW^{*} \be_l^{*}=O_p(a_n+b_n),
	\end{align*}
	which combined with \eqref{eq:thm2_2} leads to
	\begin{align*}
		[\hbLambda_K \hbPhi_{K}^{*T} \hbSigma^{*-1} \hbPhi_{K}^{*}\hbLambda_K]_{j,l}-
		[\bLambda_K \bPhi_{K}^{*T} \bSigma^{*-1} \bPhi_{K}^{*}\bLambda_K]_{j,l} =O_p(a_n+b_n).
	\end{align*}
	Hence $\lVert \hbLambda_K \hbPhi_{K}^{*T} \hbSigma^{*-1} \hbPhi_{K}^{*}\hbLambda_K- \bLambda_K \bPhi_{K}^{*T} \bSigma^{*-1} \bPhi_{K}^{*}\bLambda_K \rVert_F=O_p(a_n+b_n)$ and the result follows from \eqref{eq:thm2_1}.
\end{proof}
\fi

%We first provide some auxiliary lemmas that will be used in the proof of the main results in section \ref{S:FLM}. Here we derive  a slightly more general result without using optimal bandwidths. Recall 

For the following, recall that $w_i:=\left(\sum_{l=1}^n n_l\right)\inv$, $\upsilon_M=\sum_{m=1}^M \delta_m\inv$ and $C(t)=E((X(t)-\mu(t))Y)=\int_\cT \beta(s)\Gamma(t,s)ds$, $t\in\cT$.

\begin{lem}\label{lem:linear3}
	Suppose that \ref{a:GaussProcess}, \ref{a:eigendecay}-\ref{a:rate}, \ref{a:K}--\ref{a:Ubeta} hold and consider a sparse design with  $n_i\le N_0<\infty$, setting  $a_n=a_{n1}$ and $b_n=b_{n1}$. Then
	\begin{align}
		n\inv\sumin \lVert \hbxiiK-\tbxi_{iK} \rVert_2^2 &=O_p((a_n+b_n)^2) \label{eq:lem_linear3_1},
	\end{align}
	and
	\begin{align}
		n\inv \sumin \lVert \tbxi_{iK}\rVert_2^2 &=O_p(1) \label{eq:lem_linear3_2}.
	\end{align}
\end{lem}
\begin{proof}[Proof of Lemma \ref{lem:linear3}]
	First note that $\lVert \hmu-\mu\rVert_\infty=O(a_n)$ a.s. and $\lVert \hGamma-\Gamma\rVert_\infty=O(a_n+b_n)$ a.s., which are due to Theorem $5.1$ and $5.2$ in \cite{zhan:16}. From arguments in the proof of Theorem 2 in \cite{dai:16:1} and noting that the constant $c$ that appears in Lemma A.3 in \cite{mull:03} can be taken as a universal constant $c=2$, 
	\begin{align}
		\lVert \hbxiiK-\tbxi_{iK} \rVert_2^2&\le O((a_n+b_n)^2) \lVert \bX_i-\hbmui \rVert_2^2 +O(a_n^2)+O(a_n(a_n+b_n)) \lVert \bX_i-\hbmui \rVert_2 \quadas\label{eq:lemlinear3_0},
	\end{align}
	where the $O((a_n+b_n)^2)$, $O(a_n^2)$ and $O(a_n(a_n+b_n))$ terms are uniform in $i$. Let $\bU_{i}=(X_i(T_{i1}),\dots,X_i(T_{in_i}))^T$ be the true but unobserved values of the trajectory for the ith subject at the time points $\bTi$, so that by construction $\bX_i=\bU_i + \bepsi$. Then
	\begin{align}
		n\inv\sumin \lVert \bX_i-\hbmui \rVert_2 &= n\inv\sumin \lVert \bU_i +\bepsi -\hbmui \rVert_2 \nonumber \\
		&\le  n\inv\sumin \lVert \bU_i  -\bmui \rVert_2 + n\inv\sumin \lVert \bepsi\rVert_2+n\inv\sumin \lVert \bmui -\hbmui \rVert_2\label{eq:lemlinear3_1},
	\end{align}
	where $n\inv\sumin \lVert \bmui -\hbmui \rVert_2=O(a_n)$ almost surely. Since $n_i\le N_0$ in the sparse case, it is easy to show that $n\inv\sumin \lVert \bepsi\rVert_2=O_p(1)$ and by Jensen's inequality
	\begin{align*}
		E\left(n\inv\sumin \lVert \bU_i  -\bmui \rVert_2 \right)&\le n\inv\sumin \left( \sum_{j=1}^{n_i} E(X_i(T_{ij})-\mu(T_{ij}) )^2 \right)^{1/2}\nonumber  \\
		&=n\inv\sumin \left(\sum_{j=1}^{n_i} E(\Gamma(T_{ij},T_{ij})) \right)^{1/2} \le (\normInf{\Gamma} N_0)^{1/2}
		=O(1),
	\end{align*}
	where the first equality follow by conditioning on $T_{ij}$. This shows that $ n\inv\sumin \lVert \bU_i  -\bmui \rVert_2=O_p(1)$. Combining with \eqref{eq:lemlinear3_1} leads to
	\begin{align}
		n\inv\sumin \lVert \bX_i-\hbmui \rVert_2&=O_p(1)\label{eq:lemlinear3_2}.
	\end{align}
	By the triangle inequality
	\begin{align*}
		\lVert \bX_i-\hbmui \rVert_2^2&\le  \lVert \bU_i  -\bmui \rVert_2^2 + \lVert \bepsi\rVert_2^2 + \lVert \bmui -\hbmui \rVert_2^2  \\
		&\quad+ 2   \lVert \bU_i  -\bmui \rVert_2 \lVert \bepsi\rVert_2+ 2 \lVert \bU_i  -\bmui \rVert_2 \lVert \bmui -\hbmui \rVert_2 + 2 \lVert \bepsi\rVert_2 \lVert \bmui -\hbmui \rVert_2,
	\end{align*}
	where $\lVert \bmui -\hbmui \rVert_2\le \sqrt{N_0 \sup_{t\in\cT} ( \mu(t)-\hmu(t))^2}=O(a_n)$ a.s. and uniformly over $i$. This  along with the independence of $\bepsi$ and $\bU_i$, conditionally on $\bTi$, and using similar arguments as before, leads to $E\lVert \bX_i-\hbmui \rVert_2^2 =O(1)$ uniformly over $i$. Thus
	\begin{align}
		n\inv\sumin \lVert \bX_i-\hbmui \rVert_2^2&=O_p(1)\label{eq:lemlinear3_3}.
	\end{align}
	Combining \eqref{eq:lemlinear3_0}, \eqref{eq:lemlinear3_2} and \eqref{eq:lemlinear3_3} leads to the first result in \eqref{eq:lem_linear3_1}. Note that
	\begin{align*}
		E(\tbxi_{iK}^T \tbxi_{iK})^2 &\le E(\norm{\bLambda_K  \bPhiiK^T \bSigma_i\inv}_{\text{op},2}^4 E( \normtwo{\bXi-\bmu_i}^4\vert \bTi))\le O(1),
	\end{align*}
	where the $O(1)$ term is uniform in $i$ and the last inequality follows from $\norm{\bLambda_K}_{\text{op},2}\le \lambda_1 K$, $\norm{\bPhiiK}_{\text{op},2}\le N_0 \sum_{j=1}^K \lVert \phi_j\rVert_\infty^2 $, $\norm{\bSigma_i\inv}_{\text{op},2}\le  \sigma^{-2}$,  $E(\normtwo{\bXi-\bmu_i}^4\vert \bTi)\le O(1)$ uniformly over $i$, where the latter is a  consequence of  the Gaussian process assumption on $X_i(\cdot)$ and  $\lVert \Gamma\rVert_\infty <\infty$. 
	Thus, $E(\lVert \tbxi_{iK}\rVert_2^2) =O(1)$ uniformly in $i$ which implies $E(	n\inv \sumin \lVert \tbxi_{iK}\rVert_2^2)=O(1)$ and the second result in  \eqref{eq:lem_linear3_2}.
\end{proof}

\begin{lem}\label{lem:auxLemma_beta_1}
	Suppose that \ref{a:GaussProcess}, \ref{a:eigendecay}-\ref{a:rate}, \ref{a:betaSeries}-\ref{a:Mrate}, \ref{a:K}--\ref{a:Ubeta} hold and consider a sparse design with  $n_i\le N_0<\infty$, setting  $a_n=a_{n1}$ and $b_n=b_{n1}$. Let $\tZ_i(t):=\sum_{j=1}^{n_i} w_i K_h(T_{ij}-t) \left(\frac{T_{ij}-t}{h}\right)^r (U_{ij} Y_i-C(t))$, where $U_{ij}=X(T_{ij})-\mu(T_{ij})$ and $r=0,1$. Then
	\begin{align*}
		E[\tZ_i^2(t)]=O((n^2 h)\inv),
	\end{align*}
	where the $O((n^2 h)\inv)$ term is uniform in $i$ and $t$.
\end{lem}
\begin{proof}[Proof of Lemma \ref{lem:auxLemma_beta_1}]
	Observe 
	\begin{align*}
		&E[\tZ_i^2(t)]\\
		&= E\left( \sum_{j=1}^{n_i} w_i^2 K_h^2(T_{ij}-t) \left(\frac{T_{ij}-t}{h}\right)^{2r} (U_{ij} Y_i-C(t))^2  \right)\\
		&+ E\Big( \sum_{j=1}^{n_i} \sum_{l\neq j} w_i ^2 K_h(T_{ij}-t)K_h(T_{il}-t)\\
		&\quad\quad\quad\quad  \left(\frac{T_{ij}-t}{h}\right)^r \left(\frac{T_{il}-t}{h}\right)^r (U_{ij} Y_i-C(t)) (U_{il} Y_i-C(t)) \Big)
	\end{align*}
	and note  that for any $t_1,t_2\in\cT$, with $\mu_Y=E(Y)$, 
	\begin{align*}
		E(U(t_1) U(t_2) Y^2)&=E(U(t_1) U(t_2) [\mu_Y + \int_\cT \beta(s)U(s)ds+\epsilon_Y]^2)\\
		&= (\mu_Y^2+\sigma_Y^2) \Gamma(t_1,t_2)+2\int_\cT \mu_Y \beta(s) E(U(t_1)U(t_2)U(s))ds \\
		&\quad + \int_\cT \int_\cT \beta(s_1) \beta(s_2) E( U(t_1)U(t_2)U(s_1)U(s_2)) ds_1 ds_2\\
		&=O(1),
	\end{align*}
	where the $O(1)$ term is uniform over $t_1$ and $t_2$, which follows from $\lVert \Gamma \rVert_\infty<\infty$ and  $U(t)\sim N(0,\Gamma(t,t))$, owing  to \ref{a:GaussProcess}. This implies that  $E((U_{ij} Y_i-C(t))^2 \vert T_{ij})$ is uniformly bounded above, and by a conditioning argument it follows that
	\begin{align*}
		&E\left( \sum_{j=1}^{n_i} w_i^2 K_h^2(T_{ij}-t) \left(\frac{T_{ij}-t}{h}\right)^{2r} (U_{ij} Y_i-C(t))^2  \right)\\
		&\le O(1) E\left( \sum_{j=1}^{n_i} w_i^2 K_h^2(T_{ij}-t) \left(\frac{T_{ij}-t}{h}\right)^{2r} \right)\\
		&=O((n^2 h)\inv),
	\end{align*}
	where the last equality is due to $w_i\le n\inv$. Let  $R_{iqr,h}(t)=w_i K_h(T_{iq}-t) \left(\frac{T_{iq}-t}{h}\right)^r $, $q=j,l$. Since $E((U_{ij} Y_i-C(t)) (U_{il} Y_i-C(t))\vert T_{ij},T_{il})=O(1)$ uniformly in $i$ and $t$, similar arguments as before show that
	\begin{align*}
		&E\left( \sum_{j=1}^{n_i} \sum_{l\neq j} R_{ijr,h}(t) R_{ilr,h}(t) (U_{ij} Y_i-C(t)) (U_{il} Y_i-C(t)) \right)\\
		&\le 
		O(1)  \sum_{j=1}^{n_i} \sum_{l\neq j} E[R_{ijr,h}(t)] E[R_{ilr,h}(t)]\\
		&=O(n^{-2}),
	\end{align*}
	whence  the result follows.
\end{proof}

\begin{lem}\label{lem:auxLemma_beta_2}
	Suppose that \ref{a:GaussProcess}, \ref{a:eigendecay}-\ref{a:rate}, \ref{a:betaSeries}-\ref{a:Mrate}, \ref{a:K}--\ref{a:Ubeta} hold and consider a sparse design with  $n_i\le N_0<\infty$, setting  $a_n=a_{n1}$ and $b_n=b_{n1}$. For $r=0,1$ we have
	\begin{align}
		\lVert  \sumin \sum_{j=1}^{n_i} w_i K_h(T_{ij}-\cdot) \left(\frac{T_{ij}-\cdot}{h}\right)^r \epsilon_{ij} Y_i\rVert_{L^2}&=O_p((nh)^{-1/2}),\label{eq:lem_auxLemma_beta_2_1}
	\end{align}
	and
	\begin{align}
		\lVert  \sumin \sum_{j=1}^{n_i} w_i K_h(T_{ij}-\cdot) \left(\frac{T_{ij}-\cdot}{h}\right)^r (U_{ij} Y_i -C(\cdot))\rVert_{L^2}&=O_p\left(\left(\frac{1}{nh}+h^2\right)^{1/2}\right),\label{eq:lem_auxLemma_beta_2_2}
	\end{align}
	where $U_{ij}=X(T_{ij})-\mu(T_{ij})$.
\end{lem}
\begin{proof}[Proof of Lemma \ref{lem:auxLemma_beta_2}]
	Define $Z_i(t):=\sum_{j=1}^{n_i} w_i K_h(T_{ij}-t) \left(\frac{T_{ij}-t}{h}\right)^r \epsilon_{ij} Y_i$. Note that the $Z_i$ are independent and by independence of the $\epsilon_{ij}$ along with a conditioning argument, $E(Z_i(t))=0$ and 
	\begin{align*}
		E(\lVert \sumin  Z_i\rVert_{L^2}^2 )&=  \sumin \int_\cT E(Z_i^2(t))dt,
	\end{align*}
	\begin{align*}
		E(Z_i^2(t))&= E\left( \sum_{j=1}^{n_i} \sum_{l=1}^{n_i} w_i^2 K_h(T_{ij}-t) \left(\frac{T_{ij}-t}{h}\right)^r \epsilon_{ij}  
		K_h(T_{il}-t) \left(\frac{T_{il}-t}{h}\right)^r \epsilon_{il} Y_i^2 \right)\\
		&=
		\sum_{j=1}^{n_i} E\left(w_i^2 K_h^2(T_{ij}-t) \left(\frac{T_{ij}-t}{h}\right)^{2r} \epsilon_{ij}^2 Y_i^2 \right)\\
		&=
		E(Y^2) \sigma^2 \sum_{j=1}^{n_i} E\left(w_i^2 K_h^2(T_{ij}-t) \left(\frac{T_{ij}-t}{h}\right)^{2r} \right)=O((n^2 h)\inv),
	\end{align*}
	where the $O(h\inv)$ is uniform in $i$ and $t$. Thus $E(\lVert \sumin Z_i\rVert_{L^2}^2 )=O((nh)\inv)$ and the first result in \eqref{eq:lem_auxLemma_beta_2_1} follows. Defining $\tZ_i(t):=\sum_{j=1}^{n_i} w_i K_h(T_{ij}-t) \left(\frac{T_{ij}-t}{h}\right)^r (U_{ij} Y_i-C(t))$, we have
	\begin{align}
		E(\lVert \sumin \tZ_i  \rVert_{L^2}^2 )&= \sumin \int_\cT E[\tZ_i^2(t)] dt + 
		\sumin \sum_{k\neq i} \int_\cT E(\tZ_i(t)) E(\tZ_k(t)).\label{eq:auxLemma_beta_2_1}
	\end{align}
	By a conditioning argument, it follows that
	\begin{align*}
		\lvert E(\tZ_i(t)) \rvert &= \Big \lvert \sum_{j=1}^{n_i} w_i E\left( K_h(T_{ij}-t) \left(\frac{T_{ij}-t}{h}\right)^r (C(T_{ij})-C(t)) \right)\Big \rvert \\
		&\le \sum_{j=1}^{n_i} w_i  \int_{-t/h}^{(1-t)/h} \lvert u^r \rvert  K(u) \lvert C(t+uh)-C(t)\rvert f(t+uh) du\\
		&\le \sum_{j=1}^{n_i} w_i  \sup_{s\in [-1,1]} \lvert C'(s) \rvert \ \lVert f\rVert_\infty h \int_{-t/h}^{(1-t)/h} \lvert u^{r+1} \rvert  K(u)du\\
		&\le O\left( n\inv h\right),
	\end{align*}
	where the $O\left( n\inv h\right)$ is uniform in $i$ and $t$. This implies $\lvert  \sumin \sum_{k\neq i} \int_\cT E(\tZ_i(t)) E(\tZ_k(t)) \rvert=O(h^2)$. Combining with \eqref{eq:auxLemma_beta_2_1} and Lemma \ref{lem:auxLemma_beta_1}, the second result in \eqref{eq:lem_auxLemma_beta_2_2} follows.
\end{proof}

\begin{lem}\label{lem:auxLemma_beta_3}
	Suppose that \ref{a:GaussProcess}, \ref{a:eigendecay}-\ref{a:rate}, \ref{a:K}--\ref{a:Ubeta} hold and consider a sparse design with  $n_i\le N_0<\infty$, setting  $a_n=a_{n1}$ and $b_n=b_{n1}$. For $r=0,1$, 
	\begin{align*}
		\lVert \sumin \sum_{j=1}^{n_i}w_i K_h(T_{ij}-t) \left(\frac{T_{ij}-t}{h}\right)^r  (\mu(T_{ij})-\hmu(T_{ij}))Y_i \rVert_{L^2}&=O_p(a_n).
	\end{align*}
\end{lem}
\begin{proof}[Proof of Lemma \ref{lem:auxLemma_beta_3}]
	Setting  $Z_i:=\sum_{j=1}^{n_i}w_i K_h(T_{ij}-t) \left(\frac{T_{ij}-t}{h}\right)^r   (\mu(T_{ij})-\hmu(T_{ij}))Y_i$, note that
	\begin{align}
		E\left( \lVert \sumin Z_i \rVert_{L^2}^2\right)&=\int_\cT \sumin E[Z_i^2(t)] dt + \int_\cT \sumin \sum_{k\neq i} E[Z_i(t)Z_k(t)].\label{eq:auxLemma_beta_3_1}
	\end{align}
	Since $\lvert Z_i(t) \rvert\le \lVert \hmu-\mu \rVert_\infty \sum_{j=1}^{n_i}w_i K_h(T_{ij}-t) \left( \frac{\lvert T_{ij}-t\rvert}{h}\right)^r  \lvert Y_i\rvert$, it follows that
	\begin{align}
		&E[Z_i^2(t)] \nonumber \\
		&\le  E\Big[\lVert \hmu-\mu \rVert_\infty^2  \sum_{j=1}^{n_i}\sum_{l=1}^{n_i} w_i^2 Y_i^2 K_h(T_{ij}-t) K_h(T_{il}-t)\left( \frac{\lvert T_{ij}-t\rvert}{h}\right)^r \left( \frac{\lvert T_{il}-t\rvert}{h}\right)^r \Big] \nonumber\\
		&\le O(a_n^2) \Big\{ \sum_{j=1}^{n_i} w_i^2 E(Y^2) E\Big[K_h^2(T_{ij}-t) \left(\frac{T_{ij}-t}{h}\right)^{2r} \Big]\nonumber\\
		& +\sum_{j=1}^{n_i} \sum_{l\neq j}w_i^2 E(Y^2) E\Big[K_h(T_{ij}-t)\left( \frac{\lvert T_{ij}-t\rvert}{h}\right)^r \Big]  E\Big[K_h(T_{il}-t)\left( \frac{\lvert T_{il}-t\rvert}{h}\right)^r \Big] \Big\}\nonumber\\
		&\le O(a_n^2) [O(n^{-2} h\inv) +O(n^{-2})]\nonumber\\
		&=O(a_n^2 n^{-2} h\inv),\label{eq:auxLemma_beta_3_2}
	\end{align}
	where the first inequality follows from {Theorem $5.1$ in \cite{zhan:16}} and the term $O(a_n^2 n^{-2} h\inv)$ is uniform in $i$ and $t$. Similarly, for $k\neq i$ and setting $h_{qdr}(t):=\left(\frac{\lvert T_{qd}-t\rvert}{h}\right)^r$, $q=i,k$ and $d=j,l$, we have
	\begin{align*}
		&E(\lvert Z_i(t)Z_k(t)\rvert)\\
		&\le E\Big[ \sum_{j=1}^{n_i}\sum_{l=1}^{n_k}w_i K_h(T_{ij}-t)h_{ijr}(t)  \lvert \mu(T_{ij})-\hmu(T_{ij})\\
		&\quad\quad\quad\quad\quad\quad \rvert Y_i  w_k K_h(T_{kl}-t)  h_{klr}(t) \lvert\mu(T_{kl})-\hmu(T_{kl})\rvert Y_k \Big]\\
		&\le O(a_n^2) \sum_{j=1}^{n_i}\sum_{l=1}^{n_k} w_i w_k E[K_h(T_{ij}-t)h_{ijr}(t)] E[K_h(T_{kl}-t)h_{klr}(t)]  [E(Y)]^2\\
		&=O(a_n^2 n^{-2}),
	\end{align*}
	where the $O(a_n^2 n^{-2})$ term is uniform in $i$, $k$ and $t$. Combining  this with \eqref{eq:auxLemma_beta_3_1} and \eqref{eq:auxLemma_beta_3_2} leads to the result.
\end{proof}

\begin{lem}\label{lem:auxLemma_beta_4}
	Suppose that \ref{a:GaussProcess}, \ref{a:eigendecay}-\ref{a:rate}, \ref{a:K}--\ref{a:Ubeta} hold and consider a sparse design with  $n_i\le N_0<\infty$, setting  $a_n=a_{n1}$ and $b_n=b_{n1}$. Then
	\begin{align*}
		\lVert \hC - C \rVert_{L^2}=O_p\left(\left(\frac{1}{nh}+h^2\right)^{1/2}+a_n\right).
	\end{align*}
\end{lem}
\begin{proof}[Proof of Lemma \ref{lem:auxLemma_beta_4}]
	Proceeding similarly to the proof of Theorem 3.1 in \cite{zhan:16}, using \eqref{eq:localLinear_C},
	\begin{align*}
		\hC(t)&=\frac{S_2(t)\tR_0(t)-S_1(t)\tR_1(t)}{S_0(t)S_2(t)-S_1^2(t)},
	\end{align*}
	where
	\begin{align*}
		S_r(t)&=\sumin \sum_{j=1}^{n_i}w_i K_h(T_{ij}-t) \left(\frac{T_{ij}-t}{h}\right)^r,\\
		\tR_r(t)&=\sumin \sum_{j=1}^{n_i}w_i K_h(T_{ij}-t) \left(\frac{T_{ij}-t}{h}\right)^r C_i(T_{ij}),
	\end{align*}
	and $r=0,1,2$. Then
	\begin{align}
		\hC(t)-C(t)&=\frac{(\tR_0(t)-C(t)S_0(t))S_2(t)-(\tR_1(t)-C(t)S_1(t))S_1(t)}{S_0(t)S_2(t)-S_1^2(t)}.\label{eq:auxLemma_beta_4_1}
	\end{align}
	Since $C_i(T_{ij})=(\tX_{ij}-\hmu(T_{ij}))Y_i=(U_{ij}+\epsilon_{ij})Y_i+(\mu(T_{ij})-\hmu(T_{ij}))Y_i$, where $U_{ij}=X(T_{ij})-\mu(T_{ij})$, 
	\begin{align*}
		&\lVert \tR_0(t)-C(t)S_0(t) \rVert_{L^2} \\
		&\le \lVert \sumin \sum_{j=1}^{n_i}w_i K_h(T_{ij}-t) (U_{ij}Y_i-C(t)) \rVert_{L^2}+ \lVert \sumin \sum_{j=1}^{n_i}w_i K_h(T_{ij}-t)\epsilon_{ij} Y_i\rVert_{L^2}\\
		&\quad + \lVert \sumin \sum_{j=1}^{n_i}w_i K_h(T_{ij}-t)  (\mu(T_{ij})-\hmu(T_{ij}))Y_i \rVert_{L^2}\\
		&= O_p\left(\left(\frac{1}{nh}+h^2\right)^{1/2}\right)+O_p((nh)^{-1/2})+O_p(a_n)\\
		&=O_p\left(\left(\frac{1}{nh}+h^2\right)^{1/2}\right)+O_p(a_n),
	\end{align*}
	where the last equality follows from Lemma \ref{lem:auxLemma_beta_2} and Lemma \ref{lem:auxLemma_beta_3}. Similarly
	\begin{align*}
		&\lVert \tR_1(t)-C(t)S_1(t) \rVert_{L^2} \\
		&\le
		\lVert  \sumin \sum_{j=1}^{n_i}w_i K_h(T_{ij}-t) \left(\frac{T_{ij}-t}{h}\right) (U_{ij}Y_i-C(t))  \rVert_{L^2}\\
		&\quad +
		\lVert  \sumin \sum_{j=1}^{n_i}w_i K_h(T_{ij}-t) \left(\frac{T_{ij}-t}{h}\right) \epsilon_{ij} Y_i \rVert_{L^2}\\
		&\quad \quad +
		\lVert  \sumin \sum_{j=1}^{n_i}w_i K_h(T_{ij}-t) \left(\frac{T_{ij}-t}{h}\right) (\mu(T_{ij})-\hmu(T_{ij}))Y_i  \rVert_{L^2}\\
		&=O_p\left(\left(\frac{1}{nh}+h^2\right)^{1/2}\right)+O_p(a_n).
	\end{align*}
	These along with \eqref{eq:auxLemma_beta_4_1} and similar arguments as in  the proof of Theorem 4.1 in \cite{zhan:16} show that $S_0(t)S_2(t)-S_1^2(t)$ is positive and bounded away from $0$ with probability tending to $1$ and $\sup_{t\in \cT}\lvert S_r(t)\rvert =O_p(1)$, $r=1,2$. The result then follows.
\end{proof}

Recall that the eigenpairs of the integral operator $\hXi$ associated with $\hGamma$ are $(\hlambda_k, \hphi_k)$, and those of $\Xi$ are $(\lambda_k, \phi_k)$, $k\geq 1$. 

\begin{lem}\label{lem:auxLemma_beta_5}
	Suppose that \ref{a:GaussProcess}, \ref{a:eigendecay}-\ref{a:rate}, \ref{a:K}--\ref{a:Ubeta} hold and consider a sparse design with  $n_i\le N_0<\infty$, setting  $a_n=a_{n1}$ and $b_n=b_{n1}$. Then, setting    $\tau_M=\sumM \frac{1}{\lambda_m}$, for large enough $n$, the following relations hold almost surely,
	\begin{align}
		\sumM \frac{\lvert \hsigma_m -\sigma_m \rvert}{\lambda_m} &=  \tau_M \lVert \hC - C \rVert_{L^2}+ \tau_M^{1/2}  O(c_n^{\rho}), \label{eq:lem_auxLemma_beta_5_1} \\
		\sumM   \lvert \hsigma_m -\sigma_m \rvert \frac{\lvert \hlambda_m-\lambda_m \rvert}{\lvert \hlambda_m\rvert \lambda_m}&\le O(c_n^{2\rho})+ \lVert \hC-C \rVert_{L^2} \tau_M^{1/2}  O(c_n^{\rho}),\label{eq:lem_auxLemma_beta_5_2}\\
		\sumM   \lvert \sigma_m \rvert \frac{\lvert \hlambda_m-\lambda_m \rvert}{\lvert \hlambda_m\rvert \lambda_m}&\le O(c_n) \tau_M,\label{eq:lem_auxLemma_beta_5_3}\\
		\sumM \Big \lvert \frac{\hsigma_m}{\hlambda_m}-\frac{\sigma_m}{\lambda_m}\Big \rvert  \lVert \hphi_m-\phi_m\rVert_{L^2}&\le O(c_n^{2\rho})+ O(c_n^{\rho}) (\lVert \hC - C \rVert_{L^2}+c_n) \tau_M^{1/2}, \label{eq:lem_auxLemma_beta_5_4}\\
		\sumM \frac{\lvert \sigma_m \rvert}{\lambda_m}\lVert \hphi_m-\phi_m\rVert_{L^2}&\le O(c_n) \upsilon_M,\label{eq:lem_auxLemma_beta_5_5}
	\end{align}
\end{lem}
\begin{proof}[Proof of Lemma \ref{lem:auxLemma_beta_5}]
	First note 
	\begin{align*}
		\sumM \frac{1}{\delta_m}
		&\le \left(\sumM \frac{1}{\lambda_m \delta_m^2} \right)^{1/2}  \left(\sumM \lambda_m \right)^{1/2} \\
		&\le \left(\sumM  \frac{1}{\sqrt{\lambda_m}\delta_m}\right) \left(\sum_{m=1}^\infty \lambda_m \right)^{1/2}\\
		&=O(c_n^{\rho-1}),
	\end{align*}
	implying  $c_n \upsilon_M=O(c_n^{\rho})=o(1)$ as $\ntoinf$. By the Cauchy--Schwarz inequality and from {Theorem $5.2$ in \cite{zhan:16}}, we have $\lVert \hXi-\Xi \rVert_{\text{op}}=O(a_n+b_n)$ a.s.. Note that from the orthonormality of the $\phi_k$ and using perturbation results \citep{bosq:00}, we have $\lVert \hphi_k-\phi_k\rVert_{L^2}\le 2\sqrt{2} \lVert \hXi-\Xi \rVert_{\text{op}} /\delta_k$, $k\ge 1$, so that for any $m\ge 1$
	\begin{align}
		\lvert \hsigma_m-\sigma_m \rvert&= \lvert \langle \hC,\hphi_m \rangle_{L^2}-\langle C,\phi_m \rangle_{L^2}\rvert \nonumber \\
		&\le 2\sqrt{2}\lVert \hC - C \rVert_{L^2} \frac{\lVert \hXi-\Xi \rVert_{\text{op}}}{\delta_m} +  \lVert \hC - C \rVert_{L^2}
		+2\sqrt{2} \lVert C \rVert_{L^2}   \frac{\lVert \hXi-\Xi \rVert_{\text{op}}}{\delta_m},\label{eq:auxLemma_beta_5_1}
	\end{align}
	and from  $\delta_m\le \lambda_m$, 
	\begin{align}
		\sumM \frac{\lVert \hXi-\Xi \rVert_{\text{op}}}{\lambda_m \delta_m}&\le \tau_M^{1/2}  \sumM \frac{\lVert \hXi-\Xi \rVert_{\text{op}}}{\sqrt{\lambda_m} \delta_m}=\tau_M^{1/2}  O(c_n^{\rho})\quadas \label{eq:auxLemma_beta_5_2}.
	\end{align}
	Thus
	\begin{align*}
		\sumM \frac{\lvert \hsigma_m -\sigma_m \rvert}{\lambda_m}&\le 	\tau_M^{1/2}  O(c_n^{\rho}) \lVert \hC - C \rVert_{L^2}
		+ \tau_M \lVert \hC - C \rVert_{L^2}+ \tau_M^{1/2}  O(c_n^{\rho})\quadas\\
		&= \tau_M \lVert \hC - C \rVert_{L^2}+ \tau_M^{1/2}  O(c_n^{\rho}),
	\end{align*}
	which shows the first result in \eqref{eq:lem_auxLemma_beta_5_1}. Since $M=M(n)$ is such that $\sum_{m=1}^M \frac{1}{\sqrt{\lambda_m}\delta_m}=O(c_n^{\rho-1})$ as $n\to\infty$, then $\sum_{m=1}^M \lVert \hXi-\Xi \rVert_{\text{op}} \lambda_m^{-1/2}\delta_m^{-1}=O(c_n^\rho)=o(1)$ a.s. and $\lambda_{M}=o(1)$ as $n\to \infty$. Thus, for large enough $n$ we have $\lambda_{M}<1$ and $ \lVert \hXi-\Xi \rVert_{\text{op}} \lambda_M^{-1/2}\delta_M^{-1}\leq \sum_{m=1}^M \lVert \hXi-\Xi \rVert_{\text{op}}
	\lambda_m^{-1/2}\delta_m^{-1} \leq 1/2$ a.s., so that $\lVert \hXi-\Xi \rVert_{\text{op}} \leq \lambda_M^{1/2}\delta_M/2 \leq \delta_M/2\leq \lambda_M/2$ a.s.. This shows that there exists $n_0\ge 1$ such that for all $n\ge n_0$ it holds that $\lVert \hXi-\Xi \rVert_{\text{op}}\le \lambda_M/2$ a.s.. Then $\lvert \hlambda_m -\lambda_m \rvert\le \lVert \hXi-\Xi \rVert_{\text{op}}$ implies  $\lvert \hlambda_m \rvert \ge \lambda_m/2$ a.s.   for large enough $n$. With \eqref{eq:auxLemma_beta_5_1}, \eqref{eq:auxLemma_beta_5_2}, 
	\begin{align*}
		\sumM  \lvert \hsigma_m -\sigma_m \rvert \frac{\lvert \hlambda_m-\lambda_m \rvert}{\lvert \hlambda_m\rvert \lambda_m}&\le 2
		\sumM  \lvert \hsigma_m -\sigma_m \rvert \frac{\lVert \hXi-\Xi \rVert_{\text{op}}}{ \lambda_m^2}\\
		&\le 4 \sqrt{2} \lVert \hC-C \rVert_{L^2} \sumM  \frac{\lVert \hXi-\Xi \rVert_{\text{op}}^2}{ \lambda_m^2\delta_m}
		+
		2  \lVert \hC-C \rVert_{L^2} \sumM  \frac{\lVert \hXi-\Xi \rVert_{\text{op}}}{ \lambda_m^2}\\
		&\quad +
		4 \sqrt{2} \lVert C \rVert_{L^2} \sumM \frac{\lVert \hXi-\Xi \rVert_{\text{op}}^2}{ \lambda_m^2\delta_m}\\
		&\le  \lVert \hC-C \rVert_{L^2} O(c_n^{2\rho})+\lVert \hC-C \rVert_{L^2} \tau_M^{1/2}  O(c_n^{\rho})+O(c_n^{2\rho}) \quadas\\
		&= O(c_n^{2\rho})+ \lVert \hC-C \rVert_{L^2} \tau_M^{1/2}  O(c_n^{\rho})\quadas,
	\end{align*}
	for large enough $n$, implying the second  result in \eqref{eq:lem_auxLemma_beta_5_2}. Similarly, for large enough $n$ and a.s. 
	\begin{align*}
		\sumM   \lvert \sigma_m \rvert \frac{\lvert \hlambda_m-\lambda_m \rvert}{\lvert \hlambda_m\rvert \lambda_m}&\le 
		2 \sumM   \lvert \sigma_m \rvert \frac{\lvert \hlambda_m-\lambda_m \rvert}{\lambda_m^2}
		\le 
		O(c_n) \left( \sumM \frac{\sigma_m^2}{\lambda_m^2}\right)^{1/2}  \tau_M =O(c_n) \tau_M,
	\end{align*}
	where the last equality is due to $\sum_{m=1}^\infty \sigma_m^2/\lambda_m^2<\infty$. This shows the third result in \eqref{eq:lem_auxLemma_beta_5_3}. Now, 
	\begin{align}
		&\sumM \Big \lvert \frac{\hsigma_m}{\hlambda_m}-\frac{\sigma_m}{\lambda_m}\Big \rvert  \lVert \hphi_m-\phi_m\rVert_{L^2}\nonumber\\
		&\le 
		\sumM \frac{\lvert \hsigma_m -\sigma_m\rvert }{\lambda_m} \lVert \hphi_m-\phi_m\rVert_{L^2}
		+ \sumM \frac{\lvert \sigma_m \rvert \lvert \hlambda_m -\lambda_m\rvert }{\lvert \hlambda_m \rvert \lambda_m}\lVert \hphi_m-\phi_m\rVert_{L^2}\nonumber \\
		&\quad  +
		\sumM \frac{\lvert \hsigma_m-\sigma_m \rvert \lvert \hlambda_m -\lambda_m\rvert }{\lvert \hlambda_m \rvert \lambda_m}
		\lVert \hphi_m-\phi_m\rVert_{L^2} \label{eq:auxLemma_beta_5_3}.
	\end{align}
	From \eqref{eq:auxLemma_beta_5_1}, \eqref{eq:auxLemma_beta_5_2} and using that $\lVert \hphi_m-\phi_m\rVert_{L^2}\le 2\sqrt{2} \lVert \hXi-\Xi \rVert_{\text{op}} /\delta_m$, we obtain
	\begin{align}
		&\sumM \frac{\lvert \hsigma_m -\sigma_m\rvert }{\lambda_m} \lVert \hphi_m-\phi_m\rVert_{L^2}\nonumber \\
		&\le 
		8 \lVert \hC - C \rVert_{L^2}\sumM   \lVert \hXi-\Xi \rVert_{\text{op}}^2 \frac{1}{\lambda_m \delta_m^2}
		+
		2\sqrt{2} \lVert \hC - C \rVert_{L^2} \sumM \frac{ \lVert \hXi-\Xi \rVert_{\text{op}}}{\lambda_m \delta_m}\nonumber\\
		&\quad + 8  \lVert C \rVert_{L^2} \sumM    \frac{ \lVert \hXi-\Xi \rVert_{\text{op}}^2}{\lambda_m \delta_m^2}\nonumber\\
		&\le \lVert \hC - C \rVert_{L^2} O(c_n^{2\rho})+\lVert \hC - C \rVert_{L^2} \tau_M^{1/2}  O(c_n^{\rho})+O(c_n^{2\rho})\quadas\nonumber\\
		&= O(c_n^{2\rho})+ \lVert \hC - C \rVert_{L^2} \tau_M^{1/2}  O(c_n^{\rho})\quadas\label{eq:auxLemma_beta_5_4}.
	\end{align}
	For large enough $n$, 
	\begin{align}
		\sumM \frac{\lvert \sigma_m \rvert \lvert \hlambda_m -\lambda_m\rvert }{\lvert \hlambda_m \rvert \lambda_m}\lVert \hphi_m-\phi_m\rVert_{L^2}&\le 4\sqrt{2} \sumM \lvert \sigma_m \rvert  \frac{\lVert \hXi-\Xi \rVert_{\text{op}}^2}{\lambda_m^2 \delta_m} \quadas \nonumber\\
		&\le \left( \sumM \frac{\sigma_m^2}{\lambda_m^2}\right)^{1/2}  O(c_n^{1+\rho}) \tau_M^{1/2} \quadas\nonumber \\
		&=O(c_n^{1+\rho}) \tau_M^{1/2}\label{eq:auxLemma_beta_5_5}.
	\end{align}
	Similarly, from \eqref{eq:auxLemma_beta_5_1} we obtain
	\begin{align}
		&\sumM \frac{\lvert \hsigma_m-\sigma_m \rvert \lvert \hlambda_m -\lambda_m\rvert }{\lvert \hlambda_m \rvert \lambda_m}
		\lVert \hphi_m-\phi_m\rVert_{L^2} \nonumber \\
		&\le   4\sqrt{2}  \sumM \lvert \hsigma_m-\sigma_m \rvert  \frac{\lVert \hXi-\Xi \rVert_{\text{op}}^2}{\lambda_m^2 \delta_m}\nonumber \\
		&\le 16 \lVert \hC - C \rVert_{L^2} \sumM \frac{\lVert \hXi-\Xi \rVert_{\text{op}}^3}{\lambda_m^2 \delta_m^2}+
		4\sqrt{2} \lVert \hC - C \rVert_{L^2}  \sumM \frac{\lVert \hXi-\Xi \rVert_{\text{op}}^2}{\lambda_m^2 \delta_m}\nonumber \\
		&\quad + 16 \lVert C \rVert_{L^2} \sumM \frac{\lVert \hXi-\Xi \rVert_{\text{op}}^3}{\lambda_m^2 \delta_m^2}\nonumber\\
		&\le O(c_n^{1+2\rho}) \tau_M \lVert \hC - C \rVert_{L^2}+ O(c_n^{2\rho})\lVert \hC - C \rVert_{L^2} + O(c_n^{1+2\rho}) \tau_M\quadas \nonumber \\
		&= O(c_n^{1+2\rho}) \tau_M+O(c_n^{2\rho})\lVert \hC - C \rVert_{L^2} \quadas\label{eq:auxLemma_beta_5_6}.
	\end{align}
	Combining \eqref{eq:auxLemma_beta_5_3}, \eqref{eq:auxLemma_beta_5_4}, \eqref{eq:auxLemma_beta_5_5} and \eqref{eq:auxLemma_beta_5_6} with the fact that $c_n \tau_M\le c_n \upsilon_M=o(1)$ as $\ntoinf$, which was already shown, leads to the fourth result in \eqref{eq:lem_auxLemma_beta_5_4}. Finally
	\begin{align*}
		&\sumM \frac{\lvert \sigma_m \rvert}{\lambda_m}\lVert \hphi_m-\phi_m\rVert_{L^2}\\
		&\le 2\sqrt{2}\sumM \frac{\lvert \sigma_m \rvert \ \lVert \hXi-\Xi \rVert_{\text{op}} }{\lambda_m \delta_m}\\
		&\le \left( \sumM \frac{\sigma_m^2}{\lambda_m^2}\right)^{1/2} \lVert \hXi-\Xi \rVert_{\text{op}} \upsilon_M=O(c_n)\upsilon_M\quadas,
	\end{align*}
	which shows the last result in \eqref{eq:lem_auxLemma_beta_5_5}.
\end{proof}

The next lemma provides the $L^2$ convergence of the empirical estimate $\hbeta_M$ towards $\beta$, which is required to construct the estimated predictive distribution $\hcP_{iK}$. Recall that
\begin{align*}
	\hbeta_M(t):=\sum_{m=1}^M \frac{\hsigma_m}{\hlambda_m} \hphi_m(t), \quad \tinT,
\end{align*}
$\Theta_M=\Big\lVert \sum_{m\ge M+1} \frac{\sigma_m}{\lambda_m}\phi_m \Big\rVert_{L^2}$ and $\tau_M=\sumM \lambda_m\inv$.
\begin{lem}\label{lem:Lemma_beta}
	Suppose that \ref{a:GaussProcess}, \ref{a:eigendecay}-\ref{a:rate}, \ref{a:K}--\ref{a:Ubeta} hold and consider a sparse design with  $n_i\le N_0<\infty$, setting  $a_n=a_{n1}$ and $b_n=b_{n1}$.  Let $K\ge 1$. Then
	\begin{align}
		\lVert \hbeta_M -\beta \rVert_{L^2}=O_p(r_n),\label{eq:lem_LemmaBeta1}
	\end{align}
	and
	\begin{align}
		\int_\cT \hbeta_M(t) \hphi_k(t) dt= \int_\cT \beta(t) \phi_k(t) dt+O_p(r_n),\label{eq:lem_LemmaBeta2}
	\end{align}
	where $r_n=c_n \upsilon_M+c_n^{\rho}\tau_M^{1/2} + \tau_M  \Big[\left(\frac{1}{nh}+h^2\right)^{1/2}+a_n\Big]
	+ \Theta_M$ and $k=1,\dots,K$.
\end{lem}
\begin{proof}[Proof of Lemma \ref{lem:Lemma_beta}]
	Observe 
	\begin{align}
		\lVert \hbeta_M- \beta\rVert_{L^2}&\le  \sum_{m=1}^M  \Big \lVert \frac{\hsigma_m}{\hlambda_m} \hphi_m-\frac{\sigma_m}{\lambda_m} \phi_m  \Big\rVert_{L^2} + \Big\lVert \sum_{m\ge M+1} \frac{\sigma_m}{\lambda_m}\phi_m \Big\rVert_{L^2}\label{eq:Lemma_beta_1},
	\end{align}
	and
	\begin{align}
		\sum_{m=1}^M  \Big \lVert \frac{\hsigma_m}{\hlambda_m} \hphi_m-\frac{\sigma_m}{\lambda_m} \phi_m  \Big\rVert_{L^2} &\le
		\sum_{m=1}^M \Big \lvert \frac{\hsigma_m}{\hlambda_m}-\frac{\sigma_m}{\lambda_m}\Big \rvert  \lVert \hphi_m-\phi_m\rVert_{L^2} 
		+ \sum_{m=1}^M \Big \lvert \frac{\hsigma_m}{\hlambda_m}-\frac{\sigma_m}{\lambda_m}\Big \rvert \nonumber\\
		&\quad + \sum_{m=1}^M \frac{\lvert \sigma_m \rvert}{\lambda_m}\lVert \hphi_m-\phi_m\rVert_{L^2}.\label{eq:Lemma_beta_2}
	\end{align}
	By the triangle inequality and Lemma \ref{lem:auxLemma_beta_5}, we have that for large enough $n$
	\begin{align*}
		\sumM \Big \lvert \frac{\hsigma_m}{\hlambda_m}-\frac{\sigma_m}{\lambda_m}\Big \rvert &\le 
		\sumM \frac{\lvert \hsigma_m -\sigma_m\rvert }{\lambda_m} 
		+\sumM \frac{\lvert \hsigma_m-\sigma_m \rvert \lvert \hlambda_m -\lambda_m\rvert }{\lvert \hlambda_m \rvert \lambda_m}
		+\sumM \frac{\lvert \sigma_m \rvert \lvert \hlambda_m -\lambda_m\rvert }{\lvert \hlambda_m \rvert \lambda_m}\\
		&=
		\tau_M \lVert \hC - C \rVert_{L^2}+ O(c_n^{\rho})  \tau_M^{1/2}  +O(c_n) \tau_M \quadas\\
		&= \tau_M \lVert \hC - C \rVert_{L^2}+ O(c_n^{\rho}) \tau_M^{1/2}    \quadas,
	\end{align*}
	where the second equality is due to $c_n \tau_M =c_n^{\rho} \tau_M^{1/2} c_n^{1-\rho} \tau_M^{1/2}=o(1) c_n^{\rho} \tau_M^{1/2}$, and
	\begin{align*}
		&\sumM \Big \lvert \frac{\hsigma_m}{\hlambda_m}-\frac{\sigma_m}{\lambda_m}\Big \rvert  \lVert \hphi_m-\phi_m\rVert_{L^2}+ \sumM \frac{\lvert \sigma_m \rvert}{\lambda_m}\lVert \hphi_m-\phi_m\rVert_{L^2}\\
		&\le
		O(c_n^{2\rho})+ O(c_n^{\rho}) \lVert \hC - C \rVert_{L^2} \tau_M^{1/2}+O(c_n) \upsilon_M.
	\end{align*}
	With \eqref{eq:Lemma_beta_1}, \eqref{eq:Lemma_beta_2} and  the fact that $\upsilon_M=O(c_n^{\rho-1})$ as $\ntoinf$, which was shown in the proof of  Lemma \ref{lem:auxLemma_beta_5}, we arrive at 
	\begin{align*}
		\lVert \hbeta_M- \beta\rVert_{L^2}&\le O(c_n) \upsilon_M+O(c_n^{\rho}) \tau_M^{1/2}    + \tau_M  \lVert \hC - C \rVert_{L^2} 
		+ \Big\lVert \sum_{m\ge M+1} \frac{\sigma_m}{\lambda_m}\phi_m \Big\rVert_{L^2}
	\end{align*}
	and the result in \eqref{eq:lem_LemmaBeta1} follows from Lemma \ref{lem:auxLemma_beta_4}. Finally, recalling that $\hbeta_k=	\int_\cT \hbeta_M(t) \hphi_k(t) dt $ and $\beta_k=	\int_\cT \beta(t) \phi_k(t) dt$, we have
	\begin{align*}
		\lvert \hbeta_k-\beta_k\rvert &=\lvert \int_\cT  [\hbeta_M(t) \hphi_k(t)-\beta(t) \phi_k(t)]dt \rvert  \\
		&\le \lVert \hbeta_M-\beta\rVert_{L^2} \lVert \hphi_k-\phi_k \rVert_{L^2}+\lVert \hbeta_M-\beta\rVert_{L^2}
		+\lVert \beta\rVert_{L^2} \lVert \hphi_k-\phi_k \rVert_{L^2} \\
		&=O_p(r_n + a_n+b_n)=O_p(r_n),
	\end{align*}
	where the second equality is due to the fact that $\lVert \hphi_k-\phi_k\rVert_{L^2}\le  O(a_n+b_n) $ a.s., which follows from the proof of Lemma \ref{lem:auxLemma_beta_5}. This shows the second result in \eqref{eq:lem_LemmaBeta2}.
\end{proof}
We remark that in the sparse case when choosing the optimal bandwidth $h\asymp n^{-1/3}$, then the rate 
\begin{align*}
	\tau_M  [\left((nh)\inv+h^2\right)^{1/2}+a_n],
\end{align*}
is faster than $c_n \upsilon_M $ and thus the rate $r_n$ is equivalent to $\alpha_n$ defined as in Theorem \ref{thm:predictiveConsistency}. Recall that $\cP_{iK}$ corresponds to the true predictive distribution $\eta_{iK}\vert \bXi,\bTi$, or equivalently  $N(\beta_0 +\bbeta_K^T \tbxiiK, \bbeta_K^T \bSigmaiK \bbeta_K)$, while  $\tcP_{iK} \overset{d}{=} N(\beta_0 +\bbeta_K^T \hbxiiK, \bbeta_K^T \hbSigmaiK \bbeta_K)$ corresponds to an intermediate target, replacing population quantities by their estimated counterparts but keeping the true intercept and slope coefficients $\beta_0$ and $\bbeta_K$. Also $\hcP_{iK}$ corresponds to the estimated predictive distribution, i.e.  $\hcP_{iK} \overset{d}{=} N(\hbeta_0 +\hbbeta_K^T \hbxiiK, \hbbeta_K^T \hbSigmaiK \hbbeta_K)$. Finally, recall that $F_{iK}(t), \tF_{iK}(t)$ and $\hF_{iK}$ are the distribution functions associated with $\cP_{iK}, \tcP_{iK}$ and $\hcP_{iK}$, respectively. We require the following auxiliary lemma.

\begin{lem}\label{lem:eigenFuncK}
	Under the conditions of Theorem \ref{thm:predictiveConsistency}, it holds that
	\begin{equation*}
		\lVert \bSigma_{iK}-\hat{\bSigma}_{iK}\rVert_F=O(N_0^{5/2}(a_n+b_n)),
	\end{equation*}
	a.s.  as $n\to\infty$.
\end{lem}
\begin{proof}[Proof of Lemma \ref{lem:eigenFuncK}]
	Note that
	\begin{align}
		\bSigma_{iK}-\hat{\bSigma}_{iK}
		&=(\bLambda_K-\hbLambda_K)+\hbLambda_K \hbPhiiK^T \hbSigma_i^{-1} \hbPhiiK\hbLambda_K-\bLambda_K \bPhiiK^T \bSigma_i^{-1} \bPhiiK\bLambda_K \nonumber\\
		&=
		(\bLambda_K-\hbLambda_K)+ (\hbLambda_K \hbPhiiK^T-\bLambda_K \bPhiiK^T) \hbSigma_i^{-1} \hbPhiiK\hbLambda_K\nonumber \\
		&\quad +\bLambda_K \bPhiiK^T (\hbSigma_i^{-1} \hbPhiiK\hbLambda_K- \bSigma_i^{-1} \bPhiiK\bLambda_K).\label{lemmaS2_3}
	\end{align}
	Denoting by $C_i:=(\hbSigma_i^{-1} \hbPhiiK\hbLambda_K- \bSigma_i^{-1} \bPhiiK\bLambda_K)$, we have
	\begin{equation}\label{lemmaS2_4}
		C_i=
		(\hbSigma_i^{-1}-\bSigma_i^{-1}) (\hbPhiiK\hbLambda_K-\bPhiiK\bLambda_K )
		+
		\bSigma_i^{-1}(\hbPhiiK\hbLambda_K-\bPhiiK\bLambda_K )
		+
		(\hbSigma_i^{-1}-\bSigma_i^{-1}) \bPhiiK\bLambda_K,
	\end{equation}
	where
	\begin{equation}\label{lemmaS2_5}
		\hbPhiiK\hbLambda_K-\bPhiiK\bLambda_K = (\hbPhiiK-\bPhiiK)(\hbLambda_K-\bLambda_K)+ \bPhiiK(\hbLambda_K-\bLambda_K)+(\hbPhiiK-\bPhiiK)\bLambda_K.
	\end{equation}
	Note that $\lVert \hbPhiiK-\bPhiiK \rVert_F\leq \sqrt{N_0K} \max_{1\leq k\leq K}\lVert \hphi_k-\phi_k \rVert_\infty=O(\sqrt{N_0}(a_n+b_n))$ a.s.  as $n\to\infty$, which {follows similarly as in Proposition $1$ in \cite{dai:16:1} by employing Theorem $5.1$ and $5.2$ in \cite{zhan:16}}. Using perturbation results \citep{bosq:00}, {Theorem $5.2$ in \cite{zhan:16}} and the Cauchy Schwarz inequality, it follows that $|\hlambda_k-\lambda_k | \le  \lVert \Gamma-\hGamma\rVert_{\infty}=O(a_n+b_n)$ a.s.  as $n\to\infty$. Thus $\lVert \hbLambda_K-\bLambda_K \rVert_F \leq \sqrt{K} \max_{1\leq k\leq K}\lVert \hlambda_k-\lambda_k \rVert_\infty=O(a_n+b_n) $ a.s. as $n\to\infty$. Furthermore, from the proof of Theorem 2 in \cite{dai:16:1} we have $\lVert \hbSigma_i^{-1}-\bSigma_i^{-1} \rVert_{\text{op},2}=O(N_0(a_n+b_n))$ a.s. which implies $\lVert \hbSigma_i\inv -\bSigma_i\inv \rVert_F\leq \sqrt{N_0} \lVert \hbSigma_i\inv -\bSigma_i\inv  \rVert_{\text{op},2} =O(N_0^{3/2}(a_n+b_n))$ a.s. as $n\to\infty$. Thus, from \eqref{lemmaS2_4} and \eqref{lemmaS2_5},  $\lVert \bSigma_i\inv \rVert_{\text{op},2}\leq \sigma^{-2}$ and $\lVert \hbPhiiK \rVert_F\le \sqrt{N_0 K} \max_{1\le k\le K} \lVert \phi_k \rVert_\infty$, it follows that $\lVert \hbPhiiK\hbLambda_K-\bPhiiK\bLambda_K\rVert_F=O(\sqrt{N_0}(a_n+b_n))$ and $\lVert C_i\rVert_F=O(N_0^2(a_n+b_n))$ a.s.  as $n\to\infty$. From \eqref{lemmaS2_3} and using that
	\begin{align}
		\lVert  (\hbLambda_K \hbPhiiK^T-\bLambda_K \bPhiiK^T) \hbSigma_i\inv \hbPhiiK\hbLambda_K\rVert_F&=
		\lVert (\hbLambda_K \hbPhiiK^T-\bLambda_K \bPhiiK^T) \left( C_i +   \bSigma_i\inv  \bPhiiK\bLambda_K\right)\rVert_F\nonumber \\
		&=O(N_0(a_n+b_n))+O(N_0^{5/2}(a_n+b_n)^2) \ \text{a.s.},
	\end{align}
	as $n\to\infty$, we obtain $\lVert \bSigma_{iK}-\hat{\bSigma}_{iK}\rVert_F=O(N_0^{5/2}(a_n+b_n))$ a.s.  as $n\to\infty$, which shows the result.
\end{proof}

The following auxiliary lemmas  will be used in the proof of Theorem \ref{thm:WassDiscrep}.
\begin{lem}\label{lem:Wass_Discrep}
	Suppose that \ref{a:GaussProcess}, \ref{a:eigendecay}-\ref{a:rate}, \ref{a:K}--\ref{a:Ubeta} hold and consider a sparse design with  $n_i\le N_0<\infty$, setting  $a_n=a_{n1}$ and $b_n=b_{n1}$. Then
	\begin{align*}
		n\inv \sumin (\teta_{iK} - \heta_{iK}) \epsilon_{iY}&=O_p(\alpha_n),
	\end{align*}
	where $\heta_{iK}=\hbeta_0 + \hbbeta_K^T \hbxiiK$, and $(\hbeta_0 ,\hbbeta_K^T)^T$ are the estimates in the functional linear model as in Theorem \ref{thm:predictiveConsistency}.
\end{lem}
\begin{proof}[Proof of Lemma \ref{lem:Wass_Discrep}]
	By the Cauchy--Schwarz inequality
	\begin{align}
		\lvert n\inv \sumin (\teta_{iK} - \heta_{iK}) \epsilon_{iY}\rvert &\le \left(n\inv\sumin (\teta_{iK} - \heta_{iK})^2  \right)^{1/2} \left(n\inv\sumin \epsilon_{iY}^2  \right)^{1/2}\label{eq:lem_wassDiscrep_1},
	\end{align}
	where $\left(n\inv\sumin \epsilon_{iY}^2  \right)^{1/2}=O_p(1)$, whence $\lvert \teta_{iK} - \heta_{iK} \rvert \le \lvert \beta_0 - \hbeta_0\rvert + \lVert \hbbeta_K-\bbeta_K\rVert_2 \lVert \tbxi_{iK} \rVert_2+ \lVert \hbbeta_K \rVert_2 \lVert \tbxi_{iK}- \hbxiiK \rVert_2$, and then
	\begin{align*}
		(\teta_{iK} - \heta_{iK} )^2&\le (\beta_0 - \hbeta_0)^2+ 
		\lVert \hbbeta_K-\bbeta_K\rVert_2^2 \lVert \tbxi_{iK} \rVert_2^2 +
		\lVert \hbbeta_K \rVert_2^2 \lVert \tbxi_{iK}- \hbxiiK \rVert_2^2\\
		&+2  \lvert \beta_0 - \hbeta_0\rvert   \lVert \hbbeta_K-\bbeta_K\rVert_2 \lVert \tbxi_{iK} \rVert_2
		+2 \lvert \beta_0 - \hbeta_0\rvert  \lVert \hbbeta_K \rVert_2 \lVert \tbxi_{iK}- \hbxiiK \rVert_2 \\
		&+2  \lVert \hbbeta_K-\bbeta_K\rVert_2 \lVert \tbxi_{iK} \rVert_2 \lVert \hbbeta_K \rVert_2 \lVert \tbxi_{iK}- \hbxiiK \rVert_2.
	\end{align*}
	From Lemma \ref{lem:Lemma_beta} we have $\lvert \beta_0 - \hbeta_0\rvert =O_p(n^{-1/2})$ and  $\lVert \hbbeta_K-\bbeta_K\rVert_2=O_p(\alpha_n)$, which combined with Lemma \ref{lem:linear3} and the Cauchy--Schwarz inequality leads to
	\begin{align}
		n\inv \sumin (\teta_{iK} - \heta_{iK} )^2 &=O_p((\alpha_n)^2)\label{eq:lem_wassDiscrep_2}.
	\end{align}
	The result then follows from \eqref{eq:lem_wassDiscrep_1} and \eqref{eq:lem_wassDiscrep_2}.
\end{proof}

\begin{lem}\label{lem:Wass_Discrep2}
	Under the conditions of Theorem \ref{thm:WassDiscrep}, it holds that
	\begin{align*}
		n\inv \sumin (\eta_{iK} - \teta_{iK})^2 - \bbeta_K^T E(\bSigma_{1K}) \bbeta_K  &=O_p(n^{-1/2}).
	\end{align*}
\end{lem}
\begin{proof}[Proof of Lemma \ref{lem:Wass_Discrep2}]
	Since  $\bxi_{iK}- \tbxi_{iK} \vert \bT_i \sim N(\boldsymbol{0}, \bSigmaiK)$, by conditioning on $\bTi$, 
	\begin{align*}
		E(\eta_{iK} - \teta_{iK})^2&=E\left(E\Big[ \left(\bbeta_K^T(\bxi_{iK}-\tbxiiK)\right)^2\Big\vert \bT_i\Big]\right)=\bbeta_K^T E(\bSigma_{1K}) \bbeta_K,
	\end{align*}
	where the last equality is due to the fact that $n_i=m_0$ implies that $\bSigmaiK$ are a sequence of i.i.d.\ random positive definite matrices. Similarly, since $\eta_{iK} - \teta_{iK}\vert \bT_i \sim N(0, \bbeta_K^T\bSigmaiK\bbeta_K)$ we have $E( (\eta_{iK} - \teta_{iK})^4\vert \bT_i)=3 (\bbeta_K^T \bSigmaiK \bbeta_K)^2$ and thus
	\begin{align*}
		\Var((\eta_{iK} - \teta_{iK})^2)&=E(\Var((\eta_{iK} - \teta_{iK})^2\vert \bTi))+\Var(\bbeta_K^T \bSigmaiK \bbeta_K)\\
		&=2 E((\bbeta_K^T \bSigmaiK \bbeta_K)^2)+\Var(\bbeta_K^T \bSigmaiK \bbeta_K)\\
		&=O(1),
	\end{align*}
	where the $O(1)$ term is uniform in $i$ since $\lVert \bSigmaiK\rVert_{\text{op}}$ is uniformly bounded in the sparse case.  Since the $\eta_{iK} - \teta_{iK}$ are independent, the result then follows from the Central Limit Theorem.
\end{proof}

\begin{lem}\label{lem:series_eigengap}
	Under the assumptions of Theorem \ref{thm:WassDiscrep}, it holds that
	\begin{equation*}
		\sum_{j=1}^{M}\frac{\lambda_j}{\delta_j^2} = O\left( \sum_{j=1}^{M}\frac{1}{\lambda_j\delta_j^2} \right),
	\end{equation*}
	as $n\to\infty$.
\end{lem}
\begin{proof}[Proof of Lemma \ref{lem:series_eigengap}]
	Since $\lambda_j \to 0$ as $j\to\infty$, there exists $J^*\geq 1$ such that $\lambda_j\geq 1$ for $j\leq J^*$ and $\lambda_j<1$ whenever $j>J^*$. Note that
	\begin{align}
		\sum_{j=1}^{M}\frac{\lambda_j}{\delta_j^2} &=\sum_{j=1}^{M}\frac{1}{\lambda_j\delta_j^2} + \sum_{j=1}^{M}\left(\lambda_j-\frac{1}{\lambda_j} \right) \frac{1}{\delta_j^2} \nonumber \\
		&=\sum_{j=1}^{M}\frac{1}{\lambda_j\delta_j^2} +  \sum_{j=1}^{J^*}\left(\lambda_j-\frac{1}{\lambda_j} \right) \frac{1}{\delta_j^2} +\sum_{j=J^*+1}^{M}\left(\lambda_j-\frac{1}{\lambda_j} \right) \frac{1}{\delta_j^2},\label{lemmaS3_1}
	\end{align}
	whence it  suffices to show that the third term in \eqref{lemmaS3_1} diverges to $-\infty$ as $n\to\infty$. For this, 
	\begin{align*}
		\sum_{j=J^*+1}^{M}\left(\lambda_j-\frac{1}{\lambda_j} \right) \frac{1}{\delta_j^2} &\leq \lambda_{J^*+1}^2 \sum_{j=J^*+1}^{M} \frac{1}{\lambda_j \delta_j^2} - \sum_{j=J^*+1}^{M}\frac{1}{\lambda_j \delta_j^2}= \sum_{j=J^*+1}^{M} \frac{1}{\lambda_j \delta_j^2}  \ \left(\lambda_{J^*+1}^2-1\right).
	\end{align*}
	The result  follows from the fact that $\lambda_{J^*+1}^2-1<0$ and since $\sum_{j=1}^M \lambda_j^{-1/2} \delta_j\inv \to\infty$ as $n\to\infty$ implies $\sum_{j=J^*+1}^{M} \lambda_j\inv \delta_j^{-2} \to\infty$ as $n\to\infty$.
\end{proof}

%\section{S.3 Additional Results for Section \ref{S:FLM}}

Consider the Brownian motion as an example of a Gaussian process for which $\lambda_m=4/(\pi^2 (2m-1)^2)$ and $\phi_m(t)=\sqrt{2}\sin((2m-1)\pi t/2)$ \citep{hsin:15}. Adopting the optimal bandwidth choices as discussed in Section \ref{S:FLM} leads to $c_n\asymp (\log(n)/n)^{1/3}$.
\begin{lem}\label{lem:BrownianExample}
	Let $\rho\in(1/3,1)$. For the Brownian motion, if $M=M(n)$ satisfies 
	\begin{align}
		M(n)\asymp \left( \frac{\log(n)}{n}\right)^{(\rho-1)/15},\label{eq:Brownian_1}
	\end{align}
	then condition \ref{a:Mrate} holds and
	\begin{align}
		\tau_M\asymp  \left( \frac{\log(n)}{n}\right)^{(\rho-1)/5}\label{eq:Brownian_2},\\
		\upsilon_M\asymp  \left( \frac{\log(n)}{n}\right)^{4(\rho-1)/15}\label{eq:Brownian_3}.
	\end{align}
	Moreover, if $\sigma_m^2 \le C m^{-(8+\delta)}$ for some constant $C>0$ and $\delta>0$, then \ref{a:betaSeries} is satisfied, $\Theta_M=O\left( M^{-(1+\delta/2)}\right)$ and the rate $\alpha_n$ in Theorem \ref{thm:predictiveConsistency} satisfies the following conditions: 
	If $\rho\le (5+\delta)/(15+\delta)$, then $\alpha_n=O((\log(n)/n)^{(13\rho-3)/30})$ while if $\rho> (5+\delta)/(15+\delta)$ it holds that $\alpha_n=O((\log(n)/n)^{(1-\rho)(1+\delta/2)/15})$. The optimal rate is achieved when $\rho= (5+\delta)/(15+\delta)$ and leads to $\alpha_n=O((\log(n)/n)^q)$, where $q=((2+\delta)/(15+\delta))/3$.
\end{lem}
\begin{proof}[Proof of Lemma \ref{lem:BrownianExample}]
	For any $m\ge 1$
	\begin{align*}
		\lambda_m-\lambda_{m+1}&=\frac{32}{\pi^2} \frac{m}{(2m-1)^2(2m+1)^2},
	\end{align*}
	which is decreasing as $1\le m\to\infty$ and thus the eigengaps are given by
	\begin{align*}
		\delta_m&=\frac{32}{\pi^2} \frac{m}{(2m-1)^2(2m+1)^2}, \quad m\ge 1.
	\end{align*} 
	Since the harmonic sum $H(M)=\sumM 1/m$ satisfies $H(M)\le 1+\log(M)$ and $M=M(n)\to\infty$ as $\ntoinf$, we obtain
	\begin{align*}
		\sumM \frac{1}{\sqrt{\lambda_m}\delta_m}&=\frac{\pi^3}{64} \sumM \frac{(2m-1)^3(2m+1)^2}{m}\asymp M(n)^5.
	\end{align*}
	If $M=M(n)$ satisfies \eqref{eq:Brownian_1}, then $\sumM \lambda_m^{-1/2}\delta_m\inv \asymp c_n^{\rho-1}$ and thus condition \ref{a:Mrate} is satisfied. A simple calculation leads to 
	\begin{align*}
		\tau_M&=\sumM \frac{1}{\lambda_m}=\sumM \frac{\pi^2 (2m-1)^2}{4} \asymp M(n)^3,
	\end{align*}
	and
	\begin{align*}
		\upsilon_M&= \sumM \frac{1}{\delta_m}=\frac{\pi^2}{32} \sumM  \frac{(2m-1)^2(2m+1)^2}{m}\asymp M(n)^4.
	\end{align*}
	The results in \eqref{eq:Brownian_2} and \eqref{eq:Brownian_3} then follow. If $\sigma_m^2 \le C m^{-(8+\delta)}$ for some $C,\delta>0$, then $\sum_{m=1}^\infty \sigma_m^2/\lambda_m^2 \le O(1)  \sum_{m=1}^\infty m^{-(4+\delta)}<\infty$ and condition \ref{a:betaSeries} is satisfied. From the orthonormality of the $\phi_m$
	\begin{align*}
		\Theta_M=\Big \lVert \sum_{m\ge M+1} \frac{\sigma_m}{\lambda_m}\phi_m \Big\rVert_{L^2}\le \sum_{m\ge M+1} \frac{\lvert \sigma_m\rvert}{\lambda_m}
		&\le O(1) \sum_{m\ge M+1}m^{-(2+\delta/2)}\\
		&\le O(1) \int_{M}^\infty s^{-(2+\delta/2)}ds\\
		&=O\left(\frac{1}{M^{1+\delta/2}}\right),
	\end{align*}
	which implies $\Theta_M=(\log(n)/n)^{(1-\rho)(1+\delta/2)/15}$. Also note that $c_n\upsilon_M\asymp (\log(n)/n)^{(1+4\rho)/15}$ and $c_n^\rho \tau_M^{1/2}\asymp (\log(n)/n)^{(13\rho-3)/30}$. This implies
	\begin{align*}
		\alpha_n=c_n\upsilon_M+c_n^\rho \tau_M^{1/2}+\Theta_M \le O((\log(n)/n)^{(13\rho-3)/30}+(\log(n)/n)^{(1-\rho)(1+\delta/2)/15}).
	\end{align*}
	Thus, if $\rho\le (5+\delta)/(15+\delta)$, then $\alpha_n=O((\log(n)/n)^{(13\rho-3)/30})$. Similarly, if $\rho> (5+\delta)/(15+\delta)$, then $\alpha_n=O((\log(n)/n)^{(1-\rho)(1+\delta/2)/15})$. The optimal rate is achieved when $\rho= (5+\delta)/(15+\delta)\in(1/3,1)$ and leads to $\alpha_n=O((\log(n)/n)^q)$, where $q=((2+\delta)/(15+\delta))/3$.
\end{proof}

\references

\end{document}